\documentclass[journal]{IEEEtran}

\usepackage{booktabs}
\usepackage{graphicx}

\usepackage{amsthm,amsmath,amssymb}
\usepackage{xcolor}
\usepackage[utf8]{inputenc}
\usepackage{multirow}
\usepackage[caption=false,font=normalsize,labelfont=sf,textfont=sf]{subfig}
\usepackage{url}
\usepackage{soul}
\usepackage{enumitem}
\usepackage{mathrsfs}
\usepackage{balance}
\usepackage{cite}
\usepackage{tabularx}
\usepackage[colorlinks]{hyperref}
\usepackage[linesnumbered,ruled,vlined,nofillcomment]{algorithm2e}

\definecolor{MyBlue}{RGB}{0,0,255}

\hypersetup{citecolor=MyBlue, linkcolor=MyBlue, urlcolor=MyBlue}

\graphicspath{{figs/}}

\newtheorem{definition}{Definition}[section]
\newtheorem{theorem}{Theorem}[section]
\newtheorem{lemma}{Lemma}[section]
\newtheorem{proposition}{Proposition}[section]
\newtheorem{example}{Example}[section]
\newtheorem{property}{Property}[section]

\begin{document}

\title{GNN-based Anchor Embedding for Efficient Exact Subgraph Matching}

\author{Bin Yang, Zhaonian Zou, Jianxiong Ye
	\IEEEcompsocitemizethanks{
	\IEEEcompsocthanksitem Bin~Yang, Zhaonian~Zou and Jianxiong~Ye are with School of Computer Science and Technology, Harbin Institute of Technology, Heilongjiang, China. E-mail: bin.yang@stu.hit.edu.cn
    }
}

\markboth{Journal of \LaTeX\ Class Files,~Vol.~14, No.~8, August~2025}%
{Shell \MakeLowercase{\textit{et al.}}: A Sample Article Using IEEEtran.cls for IEEE Journals}


\maketitle

\begin{abstract}
Subgraph matching query is a fundamental problem in graph data management and has a variety of real-world applications. Several recent works utilize deep learning (DL) techniques to process subgraph matching queries. Most of them find approximate subgraph matching results without accuracy guarantees. Unlike these DL-based inexact subgraph matching methods, we propose a learning-based exact subgraph matching framework, called \textit{graph neural network (GNN)-based anchor embedding framework} (GNN-AE). In contrast to traditional exact subgraph matching methods that rely on creating auxiliary summary structures online for each specific query, our method indexes small feature subgraphs in the data graph offline and uses GNNs to perform graph isomorphism tests for these indexed feature subgraphs to efficiently obtain high-quality candidates. To make a tradeoff between query efficiency and index storage cost, we use two types of feature subgraphs, namely anchored subgraphs and anchored paths. Based on the proposed techniques, we transform the exact subgraph matching problem into a search problem in the embedding space. Furthermore, to efficiently retrieve all matches, we develop a parallel matching growth algorithm and design a cost-based DFS query planning method to further improve the matching growth algorithm. Extensive experiments on 6 real-world and 3 synthetic datasets indicate that GNN-AE is more efficient than the baselines, especially outperforming the exploration-based baseline methods by up to 1--2 orders of magnitude.
\end{abstract}


\section{Introduction}
\label{sec:introduction}

\IEEEPARstart{S}{ubgraph} matching~\cite{sun2020memory, zhang2024comprehensive} is a fundamental problem in graph data management. Given a data graph $G$ and a query graph $Q$, a subgraph matching asks for all subgraphs (also called matches) in $G$ that match $Q$ in terms of isomorphism. Consider the example in Fig.~\ref{fig:example-sm}, the subgraphs $g_1$ and $g_2$ are all matches of $Q$ in $G$. The subgraph isomorphism problem is known to be NP-complete~\cite{cook2023complexity}, so it is costly to find all exact matches.   
Subgraph matching has many practical applications, such as discovering structures in chemical networks~\cite{zhang2024rapid}, finding communities in social networks~\cite{ajoykumar2023study}, explaining neural networks in machine learning~\cite{wu2023rethinking}, etc.

\begin{figure}[!t]
	\setlength{\abovecaptionskip}{0cm}
	\centering
	\includegraphics[width=0.9\linewidth]{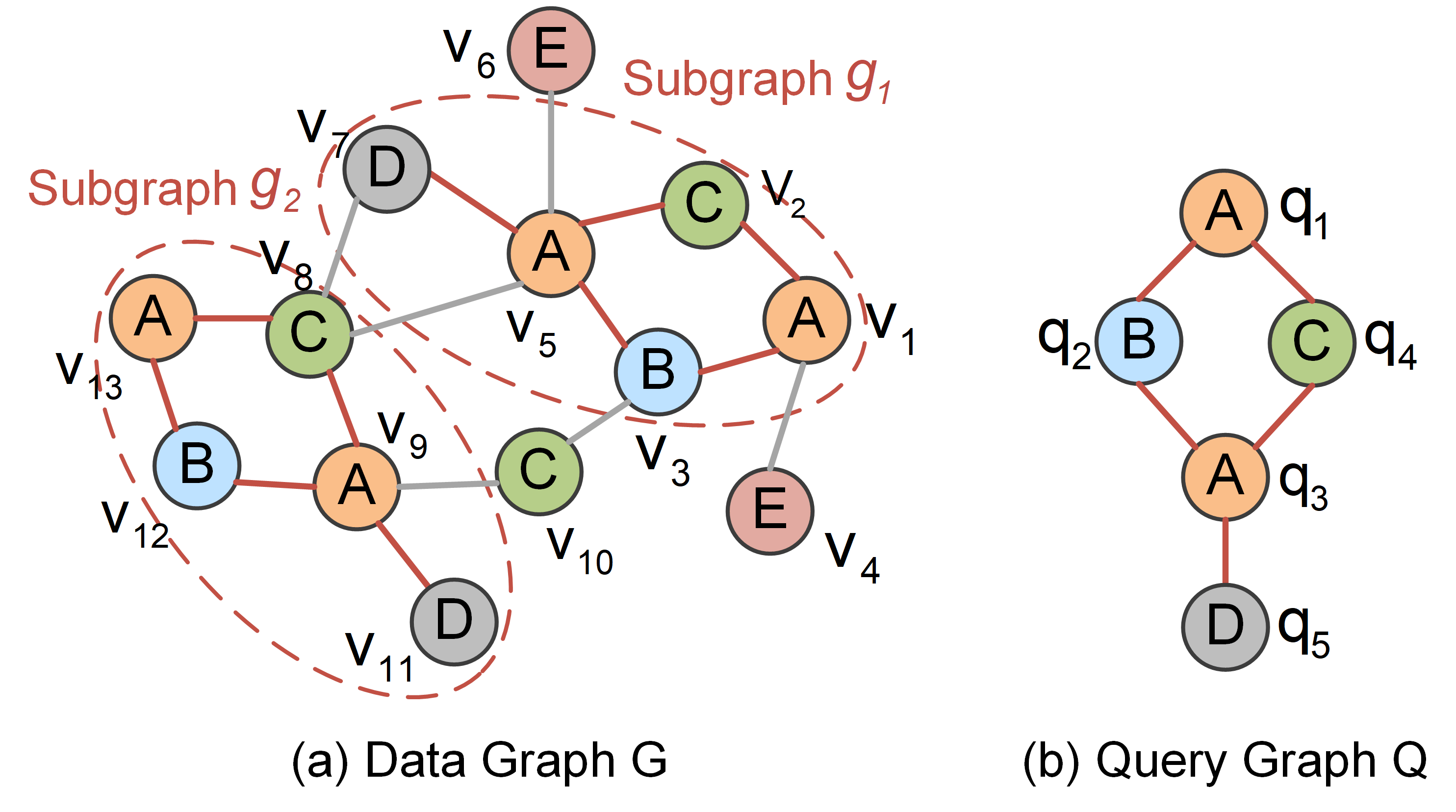}
	\caption{An example of the subgraph matching.}
	\label{fig:example-sm}
\end{figure} 

Prior works on the exact subgraph matching problem usually followed the exploration-based or the join-based methodology~\cite{sun2020memory, sun2021rapidmatch}. The exploration-based methods adopt the backtracking search, which maps query vertices to data vertices and iteratively extends intermediate results to find matches. The join-based methods convert the subgraph query to a relational query, and evolve the multi-way join to obtain all results. 

In recent years, several advanced works~\cite{liu2023d2match, roy2022interpretable, lou2020neural} utilize deep learning-based techniques such as \textit{graph neural networks} (GNNs) to handle the subgraph matching problem. However, these works can only obtain approximate matching results. Specifically, they transform the query graph and the data graph into vectors in an embedding space and learn the subgraph relationship between them. Although these GNN-based approaches are efficient, there are several drawbacks:   

(1) They usually approximately find subgraph isomorphism relationships between the query and the data graph, and are not suitable for retrieving all matches. 

(2) There are no theoretical guarantees or error bounds about the approximate accuracy.

It is known that GNNs usually provide approximate solutions, such as for classification or regression, without an accuracy guarantee. Thus, applying GNNs to handle approximate subgraph matching is relatively straightforward. However, it is not trivial to utilize GNNs to handle exact subgraph matching~\cite{ye2024efficient}. To obtain exact results for subgraph matching via GNNs, RL-QVO~\cite{wang2022reinforcement} and GNN-PE~\cite{ye2024efficient} provide two demonstrations. RL-QVO employs Reinforcement Learning (RL) and GNNs to generate a high-quality matching order for subgraph matching. However, it relies on historical queries to train the model, which limits its scalability and robustness. GNN-PE trains a GNN model to learn the dominance relationships between the vertices and their 1-hop neighborhood substructures on the data graph. These vertices are concatenated into paths as the basic matching units. GNN-PE requires that the GNNs be trained to overfit the training data set. When the data graph is large, it results in a high offline computation cost.

\noindent{\bf Our Method.}
Motivated by the above, to further solve the subgraph matching problem, in this paper, we propose a novel GNN-AE (\underline{GNN}-based \underline{A}nchor \underline{E}mbedding) framework, which allows exact subgraph matching. In contrast to existing GNN-based methods for approximate subgraph matching without theoretical guarantees of accuracy, GNN-AE can obtain all exact results for subgraph matching and retrieve all matches. Unlike RL-QVO~\cite{wang2022reinforcement}, our GNN-AE does not rely on historical query graphs to train the GNNs (i.e., only depends on the data graph). Furthermore, our GNN model does not have to overfit the training dataset, which makes our GNN-AE have a low offline computation cost compared with the GNN-PE~\cite{ye2024efficient}.

In GNN-AE, we select some edges (also called \textit{anchors}) as the basic matching units from the query graph $Q$ to the data graph $G$. We pre-store the $k$-hop graph/path structures (called \textit{anchored subgraphs/paths} for short) surrounding each edge in $G$ into the indexes to facilitate online subgraph matching queries. Although most existing exact subgraph matching methods also consider the 1-hop neighborhood of vertices, they build an auxiliary index online for each query, which occupies much time in online query processing. In our GNN-AE, index construction is \textit{offline} and \textit{one-time only}, which can help improve online query efficiency. In subgraph matching, if the $k$-hop graph around an edge in the query graph $Q$ is isomorphic to the $k$-hop graph/one of its substructures around an edge in the data graph $G$, then the two edges may match. However, graph isomorphism testing is computationally expensive. \textbf{The GNNs naturally have the order-invariant property~\cite{maron2019invariant}, which can map isomorphic graphs to the same embedding.} Thus, we express the above matching relationship via GNNs. That is, if there is a matching relationship from a query edge $(q_i, q_j)$ to a data graph edge $(u, v)$, the $k$-hop graph of $(q_i, q_j)$ and the $k$-hop graph/one of its substructures of $(u, v)$ have the same embedding via a GNN model. We store all edges in the data graph $G$ that have the same $k$-hop graph/substructure embedding vector in a collection of the index without omission. In online querying, by retrieving indexes, GNN-AE can obtain candidate matches with 100\% recall for each query edge without false dismissals to ensure exact subgraph matching results.

In summary, we make the following contributions:

{\textbullet} We propose a novel GNN-AE framework for exact subgraph matching via GNN-based anchor embeddings. (\textsection\ref{sec:overview})  

{\textbullet} We design GNN-based anchored subgraph embedding and path embedding techniques for exact subgraph retrieval. (\textsection\ref{sec:gnn-based-feature-embedding})

{\textbullet} We propose an anchor matching mechanism based on subgraphs/paths and develop an efficient parallel matching growth algorithm to obtain all exact matching locations. (\textsection\ref{sec:subgraph-matching-growth})

{\textbullet} We design a cost-model-based DFS query plan to further enhance the parallel matching growth algorithm. (\textsection\ref{sec:cost-model-dfs-plan})

{\textbullet} We conduct experiments on 6 real and 3 synthetic datasets to verify the effectiveness and efficiency of GNN-AE. (\textsection\ref{sec:experimental})

Related works on subgraph matching are reviewed in Section~\ref{sec:related-work} and Section~\ref{sec:conclusion} concludes this paper.

\section{Problem Definition}
\label{sec:problem} 

We focus on undirected vertex-labeled graphs. Let $\Sigma$ be a set of labels. A graph is defined as $(V, E, L, \Sigma)$, where $V$ is a set of vertices, $E$ is a set of edges, and $L: V \to \Sigma$ is a function that associates vertex $v \in V$ with a label in $\Sigma$. 

\begin{table}[!t]
	\centering
	\fontsize{6.8pt}{8pt}\selectfont 
	\renewcommand\arraystretch{1.2}
	\caption{Notations and Descriptions}
	\label{tab:notations}
	\begin{tabular}{|p{1.4cm}||p{6cm}|} 
		\hline
		Notations & Descriptions \\ 
		\hline
		\hline
		$Q$, $G$ and $g$ & query graph, data graph and matched subgraph \\ 
		\hline
		$S^k_G(u, v)$ & an anchored subgraph of the data graph edge $(u, v)$ \\ 
		\hline
		$o(S^k_G(u, v))$ & an embedding of anchored subgraph $S^k_G(u, v)$ via GNNs \\
		\hline
		$P^k_G(u, v)$ & an anchored path of the data graph edge $(u, v)$ \\ 
		\hline
		$c(P^k_G(u, v))$ & an encoded embedding of anchored path $P^k_G(u, v)$ \\
		\hline
		$I_S, I_{S'}, I_P$ & hash indexes of anchored subgraphs and anchored paths \\
		\hline
		$C_{q_iq_j}$ & candidate matches of the query edge $(q_i, q_j)$ \\
		\hline
	\end{tabular}
\end{table} 

\begin{definition}[Graph Isomorphism]
	\label{def:graph-isomorphism}
	A graph $G = (V, E, L, \Sigma)$ is isomorphic to a graph $G' = (V', E', L', {\Sigma}’)$, denoted as $G \cong G'$, if there exists a bijecton $f: V \to V'$ such that (1) $\forall v \in V$, $L(v) = L'(f(v))$; (2) $\forall u, v \in V$, $(u, v) \in E$ if and only if $(f(u), f(v)) \in E'$. 
\end{definition}

A graph $G'$ is \textit{subgraph isomorphic} to a graph $G$, denoted as $G' \sqsubseteq G$, if $G'$ is isomorphic to a subgraph $g$ of $G$ (i.e., $g \subseteq G$). The subgraph isomorphism problem is NP-complete~\cite{cook2023complexity}.

\begin{definition}[Subgraph Matching Query]
	\label{def:subgraph-matching}
	Given a data graph $G$ and a query graph $Q$, a subgraph matching query finds all subgraphs $g$ in $G$ that are isomorphic to $Q$.
\end{definition}

In this paper, we consider exact subgraph matching, i.e., retrieving all subgraph locations on the data graph $G$ that are strictly isomorphic to the query graph $Q$. 

\section{An Overview}
\label{sec:overview}

Our GNN-AE is an \textit{exact subgraph matching algorithm} that can find all matches of the query graph. In contrast to many traditional exact subgraph matching methods, our algorithm is an index-based algorithm that relies on an \textit{index created offline and one-time only} instead of creating a summary online for each query. Moreover, GNN-AE adopts an \textit{edge-oriented matching} approach rather than a vertex-by-vertex expansion matching. Furthermore, GNN-AE is a \textit{learning-based algorithm} that utilizes machine learning models to perform graph isomorphism tests for small feature subgraphs. The framework of GNN-AE is illustrated in Algorithm~\ref{alg:gnn-ae-framework}, and Fig.~\ref{fig:example-gnnae} shows an example of GNN-AE for exact subgraph matching. 

The GNN-AE framework consists of two phases: offline pre-computation and online subgraph matching. We pre-process the data graph into embeddings via GNNs and path encoding (lines 1-10 and 12), build indexes (lines 11 and 13) in the offline pre-computation phase, and then answer online subgraph matching queries over indexes (lines 14-21).

Existing subgraph matching methods usually build matching relationships from the query to the data graph based on vertices or edges. The edge directly constrains the two endpoints to be neighbors, which can reduce the computations in matching. Thus, in GNN-AE, we set edges as the basic matching unit. We select some edges of the query graph to build matching relationships with edges in the data graph. These edges in the query or the data graph are named \underline{\textbf{\textit{anchors}}}, which anchor the matching from the query graph to the data graph.

\begin{algorithm}[t]
	\footnotesize
	\caption{The GNN-Based Anchor Embedding (GNN-AE) Framework for Exact Subgraph Matching}
	\label{alg:gnn-ae-framework}
	\KwIn{a data graph $G$, a query $Q$, vertex degree threshold $d^*$, and the hop value $k$ of the anchored subgraph/path}
	\KwOut{all subgraphs $g$ ($\subseteq G$) that are isomorphic to $Q$}  
	\tcp{\textbf{Offline Pre-Computation Phase}}
	initial a training data set $D = \emptyset$ and a path set $P = \emptyset$ \\
	\For{each edge $(u, v) \in G$}{
		\eIf{degree $d_u \leq d^*$ \textbf{or} $d_v \leq d^*$}{
			obtain anchored $k$-radius subgraphs $S^k_G(u, v)$ of $(u, v)$  \\ 
			add all possible anchored subgraphs $S^k_G(u, v)$ to $D$
		}{
			obtain anchored $k$-radius paths $P^k_G(u, v)$ of $(u, v)$ \\
			add all possible anchored paths $P^k_G(u, v)$ to $P$
		}
	}
	\tcc{train a GNN for subgraph embeddings} 
	train a GNN model $\mathcal{M}$ over the training data set $D$ \\
	obtain embeddings $o(S^k_G(u, v))$ for $S^k_G(u, v)$ in $D$ via $\mathcal{M}$ \\
	\tcc{build indexes over graph $G$} 
	build two hash indexes $I_S$ and $I_{S'}$ over embedding vectors for anchored subgraphs in $G$ \\
	encode each anchored path $P^k_G(u, v)$ in $P$ \\
	build a hash index $I_P$ over path encodings $c(P^k_G(u, v))$ for anchored paths in $G$   \\  
	\tcp{\textbf{Online Subgraph Matching Phase}}
	compute cost-model-based DFS query plan $\varphi$ for $Q$ \\ 
	\tcc{retrieve candidate matches}
	\For{each DFS edge $(q_i, q_j) \in Q$}{
		extract its maximum anchored $k$-radius subgraphs $S^k_Q(q_i, q_j)$ and maximum anchored $k$-radius paths $P^k_Q(q_i, q_j)$ \\
		get embeddings $o(S^k_Q(q_i, q_j))$ via GNN $\mathcal{M}$ \\
		get path encodings $c(P^k_Q(q_i, q_j))$ \\
		find candidate matches $C_{q_iq_j}$ by retrieving $I_S$, $I_{S'}$ and $I_P$ \\
	}
	\tcc{refine and obtain all exact matches}
	assemble candidate matches $C$ along the DFS edge order and refine subgraphs $g$ via hash join \\
	\Return{subgraphs $g$ ($\cong Q$)} 
\end{algorithm} 

\noindent{\bf The Offline Pre-computation Phase:}
In this phase, the data graph $G$, a vertex degree threshold $d^*$, and a hop value $k$ are given as input. Since edges are simple in structure, building matching relationships relying only on them results in a large candidate match space for the query. Existing exact subgraph matching methods usually consider the 1-hop neighbors of vertices to obtain fewer candidate matches. If the 1-hop neighborhood of the query vertex $q$ is not a subgraph of the 1-hop neighborhood of the data graph vertex $u$, then $q$ must not match $u$. We consider that the feature subgraphs within the $k$-hop neighborhood of endpoints on edges in $G$ is stored offline, and then perform online subgraph matching queries on these feature subgraphs to improve query efficiency. Unlike existing subgraph matching methods that consider vertex neighborhoods online, our index construction is \textit{offline} and \textit{one-time only}. However, obtaining all subgraphs for the high-degree vertices in the data graph $G$ results in a high offline index storage cost. Thus, we extract paths for high-degree vertices and graph structures for low-degree vertices.

Specifically, in GNN-AE, we design the edge-based $k$-radius graph structure and path structure. For each edge $(u, v)$ with $L(u) \leq L(v)$ in $G$, if it has the vertex degree $d_u \leq d^*$ (i.e., the low-degree vertex), we extract all possible subgraphs (which contain the edge $(u, v)$) within the $k$-hop graph structure centered on $u$ as features. If the vertex degree $d_u > d^* \land d_v \leq d^*$, we obtain all possible subgraphs within the $k$-hop graph structure centered on $v$. We call these subgraphs the \textit{anchored $k$-radius subgraphs} (lines 2-5). If the edge $(u, v)$ has only high-degree vertices (i.e., $d_u > d^* \land d_v > d^*$), we extract all possible paths (which contain $(u, v)$) within the $k$-hop path structures extending from the vertex $u$ as features. We call these paths the \textit{anchored $k$-radius paths} (lines 6-8).

In order to build the matching relationship from the query graph $Q$ to the data graph $G$, we handle anchored subgraphs and anchored paths in different ways:

\begin{figure}[!t]
	\setlength{\abovecaptionskip}{0.1cm}
	\centering
	\includegraphics[width=0.96\linewidth]{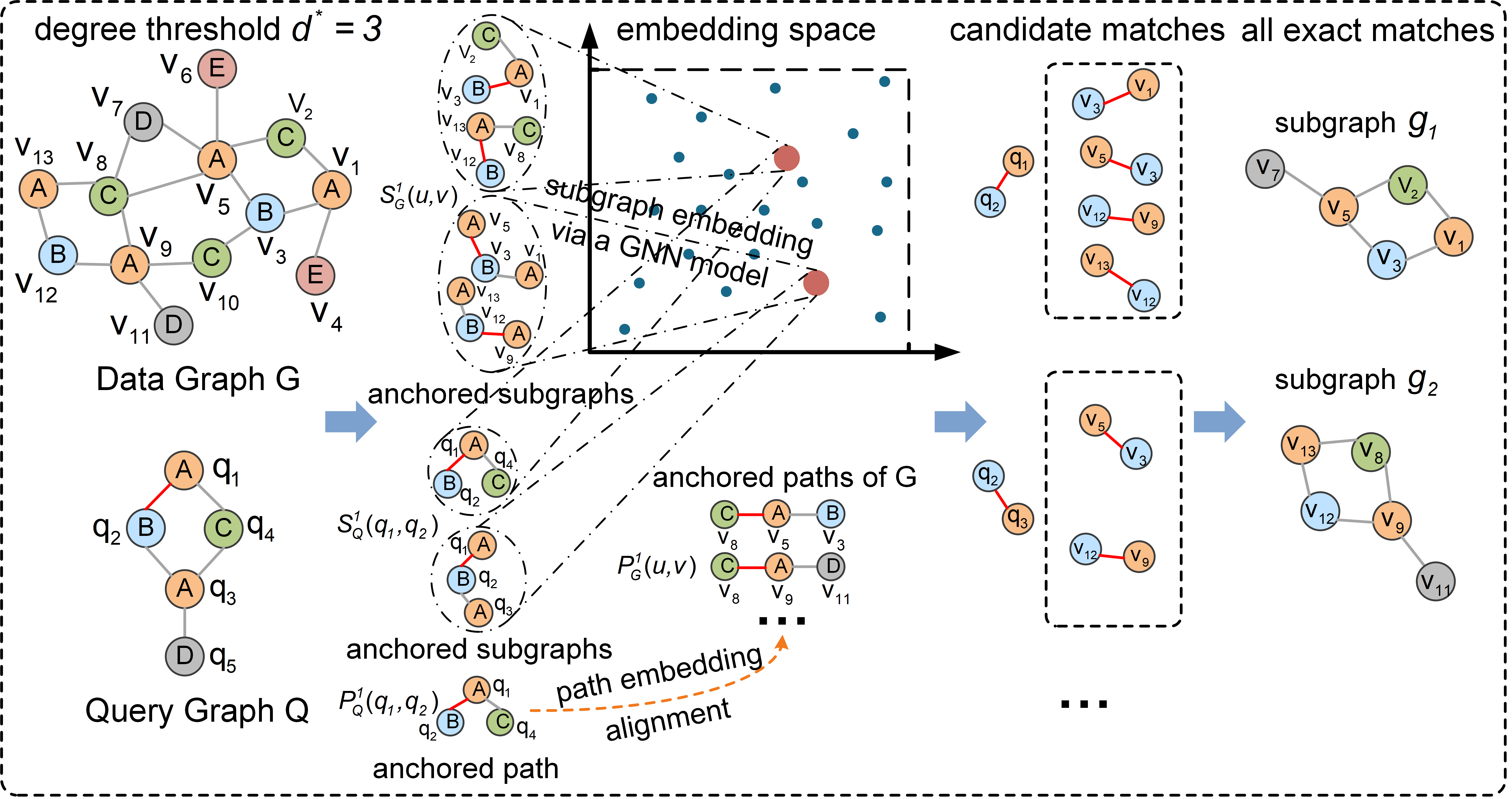}
	\caption{An example of GNN-AE for efficient exact subgraph matching.}
	\label{fig:example-gnnae}
	\vspace{-0.2em}
\end{figure} 

(1) For anchored subgraphs of $G$, we train a GNN model over them, and then obtain their embeddings (lines 9-10). For an edge $(q_i, q_j)$ with $L(q_i) \leq L(q_j)$ in the query $Q$, we extract its $k$-hop graph structures (which contain $(q_i, q_j)$) centered on $q_i$ and $q_j$ respectively. In online subgraph matching, if $(q_i, q_j)$ matches a data graph edge $(u, v)$, the $k$-hop graph structure of $(q_i, q_j)$ is isomorphic to a possible anchored subgraph of $(u, v)$. However, graph isomorphism testing is computationally expensive. \textbf{The GNNs naturally have the order-invariant property~\cite{maron2019invariant}, which can map isomorphic graphs to the same embedding.} Thus, we express the above matching relationship via GNNs. That is, we also transform the $k$-hop graph structures of $Q$ into embeddings via the same GNN model used in $G$. If the query edge $(q_i, q_j)$ matches the data graph edge $(u, v)$, the $k$-hop graph of $(q_i, q_j)$ and an anchored subgraph of $(u, v)$ have the same embedding via the GNN model. For example, in Fig.~\ref{fig:example-gnnae}, let $k = 1$, and the anchored subgraphs in $S^1_G(u, v)$ from $G$ and the graphs in $S^1_Q(q_1, q_2)$ from $Q$ have the same embedding. 

(2) For anchored paths of $G$, we directly encode them as embeddings due to their simplicity (line 12). For an edge $(q_i, q_j)$ with $L(q_i) \leq L(q_j)$ in the query $Q$, we extract all its $k$-hop paths extending from $q_i$. For example, in Fig.~\ref{fig:example-gnnae}, the path $P^1_Q(q_1, q_2)$ is the 1-hop path of $(q_1, q_2)$. If a query edge $(q_i, q_j)$ matches a data graph edge $(u, v)$, a path of $(q_i, q_j)$ and an anchored path of $(u, v)$ have the same encoded embedding.  

We build hash indexes $I_S$ and $I_{S'}$ to store key-value pairs $\{o(S^k_G(u, v)): \Theta\}$ consisting of the graph embedding $o(S^k_G(u, v))$ and the set $\Theta$ of all edges on $G$ corresponding to $o(S^k_G(u, v))$ (line 11). Similarly, for path embedding $c(P^k_G(u, v))$, we build a hash index $I_P$ to store $c(S^k_G(u, v))$ and all its corresponding data graph edges (line 13). They guarantee a 100\% matching recall for each query edge without omission. Thus, GNN-AE can obtain exact matching results.

\noindent{\bf The Online Subgraph Matching Phase:}
In this phase, a query graph $Q$ is given as input. In order to build the matching from query vertices to data graph vertices, we perform a depth-first search (DFS) on $Q$ to cover all query vertices. Since multiple DFS can be chosen in $Q$, we design a cost model to find a low-cost DFS strategy (line 14). For each edge $(q_i, q_j) \in Q$ in the DFS traversal, we extract its $k$-hop graphs (including $(q_i, q_j)$) centered on $q_i$ and $q_j$ respectively, and all $k$-hop paths. We get embeddings of $k$-hop graphs via the trained GNN model, and encode all $k$-hop paths as path embeddings. We take embeddings as keys, and retrieve hash indexes $I_S$, $I_{S'}$, and $I_P$ respectively to obtain candidate matches for $(q_i, q_j)$ (lines 15-19). Finally, we organize the candidate matches for each edge in $Q$ into a table and assemble them into all exact matches via hash joins (line 20).

\section{GNN-based Feature Embedding}
\label{sec:gnn-based-feature-embedding}

In this section, we introduce the feature subgraphs (anchored subgraph and anchored path) adopted in our GNN-AE framework and present the method for deriving the embeddings of the feature subgraphs using graph neural networks (GNN).

\subsection{Feature Subgraphs}
\label{sub:feature-subgraphs}

An edge (also called an anchor) in the query $Q$ usually has a large number of matching edges in the data graph $G$. It is generally inefficient to obtain the results by joining the matched edges of the edges of $Q$. This is because using single edges as the units of joining often generates many invalid results, which wastes a lot of computations. Thus, it is essential to adopt suitable subgraphs as the units of joining.

In our GNN-AE framework, we employ two kinds of subgraphs, namely \emph{anchored subgraphs} and \emph{anchored paths}, as the units of joining.

\begin{figure}[!t]
	\setlength{\abovecaptionskip}{0cm}
	\centering
	\includegraphics[width=0.9\linewidth]{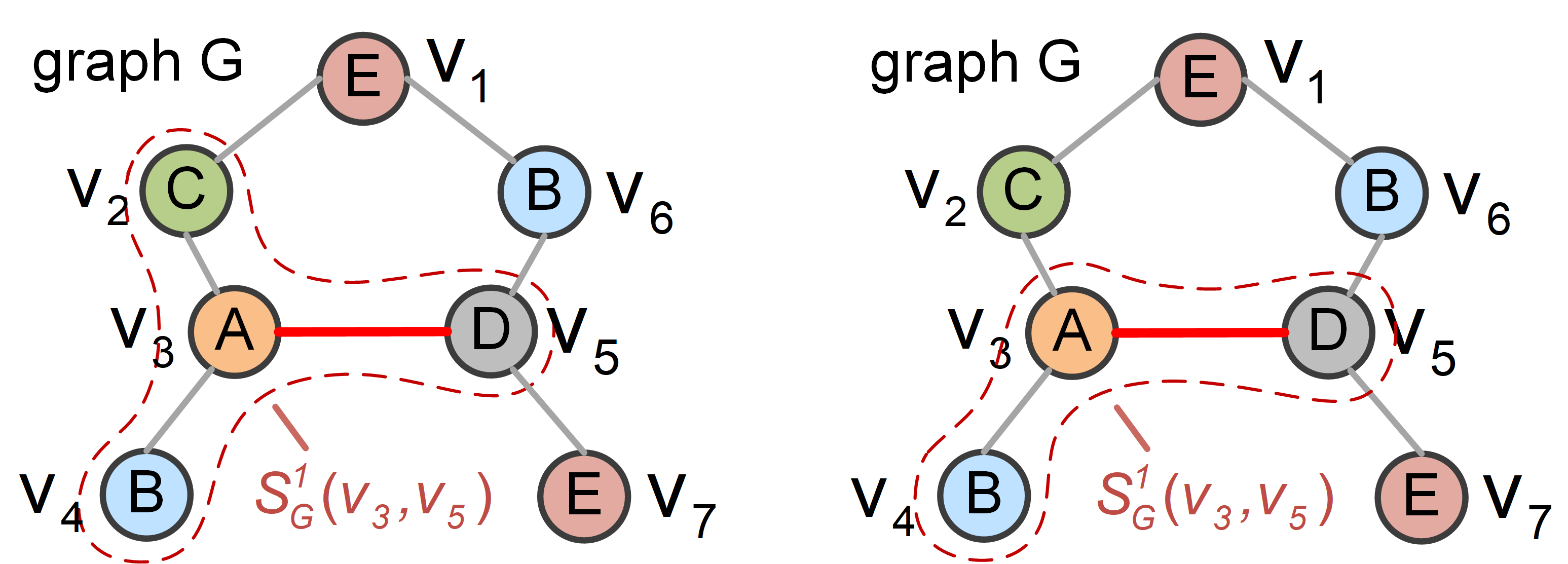}
	\caption{An example of the possible anchored subgraphs.}
	\label{fig:anchor-subgraphs}
	\vspace{0em}
\end{figure} 

\begin{definition}[Anchored Subgraph]
	\label{def:anchored-subgraph}
	Given an edge $(u, v)$ in a vertex-labeled graph $G = (V, E, L, \Sigma)$ where $L(u) \leq L(v)$ and an integer $k \geq 1$, the subgraph of $G$ spanned by the edge $(u, v)$ and some vertices $w$ with distance $d(u, w) \leq k$ is called the anchored $k$-radius subgraph with respect to $(u, v)$, denoted as $S^k_G(u, v)$.  
\end{definition}

In particular, the subgraph of $G$ spanned by the edge $(u, v)$ and all vertices $w$ with $d(u, w) \leq k$ is called the \textbf{\textit{maximum anchored $k$-radius subgraph}}. Obviously, if $L(u) < L(v)$, there is a unique maximum anchored $k$-radius subgraph $S^k_G(u, v)$; if $L(u) = L(v)$, there are two maximum anchored $k$-radius subgraph $S^k_G(u, v)$. Fig.~\ref{fig:anchor-subgraphs} illustrates two possible anchored 1-radius subgraphs $S^1_G(v_3, v_5)$ for edge $(v_3, v_5)$, where the left side is the maximum anchored 1-radius subgraph. Intuitively, an anchored subgraph is a graph containing the edge $(u, v)$ and spanned by extending vertices centered on $u$.

If the degrees of $u$ and $v$ are very large in $G$, $S^k_G(u, v)$ is also very large, making itself ineffective to be a feature. Thus, in this situation, another kind of substructure is adopted.

\begin{figure}[!t]
	\setlength{\abovecaptionskip}{0cm}
	\centering
	\includegraphics[width=0.98\linewidth]{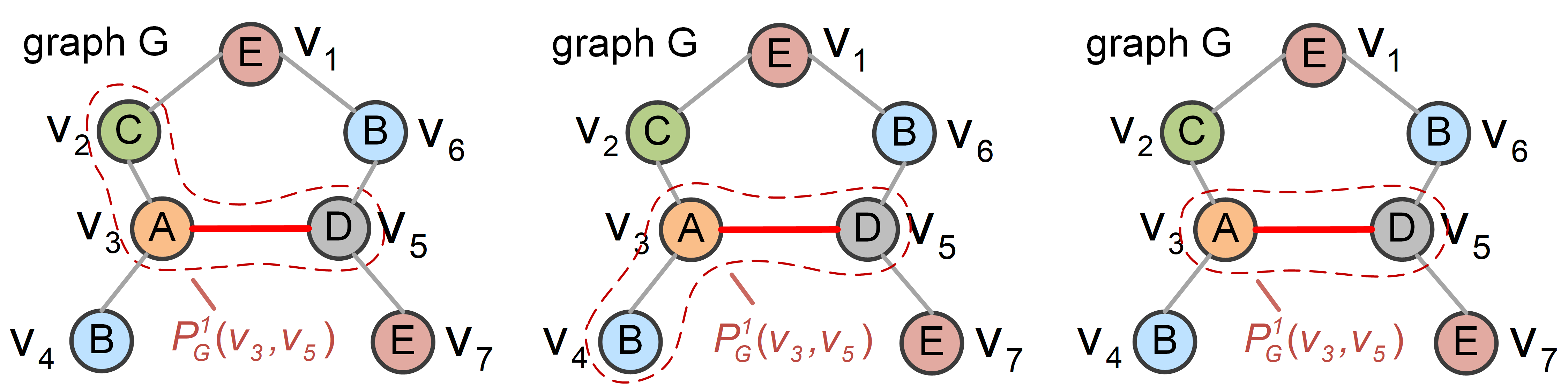}
	\caption{An example of the possible anchored paths.}
	\label{fig:anchor-paths}
	\vspace{0em}
\end{figure} 

\begin{definition}[Anchored Path]
	\label{def:anchored-path}
	Given an edge $(u, v)$ in a vertex-labeled graph $G = (V, E, L, \Sigma)$ where $L(u) \leq L(v)$ and an integer $k \geq 1$, a path in $G$ spanned by $u$, $v$, and some vertices $w$ with distance $d(u, w) \leq k$ is called an anchored $k$-radius path with respect to $(u, v)$, denoted as $P^k_G(u, v)$. 
\end{definition}

In particular, the path of $G$ spanned by $u$, $v$, and all vertices $w$ with $d(u, w) \leq k$ is called a \textbf{\textit{maximum anchored $k$-radius path}}. Unlike anchored $k$-radius subgraph, there are usually multiple maximum anchored $k$-radius paths $P^k_G(u, v)$ with respect to $(u, v)$. Fig.~\ref{fig:anchor-paths} illustrates three possible anchored 1-radius paths $P^1_G(v_3, v_5)$ with respect to edge $(v_3, v_5)$, of which the first two are the maximum anchored paths. Intuitively, an anchored path is a path containing the edge $(u, v)$ and spanned by extending vertices starting from $u$.

In our GNN-AE framework, we set a threshold $d^*$ on vertex degree. For an edge $(u, v)$ with $L(u) \leq L(v)$, if the vertex degree $d(u) \leq d^*$ or $d(v) \leq d^*$, the anchored $k$-radius subgraph $S^k_G(u, v)$ is adopted as features; otherwise, all anchored $k$-radius paths $P^k_G(u, v)$ are adopted as features.

Some existing works consider other features, such as trees \cite{lin2014graph} or stars~\cite{wang2019efficient}, to aid matching. However, subgraphs are more complex than trees or stars in structure, which can help filter out more invalid (i.e., does not belong to the final results) candidate matches for a query as much as possible in subgraph matching. Furthermore, it is usually acceptable to obtain all possible anchored subgraphs for an edge $(u, v)$ with $d(u) \leq d^* \lor d(v) \leq d^*$, especially when the size or the maximum vertex degree of the data graph is not large. For an edge $(u, v)$ with $d(u) > d^* \land d(v) > d^*$, although the tree or star has a more complex structure, all its possible substructures result in a high offline storage cost. Thus, we adopt paths as features.

In real-world graphs, vertex degrees usually follow the power-law distribution~\cite{barabasi1999}, and only a small fraction of vertices have high degrees~\cite{ye2024efficient}. For example, 85\% of vertices have degree $\leq 10$ in DBLP graph data~\cite{sun2020memory}. Therefore, we can set $d^* = 10$. Specifically, \textbf{when the maximum vertex degree $d_{max}(u)$ of the data graph is large (e.g., $d_{max}(u) > 8$), we adopt the anchored 1-radius subgraph to avoid the combinatorial explosion of the number of anchored substructures}. According to Lemma~\ref{lem:num-substructure}, the edge $(u, v)$ has $2^{d(u)-1}$ possible anchored 1-radius subgraphs. We can further directly deduce that the edge $(u, v)$ with $d(u) \leq d^* \lor d(v) \leq d^*$ has at most $2^{d^*-1}$ anchored 1-radius subgraphs. When $d^* = 10$, there are no more than 512 ($= 2^{9}$) anchored subgraphs for edge $(u, v)$, which is acceptable. \textbf{When the maximum vertex degree $d_{max}(u)$ of the data graph is not large (e.g., $d_{max}(u) \leq 8$), we can adopt the anchored 2-radius or 3-radius subgraph to obtain a better filtering effect} in subgraph matching due to their more complex structures. \textbf{For such data graphs, it is also acceptable to extract all possible anchored 2-radius or 3-radius subgraphs}. In fact, there are lots of data graphs/networks with a low maximum vertex degree in the real world, such as DNA molecular graph structures~\cite{wu2025overview}, urban road networks~\cite{chen2024global}, IoT sensor networks in a smart park~\cite{darabkh2025smart}, base station networks for cellular mobile communications~\cite{liu2024mean, lin2025performance}, and so on.
 
\noindent\textbullet~~\emph{DNA molecular graph structure:} A DNA molecule contains millions of nucleotides, each of which can serve as a vertex in the graph structure. Each nucleotide is linked to the nucleotide above and below it by a phosphodiester bond, and to the nucleotide in the complementary chain through a base. Therefore, the maximum degree of each nucleotide vertex is only 3.

\noindent\textbullet~~\emph{Urban road network:} A road can be considered as an edge in the graph, and a road intersection as a vertex. In the road network, there are many types of intersections, such as crossroads, three-way intersections, and highway ramps. Each road intersection type can be mapped to a unique vertex label. In the real world, the number of forks at a road intersection usually does not exceed 8. Thus, the road network is a typical graph structure with a low maximum vertex degree.

\noindent\textbullet~~\emph{IoT sensor network in a smart park:} Vertices are sensors in the Internet of Things (IoT), such as temperature sensors, security cameras, and smart parking detectors. The number of vertices in an IoT sensor network of a smart park can reach hundreds of thousands. Due to the limitations of wireless transmission distance and energy consumption, a sensor is usually only connected to 3-5 adjacent sensors for data backup or relay transmission. Therefore, the maximum degree of an IoT sensor network in a smart park is usually only 5.

\noindent\textbullet~~\emph{Base station network for cellular mobile communications:} Each base station is a vertex. To avoid signal interference and ensure coverage, a base station usually only establishes data connections with 5-6 neighboring base stations. Thus, the base station network is another a typical data graph with a low maximum vertex degree.

\begin{lemma}\label{lem:num-substructure}
	For an edge $(u, v)$ in graph $G$, if vertex $u$ has degree $d(u)$, there are $2^{d(u)-1}$ possible anchored 1-radius subgraphs.
\end{lemma}

\begin{proof}
	According to Definition~\ref{def:anchored-subgraph}, the anchored 1-radius subgraph with respect to $(u, v)$ is a subgraph of $G$ that must contain the edge $(u, v)$ and some vertices $w$ with distance $d(u, w) \leq 1$. Let the 1-hop subgraph starting from vertex $u$ be $S = (V_S, E_S, L_S, \Sigma_S)$. The number of possible anchored 1-radius subgraphs of $(u, v)$ is equivalent to the number of subsets of the non-$(u, v)$ edge sets $\Theta = \{e | e \in E_S \wedge e \neq (u, v)\}$ in $S$, where $|\Theta| = d(u)-1$. Thus, there are $C^0_{d(u)-1} + C^1_{d(u)-1} + \dots + C^{d(u)-1}_{d(u)-1} = 2^{d(u)-1}$ possible anchored 1-radius subgraphs. The lemma holds.  
\end{proof}

\subsection{GNN Model for Anchored Subgraph Embedding}
\label{sub:gnn-for-anchor-embedding}

If an anchored subgraph of the query is isomorphic to an anchored subgraph of the data graph, there may be a match. However, graph isomorphism testing is computationally expensive. Therefore, we adopt a model to replace the graph isomorphism testing. We determine the isomorphism relationship of anchored subgraphs by comparing the embeddings of the model. The model should satisfy two conditions: 
\begin{itemize}
	\item[\textit{C1.}] Isomorphic anchored subgraphs have the same embedding via the model to ensure no matches are missed.
	\item[\textit{C2.}] Non-isomorphic subgraphs have different embeddings via the model as much as possible to reduce invalid matches.
\end{itemize}
We employ a GNN model to replace the graph isomorphism testing. The GNN model takes an anchored subgraph $S^k_G(u, v)$ as input and an $m$-dimensional embedding $o(S^k_G(u, v))$ as output, which can satisfy the above conditions. 

\begin{figure}[!t]
	\setlength{\abovecaptionskip}{0cm}
	\centering
	\includegraphics[width=0.98\linewidth]{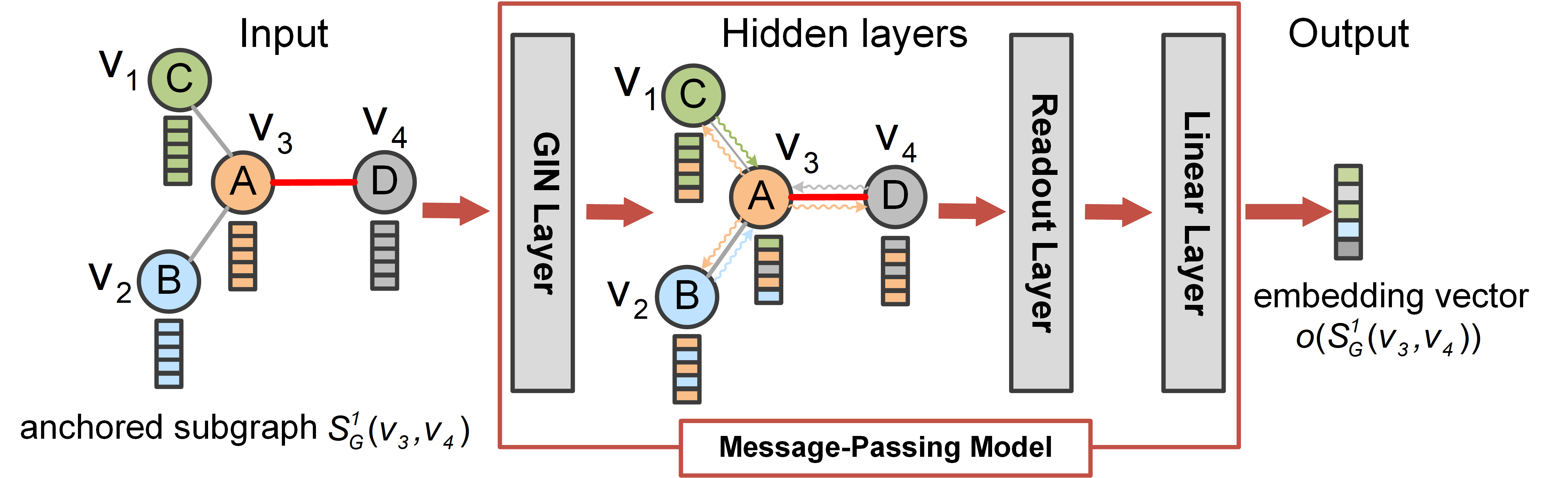}
	\caption{The GNN model for the anchored subgraph embedding.}
	\label{fig:gnn-model}
	\vspace{-0.2em}
\end{figure} 

\noindent{\bf GNN Model Architecture.}
Fig.~\ref{fig:gnn-model} illustrates the GNN model. Specifically, we adopt a classic \textit{Graph Isomorphism Network} (GIN)~\cite{xu2018powerful}) model for subgraph matching. Here are the reasons: (1) For the above condition \textit{C1}, the GNN model (e.g., GIN) naturally has the \textit{order-invariant} property~\cite{maron2019invariant} to enable isomorphic anchored subgraphs to have the same embedding. Theorem~\ref{the:isomorphic-equal} also provides a rigorous proof. Some existing graph encoding or embedding methods cannot satisfy the condition \textit{C1}, such as DeepWalk~\cite{perozzi2014deepwalk}, and GraphSAGE~\cite{Inductive2017}. Because they introduce the randomization step (e.g., sampling, random walks) into their execution. (2) For the above condition \textit{C2}, the GIN model is as strong as the Weisfeiler-Lehman (WL) graph isomorphism testing to make non-isomorphic anchored subgraphs have different embeddings as much as possible, and stronger than other embedding methods, such as the graph factorization embedding method~\cite{choudhury}. Theorem~\ref{the:wl-gin} provides a rigorous theoretical guarantee. (3) The GIN model has a simple architecture and high execution efficiency. Therefore, the classic GIN model is suitable for subgraph matching. If a more advanced GNN model satisfies conditions \textit{C1} and \textit{C2}, and has higher execution efficiency than GIN in the future, the GIN model can be replaced by it.

\begin{property}[Order-Invariant]
	\label{def:order-invariant}
	Assume $A \in \mathbb{R}^{n \times n}$ is the adjacency matrix of a graph $G$ and $P \in \mathbb{R}^{n \times n}$ is a permutation matrix, the function $f$ is order-invariant if $f(P^TAP) = f(A)$.
\end{property}

Property~\ref{def:order-invariant} indicates that no matter how the vertices in the graph $G$ are numbered, the output of the function $f$ is the same. That is, the output of function $f$ regardless of the assignment of vertex IDs on the graph $G$. The GNN model satisfies the order-invariant property, which means that the embedding generated by the GNN focuses on the topological structure of the graph $G$ rather than the vertex ID assignment. This is the theoretical basis for the GNN model to support graph isomorphism testing.

\underline{\textit{Input:}}
The input of the GNN model is an anchored subgraph $S^k_G(u, v)$. In Fig.~\ref{fig:gnn-model}, each vertex $v_i$ in the input anchored subgraph $S^k_G(v_3, v_4)$ is associated with an initial feature embedding, which is obtained via label encoding~\cite{bisong2019}. 

\underline{\textit{Hidden Layers:}}
In each GIN layer, vertex feature embeddings on neighbors are aggregated. In $t$-th round, the new representation $h^{(t+1)}_v$ for $v$ is computed:
\begin{align}\label{eqn:gin-layer}
	a^{(t)}_v &= \sum_{u \in N(v)} h^{(t)}_u,\\ 
	h^{(t+1)}_v &= MLP^{(t+1)}\left((1 + \epsilon^{(t+1)}) h^{(t)}_v, a^{(t)}_v\right). 
\end{align}
where $\epsilon^{(t+1)}$ is a learnable parameter or a fixed scalar. Finally, all hidden representations $h^{(t)}_v$ at each layer are aggregated to form the representation of the anchored subgraph $S^k_G(u, v)$,
\begin{equation}\label{eqn:gin-readout}
	h^{(t)}_G = readout(\{h^{(t)}_v | v \in S^k_G(u, v)\}).
\end{equation}
Generally, the \textit{readout} function can be an operator that sums or means all vertex hidden representations in $S^k_G(u, v)$. 

A linear layer performs a linear transformation for $h^{(t)}_G \in \mathbb{R}^{n \times 1}$ ($n$ is the $readout$ hidden dimension) and obtains the embedding vector $o(S^k_G(u, v))$ with dimension $m$ for $S^k_G(u, v)$. 
\begin{equation}\label{eqn:gin-fully}
	o(S^k_G(u, v)) = Wh^{(t)}_G,
\end{equation}
where $W$ is a learnable weight matrix and $W \in \mathbb{R}^{m \times n}$. 

\underline{\textit{Output:}}
We output a vector $o(S^k_G(u, v))$ with dimension $m$ as the embedding of the anchored subgraph $S^k_G(u, v)$. 

\noindent{\bf GNN Model Training.}
For each edge $(u, v)$ with $d(u) \leq d^* \lor d(v) \leq d^*$, we extract all its possible anchored subgraphs. As mentioned in subsection~\ref{sub:feature-subgraphs}, when the maximum vertex degree $d_{max}(u)$ of the data graph is large (e.g., $d_{max}(u) > 8$), we adopt the anchored 1-radius subgraph to avoid the combinatorial explosion of the number of anchored substructures. When the maximum vertex degree $d_{max}(u)$ of the data graph is not large (e.g., $d_{max}(u) \leq 8$), we can adopt the anchored 2-radius or 3-radius subgraph to obtain a better filtering effect. For such data graphs, it is also acceptable to extract all possible anchored 2/3-radius subgraphs.

\begin{figure}[!t]
	\setlength{\abovecaptionskip}{0cm}
	\centering
	\includegraphics[width=0.98\linewidth]{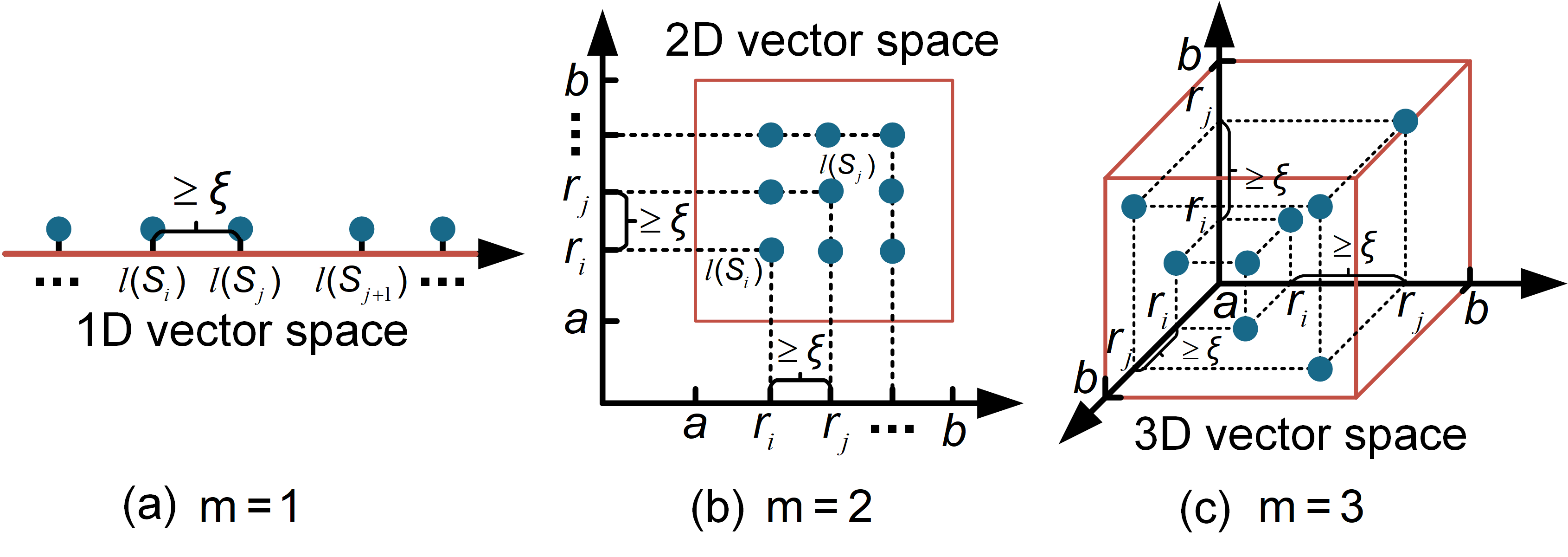}
	\caption{The labels of the training data set for the GNN model.}
	\label{fig:gnn-labels}
	\vspace{-0.2em}
\end{figure}

\underline{\textit{Training Data Set:}}
Algorithm~\ref{alg:gnn-training} describes the GNN model training for anchored subgraph embedding. As described in subsection~\ref{sub:feature-subgraphs}, all possible anchored subgraphs for each edge $(u, v)$ with $d(u) \leq d^* \lor d(v) \leq d^*$ in the data graph $G$ constitute the training data set $D$. In order to make non-isomorphic anchored subgraphs learn different embeddings as much as possible via the GNN model, we assign a unique $m$-dimensional ($m \geq 1$) vector as a label to each anchored subgraph in $D$, and make the vectors far enough from each other (e.g., the distance at least $\xi$). If $m = 1$, when $D$ is large, the 1-dimensional vector results in the label range being too wide, which is not conducive to model convergence. For example, in Fig.~\ref{fig:gnn-labels}(a), in order to keep the distance between vectors at least $\xi$, the 1-dimensional labels have a wide range. However, when $m > 1$, labels can be restricted to a small space (e.g., a rectangle or a cube in Fig.~\ref{fig:gnn-labels}(b)(c).), and kept far enough from each other, which is beneficial for model convergence. Thus, we adopt an $m$-dimensional vector as the label, where $m > 1$. Specifically, we randomly choose a range $[a, b]$, and generate $\lceil \sqrt[m]{|D|} \rceil$ values $\{r_1, r_2, \dots, r_{\lceil \sqrt[m]{|D|} \rceil}\}$ from $[a, b]$, where the distance between any two values is at least $\xi$. The value of each dimension in the $m$-dimensional vector is $r \in \{r_1, r_2, \dots, r_{\lceil \sqrt[m]{|D|} \rceil}\}$. Therefore, we can combine ${\lceil \sqrt[m]{|D|} \rceil}^m \geq |D|$ types of $m$-dimensional vectors, and the Euclidean distance between any two vectors is at least $\xi$. These vectors are assigned to anchored subgraphs in $D$ as labels. 

\begin{algorithm}[t]
	\footnotesize
	\caption{GNN Model Training for Embedding}
	\label{alg:gnn-training}
	\KwIn{a data graph $G$, vertex degree threshold $d^*$, the hop value $k$ of the anchored subgraph, and training epochs $\rho$}
	\KwOut{a trained GNN model $\mathcal{M}$}  
	\tcp{\textbf{Generate a training data set $D$}}
	initial a training data set $D = \emptyset$  \\
	\For{each edge $(u, v) \in G$, where $L(u) \leq L(v)$}{
		\uIf{degree $d_u \leq d^*$}{
			extract all possible anchored subgraphs $S^k_G(u, v)$ for $(u, v)$ \\
			add all anchored subgraphs $S^k_G(u, v)$ to $D$ 
		}
		\ElseIf{degree $d_v \leq d^*$}{
			extract all possible anchored subgraphs $S^k_G(u, v)$ spanned by $(u, v)$ and some vertices $w$ with $d(v, w) \leq k$ \\
			add all anchored subgraphs $S^k_G(u, v)$ to $D$
		}
	}
	attach a unique label $l(S^k_G(u, v))$ for each graph $S^k_G(u, v)$ in $D$ \\
	\tcp{\textbf{Train a GNN Model Until Epochs $\rho$}}
	\While{$epochs < \rho$}{
		\For{each batch $B \subseteq D$}{
			obtain embedding vectors of graph in $B$ via $\mathcal{M}$ \\
			compute the loss function $\mathcal{L}(D)$ of $\mathcal{M}$ by Equation~\ref{eqn:gin-loss-function} \\
			$M.update(\mathcal{L}(D))$
		}
	}
	\Return{the trained GNN model $\mathcal{M}$} 
\end{algorithm}

\underline{\textit{Training Loss of the GNN Model:}}
Existing method GNN-PE~\cite{ye2024efficient} requires that the GNN model be trained until it overfits the training data set to ensure exact subgraph matching, which results in a high offline computation cost. In our GNN-AE, we store all edges with the same feature embedding in the data graph $G$ offline in an index collection without omission. We only need to ensure that isomorphic features have the same embedding, which can guarantee exact subgraph matching. The GNN model is naturally able to map isomorphic anchored subgraphs to the same embedding. Thus, our GNN model does not have to overfit the training data set.

In addition, the output of two non-isomorphic anchored subgraphs via the GNN model may be mapped to the same embedding, which is called \textbf{\textit{embedding conflict}}. An excellent GNN model should \textit{minimize occurrences of embedding conflicts, which can help the query edge filter out more invalid (i.e., does not belong to the final match results) candidate matches in subgraph matching} (refer to subsection~\ref{sub:matching-anchor-graph-embedding}). Therefore, for any two anchored subgraphs $S, S' \in D$, we define the loss function as follows:
\begin{equation}\label{eqn:gin-loss-function}
	\mathcal{L}(D) = \sum_{S, S' \in D} max(0, \|l(S) - l(S')\| - \|o(S) - o(S')\|).
\end{equation}
where $o(S)$ is the output embedding of the anchored subgraph $S$ via the GNN, and $l(S)$ is the attached unique label (i.e., a $m$-dimensional vector) of $S$ (line 9 of Algorithm~\ref{alg:gnn-training}). In the offline phase, we do not need to perform graph isomorphism testing for the anchored subgraphs in $D$ to distinguish isomorphic graphs. There are two cases in the loss function:

\noindent\textbullet~\emph{Case~1: $S \cong S'$.} If $S$ is isomorphic to $S'$, the GNN model naturally maps $S$ and $S'$ to the same embedding, that is, $o(S) = o(S')$. Thus, the formula in $\sum$ degenerates into $max(0, \|l(S) - l(S')\|)$. Since we assign different vectors to $S$ and $S'$ as labels before training, $\|l(S) - l(S')\| > 0$ and $\|l(S) - l(S')\|$ does not change with training epochs. In other words, for any two anchored subgraphs $S, S' \in D$, their loss in $\mathcal{L}(D)$ is fixed in training epochs.
	
\noindent\textbullet~\emph{Case~2: $S \not\cong S'$.} If $S$ is not isomorphic to $S'$, the loss function will make the distance between $o(S)$ and $o(S')$ not less than the distance between $l(S)$ and $l(S')$ as the training progresses. When the distance between $o(S)$ and $o(S')$ is less than the distance between $l(S)$ and $l(S')$, $max(0, \|l(S) - l(S')\| - \|o(S) - o(S')\|) > 0$, otherwise $max(\cdot) = 0$. Therefore, Equation~\ref{eqn:gin-loss-function} can make the dispersion of non-isomorphic $S$ and $S'$ in the embedding space no weaker than our assigned $l(S)$ and $l(S')$.

In summary, the GNN model ensures that isomorphic anchored subgraphs have the same embedding. Meanwhile, the loss function $\mathcal{L}(D)$ makes non-isomorphic subgraphs have different embeddings as much as possible. After training, we obtain embeddings for each anchored subgraph in $D$ via the trained GNN model (refer to line 10 of Algorithm~\ref{alg:gnn-ae-framework}).

\noindent{\bf The Quality of Anchored Subgraph Embeddings.}
We provide a theoretical guarantee on the quality of these anchored subgraph embeddings generated by the GNN model.

\underline{\textit{GNN Model Expressive Power vs. WL Test:}}
The Weisfeiler -Lehman (WL) graph isomorphism test~\cite{leman1968reduction} is a powerful measure that can distinguish most non-isomorphic graphs. It can measure the power of distinguishing non-isomorphic graphs for the model~\cite{xu2018powerful}. Here, we provide rigorous theorems for the GNN model in mapping isomorphic graphs and distinguishing non-isomorphic graphs from a WL test perspective. 


\begin{lemma}
	\label{lem:gnn-wl}
	Let $G$ and $G'$ be any two non-isomorphic graphs. If a graph neural network $M: G \to \mathbb{R}^m$ maps $G$ and $G'$ to different embeddings, the Weisfeiler-Lehman test also decides $G$ and $G'$ are not isomorphic.
\end{lemma}

The proof of Lemma~\ref{lem:gnn-wl} can be found in~\cite{xu2018powerful}. It shows that classic GNN can be as powerful as the WL test in distinguishing non-isomorphic graphs. Thus, we have Theorem~\ref{the:isomorphic-equal}. 

\begin{theorem}\label{the:isomorphic-equal}
	For two graphs $G$ and $G'$, if $G \cong G'$, we have their embeddings $o(G) = o(G')$.
\end{theorem}

\begin{proof}
	If $G \cong G'$, they must satisfy the WL test~\cite{leman1968reduction, xu2018powerful}. Then, the proof immediately follows from the inverse negative proposition of Lemma~\ref{lem:gnn-wl} that if the WL test decides that $G \cong G'$, a graph neural network $\mathcal{M}$ maps $G$ and $G'$ to the same embedding vector (i.e., $o(G) = o(G')$). 
\end{proof}

Theorem~\ref{the:isomorphic-equal} shows that for any two isomorphic anchored subgraphs, the GNN can map them to the same embedding, which is a theoretical guarantee that we do not lose any candidate matches for exact subgraph matching.

\begin{theorem}\label{the:wl-gin}
	For any $t \in \mathbb{Z}^+$, if the vertex degrees and the feature dimension are finite, there exists a $t$-layered GIN model that satisfies, for two graphs $G$ and $G'$, if the 1-WL distinguishes $G$ and $G'$ as non-isomorphic within $t$ rounds, the embeddings $o(G)$ and $o(G')$ via the GIN are different. 
\end{theorem}

\begin{proof}
	Given two graphs $G$ and $G'$, in each round of the 1-WL algorithm, the multi-set labels rehashing via the hash function is the key for 1-WL to distinguish $G$ and $G'$ as non-isomorphic. In the GIN, the aggregation and update functions can be seen as the hash function of the 1-WL~\cite{sato2020survey, zhang2024expressive}. Since the vertex degrees and the feature dimension of $G$ and $G'$ are finite, the round that 1-WL stops can be a constant $t$. In the GIN, the vertex features it supports are countable (refer to Lemma 5 in~\cite{xu2018powerful}). Thus, the MLPs in GIN can approximate an injective function~\cite{sato2020survey}, which makes the GIN as powerful as the hash function in 1-WL in terms of rehashing power. Thus, for two graphs $G$ and $G'$, if the 1-WL distinguishes $G$ and $G'$ as non-isomorphic within $t$ rounds, the embeddings $o(G)$ and $o(G')$ via the GIN are different. The theorem holds.
\end{proof}

Theorem~\ref{the:wl-gin} shows that \textit{the GIN is as strong as the 1-WL test, which provides a theoretical basis for applying a GIN model rather than other embedding methods} (e.g., graph factorization emdedding~\cite{choudhury}) in our framework. Besides, GIN is more flexible than the WL test in its ability to map graphs to the embedding space and handle continuous vertex features.   

\underline{\textit{The GNN Training w.r.t. The Embedding Quality:}}
We train a GNN model (with 2 GIN layers modeled by a 1-layer MLP, SUM as the readout function, the $readout$ hidden dimension $n = 10$, and the output embedding vector dimension $m = 3$) over a large training data set $D$, where by default $|V(G)| = 500K$, $avg\_deg(G) = 5$ and the vertex label domain size $|\Sigma| = 100$. The learning rate of the Adam optimizer $\eta = 0.001$ and the batch size is $1K \sim 2K$.

\begin{figure}[!t]
	\centering
	\begin{minipage}{\linewidth}
		\centering
		\subfloat[$|V(G)|=500K$]{\includegraphics[width=0.4\linewidth,height=1.1in]{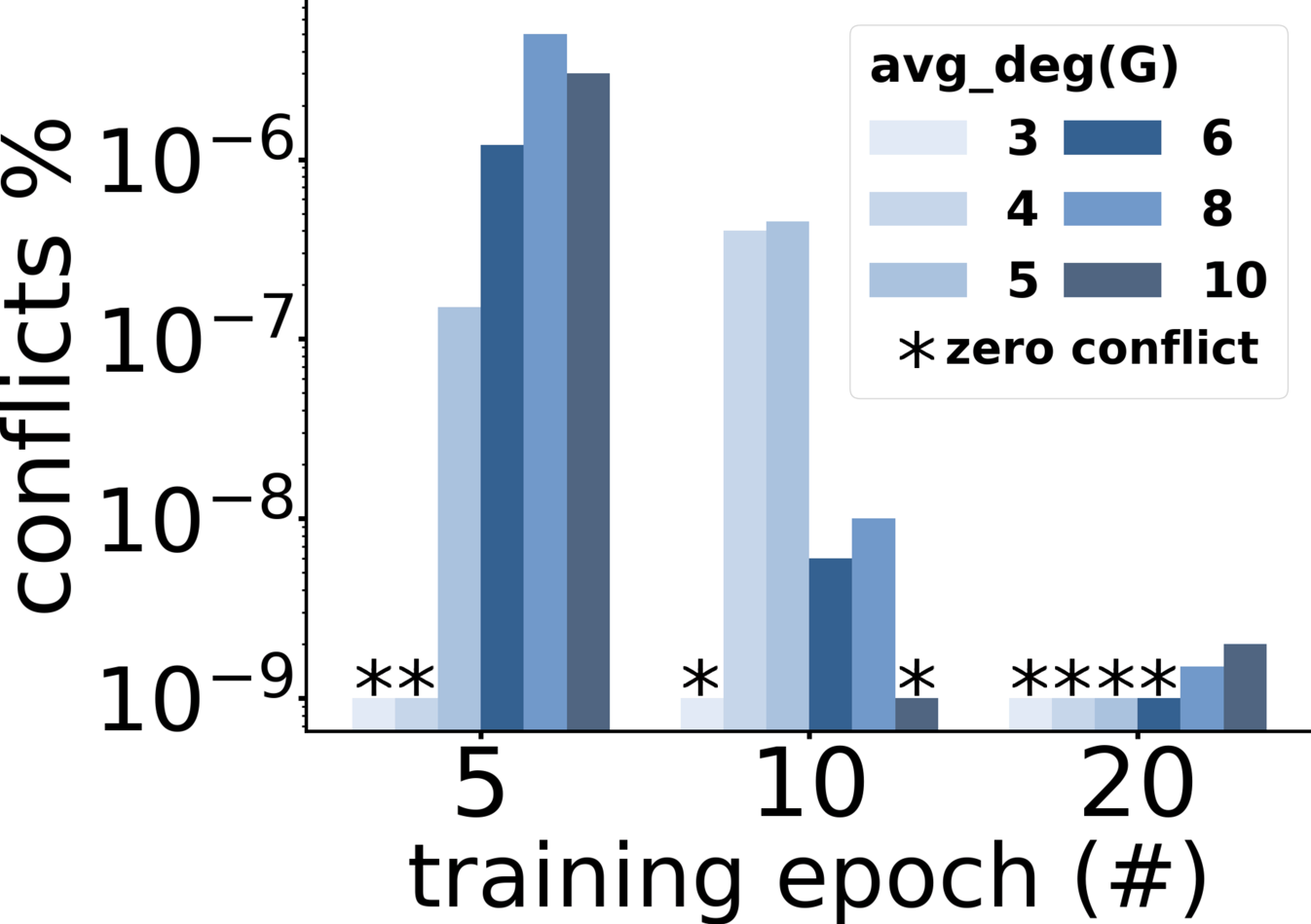}}
		\hspace{0.5cm}	
		\subfloat[$avg\_deg(G)=5$]{\includegraphics[width=0.4\linewidth,height=1.1in]{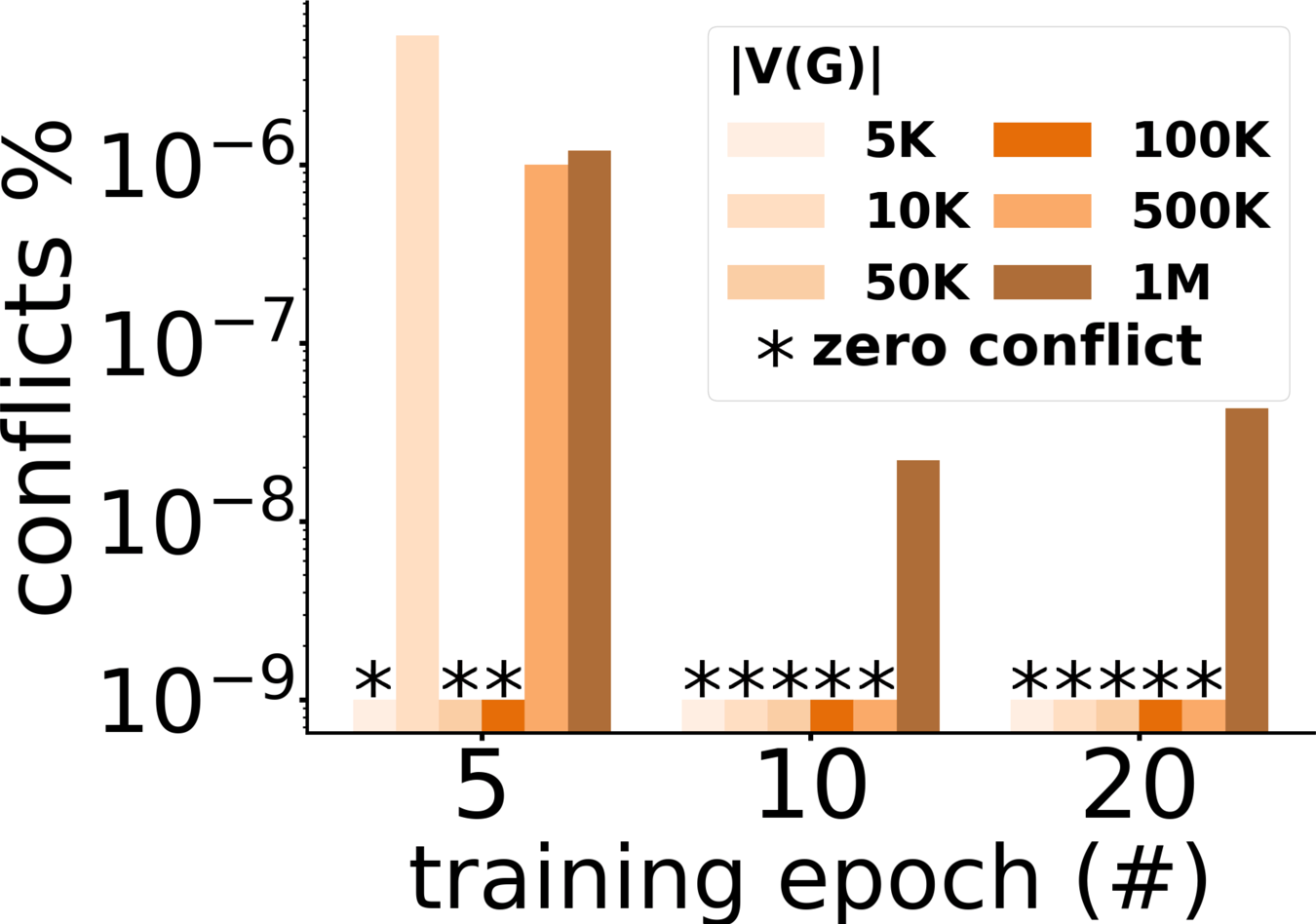}}
	\end{minipage}
	\caption{Illustration of the anchored subgraph embedding quality.}
	\label{fig:conlict-train-data}
	\vspace{-0.2em}
\end{figure}

Fig.~\ref{fig:conlict-train-data} reports the anchored 1-radius subgraph embedding conflict ratio of our GNN-AE. We vary $avg\_deg(G)$ from 3 to 10 in Fig.~\ref{fig:conlict-train-data}(a) and adjust $|V(G)|$ from $5K$ to $1M$ in Fig.~\ref{fig:conlict-train-data}(b). We adjust the training epochs $\rho \in \{5, 10, 20\}$. From the Fig.~\ref{fig:conlict-train-data}, the anchored subgraph embedding conflict ratio via the GNN model is not greater than $10^{-5}\%$. This means that almost any two non-isomorphic anchored subgraphs are mapped to non-conflicting embeddings. 

\noindent{\bf Complexity Analysis of the GNN Training.}
For each GIN layer, let the number of layers in the MLP be $\beta$. The initial feature embedding dimension of the input and the hidden dimension in the MLP are both $n$. Assume we obtain the anchored 1-radius subgraphs in the data graph $G$. For an anchored subgraph $S$ in $D$, $|E(S)| \le d_v$ and $|V(S)| = |E(S)| + 1$. The time complexity of the computation on the GIN layer is $O(\beta \cdot n^2 \cdot |V(S)| + |E(S)|) = O(\beta \cdot n^2 \cdot d_v)$. The upper bound of the number of anchored subgraphs in the training data set $D$ is $\sum_{v \in V(G) \land d_v \le d^*}{(2^{d_v} - 1)}$. The time complexity of the $readout$ function is $O(n \cdot |V(S)|)$, and the linear layer is $O(m \cdot n)$. We set the number of training epochs to be $\rho$. Thus, the total time complexity of the GNN training is $O(\rho \cdot \sum_{v \in V(G) \land d_v \le d^*}{2^{d_v}} \cdot (\beta \cdot n^2 \cdot d_v + m \cdot n))$.

\subsection{Matching via Anchored Subgraph Embedding}
\label{sub:matching-anchor-graph-embedding}

As described in subsection~\ref{sub:feature-subgraphs}, in our GNN-AE, we set a threshold $d^*$ on vertex degree of the data graph $G$. For an edge $(u, v)$ with $L(u) \leq L(v)$, if the vertex degree $d(u) \leq d^*$ or $d(v) \leq d^*$, the anchored $k$-radius subgraph $S^k_G(u, v)$ is adopted as features; otherwise, all anchored $k$-radius paths $P^k_G(u, v)$ are adopted as features. Specifically, we divide the edges $(u, v)$ with $L(u) \leq L(v)$ in $G$ into four types:
\begin{itemize}[leftmargin=1em]
	\item If $d_u \leq d^* \land d_v \leq d^*$, $(u, v)$ is named a sparse-sparse edge.
	\item If $d_u \leq d^* \land d_v > d^*$, $(u, v)$ is named a sparse-dense edge.
	\item If $d_u > d^* \land d_v \leq d^*$, $(u, v)$ is named a dense-sparse edge.
	\item If $d_u > d^* \land d_v > d^*$, $(u, v)$ is named a dense-dense edge.
\end{itemize}
In this subsection, for the first three types of edges $(u, v)$, we build the subgraph matching relationships from the query graph $Q$ to the data graph $G$ via the embeddings of anchored subgraphs. For the fourth type of edges $(u, v)$, we build the matching relationships via the embeddings of anchored paths, which will be discussed in subsection~\ref{sub:matching-anchor-path-embedding}.

\noindent{\bf Sparse-Sparse Edge \& Sparse-Dense Edge.}  
Fig.~\ref{fig:query-emb-graph} illustrates an example of the subgraph matching relationship from the query $Q$ to the data graph $G$ via anchored 1-radius subgraph embeddings of the sparse-sparse and sparse-dense edges.

\begin{figure}[!t]
	\setlength{\abovecaptionskip}{-0.05cm}
	\centering
	\includegraphics[width=\linewidth]{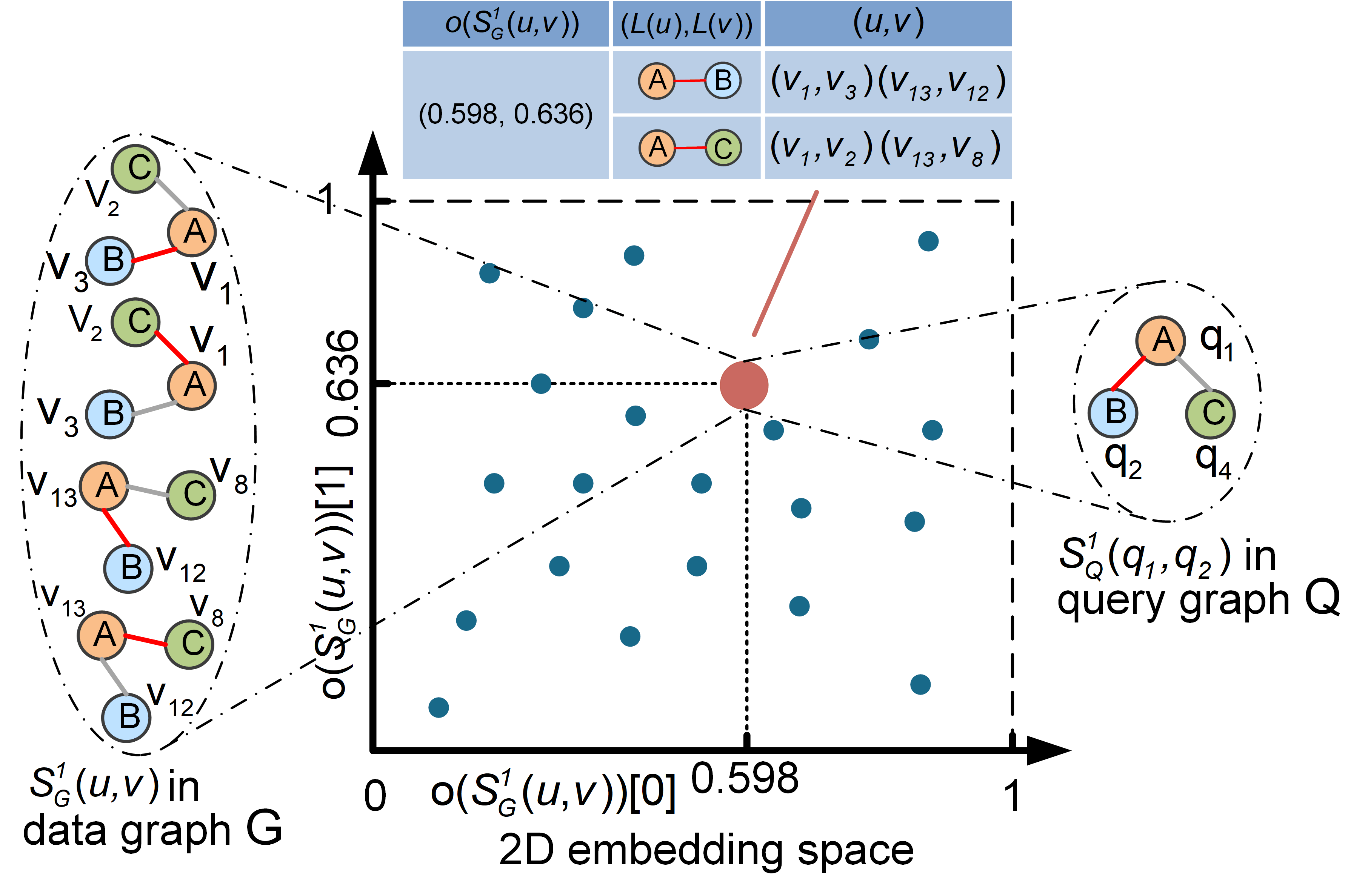}
	\caption{An example of building subgraph matching relationships via the anchored subgraph embeddings of sparse-sparse and sparse-dense edges.}
	\label{fig:query-emb-graph}
	\vspace{0em}
\end{figure} 

\begin{example}
	\label{exa:sm-relationships}
	Consider the data graph $G$ in Fig.~\ref{fig:example-sm}, let the vertex degree threshold $d^* = 3$. Edges $(v_1, v_3)$, $(v_1, v_2)$ and $(v_{13}, v_{12})$ are sparse-sparse edges, and $(v_{13}, v_8)$ is a sparse-dense edge. Their anchored 1-radius subgraphs have isomorphic graph structures mapped to the same position $(0.598, 0.636)$ in the embedding space via the GNNs. Each embedding maintains an entry, which consists of the embedding vector $o(S^1_G(u, v))$, edge labels $(L(u), L(v))$, and edges $(u, v)$. In Fig.~\ref{fig:query-emb-graph}, we merge edges with the same anchored subgraph embedding in $G$ into the same index entry. For an edge $(q_1, q_2)$ from the query graph $Q$, we extract its maximum anchored 1-radius subgraph $S^1_Q(q_1, q_2)$. Since $S^1_Q(q_1, q_2)$ is isomorphic to the above anchored subgraphs $o(S^1_G(u, v))$ from the data graph $G$, $S^1_Q(q_1, q_2)$ must be mapped to the same embedding position $(0.598, 0.636)$ by the GNN model. After aligning the edge labels $(L(u), L(v))$, we obtain the candidate matches $(v_1, v_3), (v_{13}, v_{12})$ (with sparse-sparse edge or sparse-dense edge types) in $G$ of the edge $(q_1, q_2)$.  
\end{example}

Due to embedding conflicts (i.e., two non-isomorphic graphs are mapped to the same embedding vector via the GNN model), some non-isomorphic anchored subgraphs may have the same embedding. The merge operation in Example~\ref{exa:sm-relationships} does not break the exact subgraph matching, which still guarantees a 100\% matching recall for each query edge. \textit{However, the non-isomorphic anchored subgraphs with embedding conflicts in the data graph $G$ will result in more candidate edges $(u, v)$ in an index entry. Furthermore, this results in more invalid candidate matches for a query edge and reduces query efficiency.} \textbf{Thus, it is important to ensure a high embedding quality in the embedding space (refer to subsection~\ref{sub:gnn-for-anchor-embedding})}. 

\noindent{\bf Dense-Sparse Edge.} 
For a dense-sparse edge $(u, v)$, the vertex degrees $d_u > d^* \land d_v \leq d^*$. According to the anchored subgraph definition~\ref{def:anchored-subgraph}, we need to obtain all possible subgraphs of $G$ spanned by the edge $(u, v)$ and some vertices $w$ with $d(u, w) \leq k$. However, due to the high degree of vertex $u$, there are a large number of anchored subgraphs. For example, a dense-sparse edge $(u, v)$ has at least $2^{d^*}$ anchored 1-radius subgraphs, which is unacceptable. Thus, for a dense-sparse edge, we obtain all its possible anchored subgraphs spanned by $(u, v)$ and some vertices $w$ with $d(v, w) \leq k$, denoted as $S'^k_G(u, v)$. Intuitively, this anchored subgraph $S'^k_G(u, v)$ is a subgraph of $G$ containing edge $(u, v)$ and spanned by extending vertices centered on the low degree vertex $v$.   

\begin{figure}[!t]
	\setlength{\abovecaptionskip}{-0.05cm}
	\centering
	\includegraphics[width=0.9\linewidth]{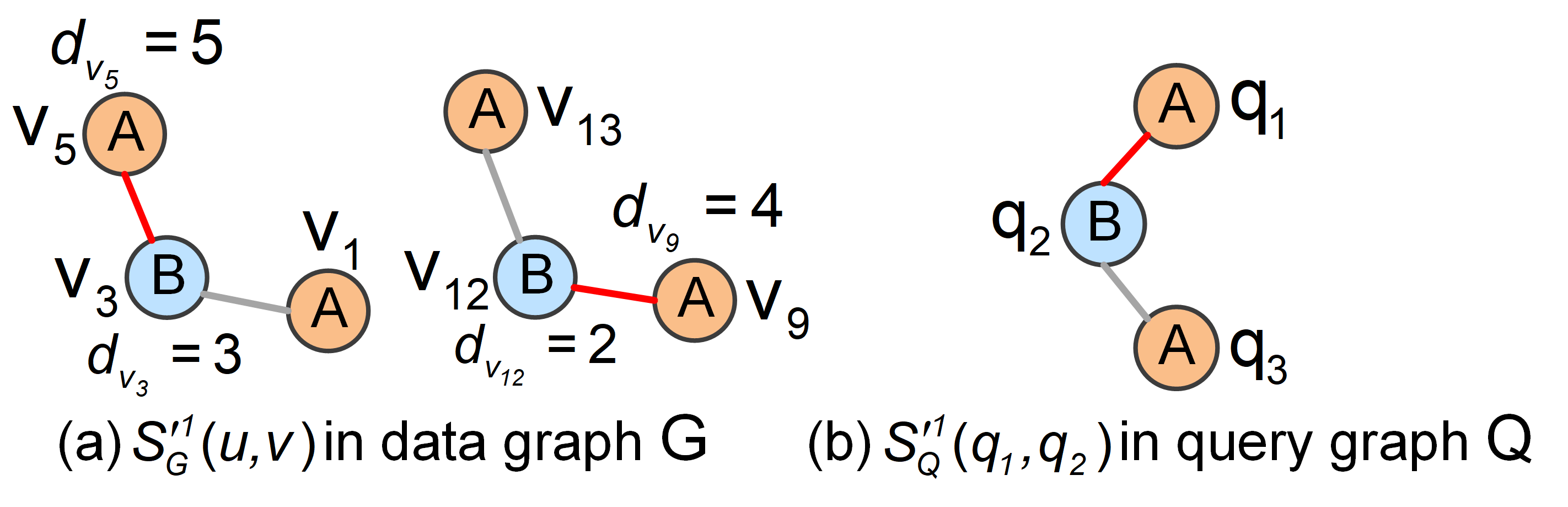}
	\caption{An example of the anchored subgraphs $S'^1_G(u, v)$ of the edge $(u, v)$ in the data graph $G$, and $S'^1_Q(q_1, q_2)$ of the edge $(q_1, q_2)$ in the query $Q$.}
	\label{fig:query-negative-graph}
	\vspace{0em}
\end{figure} 

\begin{example}
	\label{exa:query-negative-graph}
	Consider the data graph $G$ in Fig.~\ref{fig:example-sm}, let the acceptable vertex degree threshold $d^* = 3$. Fig.~\ref{fig:query-negative-graph}(a) illustrates an example of the anchored 1-radius subgraphs of the edge $(u, v)$ in the data graph $G$. In Fig.~\ref{fig:query-negative-graph}(a), the left is an anchored 1-radius subgraph $S'^1_G(v_5, v_3)$ of the dense-sparse edge $(v_5, v_3)$, and the right is an anchored 1-radius subgraph $S'^1_G(v_9, v_{12})$ of the dense-sparse edge $(v_9, v_{12})$. Fig.~\ref{fig:query-negative-graph}(b) illustrates the maximum anchored 1-radius subgraph $S'^1_Q(q_1, q_2)$ of the edge $(q_1, q_2)$ in the query graph $Q$.
\end{example}

For the anchored subgraphs $S'^1_G(u, v)$ and $S'^1_Q(q_1, q_2)$ in Fig.~\ref{fig:query-negative-graph}, we build subgraph matching relationships from the query $Q$ to the data graph $G$ similar to that in Fig.~\ref{fig:query-emb-graph}. $S'^1_G(u, v)$ and $S'^1_Q(q_1, q_2)$ in Fig.~\ref{fig:query-negative-graph} are mapped to the same embedding via the trained GNN model. After aligning edge labels $(L(u), L(v))$, we obtain candidate matches $(v_5, v_3), (v_9, v_{12})$ (with the dense-sparse edge type) in $G$ of the edge $(q_1, q_2)$. 

In Fig.~\ref{fig:query-emb-graph} and Fig.~\ref{fig:query-negative-graph}, we adopt anchored 1-radius subgraphs as examples. If the size or the maximum vertex degree of the data graph $G$ is not large, an anchored 2-radius or 3-radius subgraph is better. Because extracting all possible anchored 2-radius or 3-radius subgraphs is also acceptable, and due to the more complex structure, they can obtain fewer candidates for the query edge than the anchored 1-radius subgraph.

\noindent{\bf Complexity Analysis.}  
The anchored subgraph embeddings are obtained via the GNN model in subsection~\ref{sub:gnn-for-anchor-embedding}. We build a hash index $I_S$ for the sparse-sparse and sparse-dense edges and a hash index $I_{S'}$ for the dense-sparse edges to store embedding vectors and their additional maintained index entries. The time cost of building the indexes $I_S$ and $I_{S'}$ as follows:

When we obtain the anchored 1-radius subgraphs in the data graph $G$, the time complexity of obtaining anchored subgraph embeddings is $O(\sum_{v \in V(G) \land d_v \le d^*}{2^{d_v}} \cdot (\beta \cdot n^2 \cdot d_v + m \cdot n))$ (refer to the complexity analysis of the GNN training in subsection~\ref{sub:gnn-for-anchor-embedding}). During the index building, we merge the edges for the anchored subgraphs with the same embedding vector into the same index entry. The time complexity of the merge operation is $O(\sum_{v \in V(G) \land d_v \le d^*}{(2^{d_v} -1}))$. In summary, the total time complexity of building hash indexes for the anchored subgraphs is $O(\sum_{v \in V(G) \land d_v \le d^*}{2^{d_v}} \cdot (\beta \cdot n^2 \cdot d_v + m \cdot n))$.

\subsection{Matching via Anchored Path Embedding}
\label{sub:matching-anchor-path-embedding}

As described in subsection~\ref{sub:feature-subgraphs}, if the vertex degrees $d_u$ and $d_v$ of the edge $(u, v)$ are both very large in the data graph $G$, the number of anchored subgraphs $S^k_G(u, v)$ is also very large, making them ineffective to be a feature. For example, an edge $(u, v)$ with $d_u > d^* \land d_v > d^*$ has at least $2^{d^*}$ anchored 1-radius subgraphs. Extracting all possible anchored subgraphs will result in a high offline index storage cost. Thus, we adopt anchored paths (refer to Definition~\ref{def:anchored-path}) as features of dense-dense edges. Here, we discuss how to build matching relationships from the query graph $Q$ to the data graph $G$ via the anchored path embeddings for the dense-dense edge type.

\begin{figure}[!t]
	\setlength{\abovecaptionskip}{-0.05cm}
	\centering
	\includegraphics[width=\linewidth]{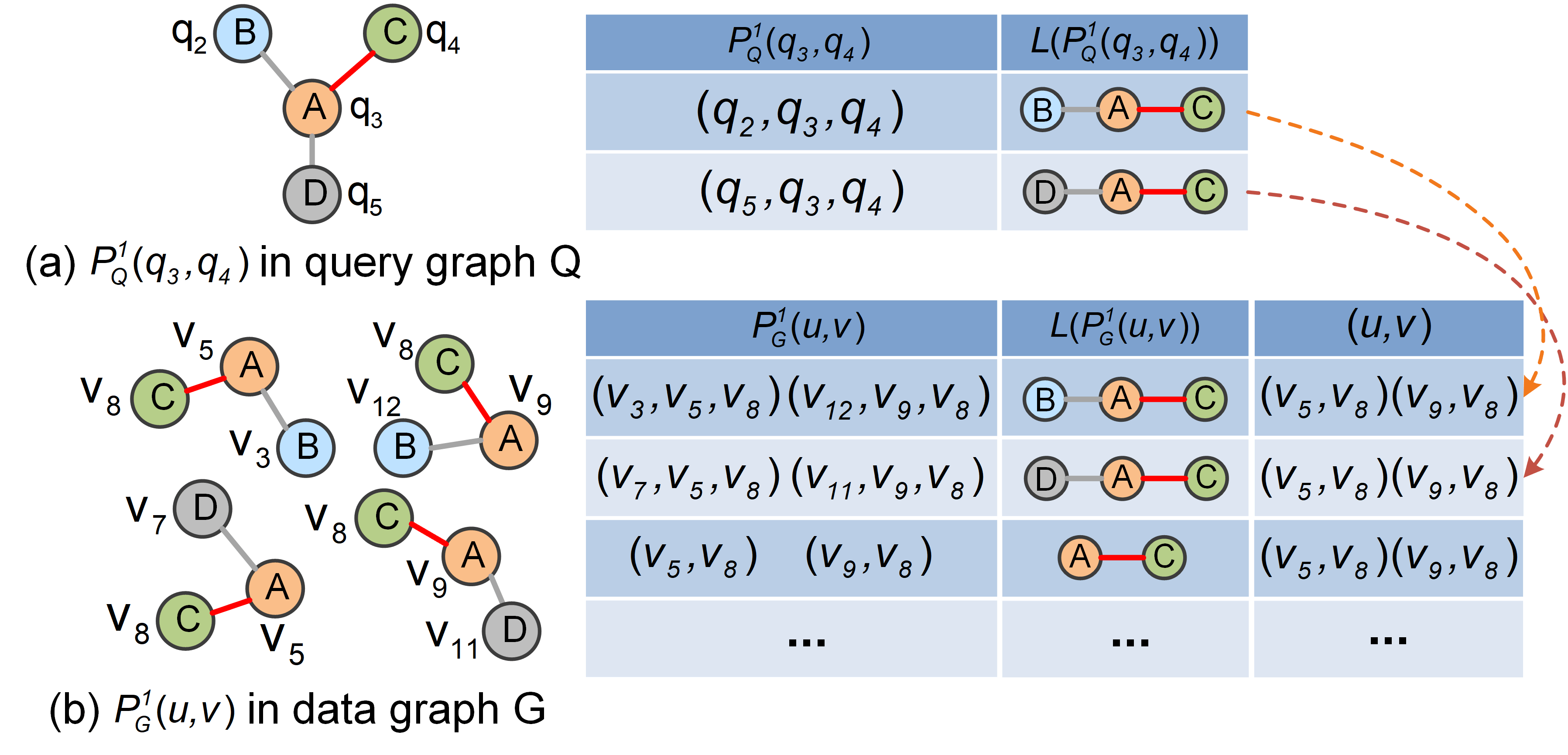}
	\caption{An example of building subgraph matching relationships via the anchored 1-radius path alignment.}
	\label{fig:path-alignment}
	\vspace{0em}
\end{figure} 

Unlike the anchored subgraph, the anchored path is simple in structure. There are usually multiple maximum anchored $k$-radius paths $P^k_G(u, v)$ with repect to $(u, v)$. Compared with the anchored subgraph, all possible anchored paths are greatly reduced. For example, an edge $(u, v)$ with $L(u) \leq L(v)$ has only $d_u$ possible anchored 1-radius paths. The lower bound on the number of anchored 1-radius paths for a dense-dense edge $(u, v)$ is $d^* + 1$, which is much less than $2^{d^*}$. Thus, extracting all possible anchored paths is acceptable.

\noindent{\bf Path Alignment \& Encoding.}  
Since the anchored path has a simple structure, we provide a \textit{path alignment} to directly build subgraph matching relationships from the query graph $Q$ to the data graph $G$ for the dense-dense edge type.

\begin{example}
	\label{exa:anchor-path-alignment}
	Consider the data graph $G$ and the query graph $Q$ in Fig.~\ref{fig:example-sm}, let the vertex degree threshold $d^* = 3$. $(v_5, v_8)$ and $(v_9, v_8)$ are dense-dense edges. We obtain all their possible anchored 1-radius paths. In Fig.~\ref{fig:path-alignment}, anchored paths $(v_3, v_5, v_8)$ and $(v_{12}, v_9, v_8)$ in the data graph $G$ have the isomorphic path structures. Thus, they are merged into the same index entry. For an edge $(q_3, q_4)$ in the query graph $Q$, its maximum anchored 1-radius path $(q_2, q_3, q_4)$ that can be aligned with anchored paths in the above entry in terms of path labels, so we obtain candidate matches $\{(v_5, v_8), (v_9, v_8)\}$ of the edge $(q_3, q_4)$. Similarly, we can obtain candidate matches $\{(v_5, v_8), (v_9, v_8)\}$ for $(q_3, q_4)$ via the maximum anchored 1-radius path $(q_5, q_3, q_4)$. The intersection $\{(v_5, v_8), (v_9, v_8)\}$ of two candidate match sets is the final candidate matches (with the dense-dense edge type) in $G$ of the edge $(q_3, q_4)$.
\end{example}

We provide a path encoding for efficient anchored path alignment. We encode an anchored path $(u_k, \dots, u_1, u, v)$ as a tuple of vertex labels $(L(u_k), \dots, L(u_1), L(u), L(v))$ on the path. For example, for the anchored path $(v_3, v_5, v_8)$, we encode it as a 3-tuple of vertex labels $(B, A, C)$. For a graph $G = (V, E, L, \Sigma)$, the upper bound on the number of anchored 1-radius path embeddings is $|\Sigma|^3 + |\Sigma|^2$. The proof is intuitive. The anchored 1-radius path has two possible cases $(u_1, u, v)$ and $(u, v)$. The number of possible anchored path embeddings is equivalent to the number of combinations of vertex labels on the path. Therefore, there are at most $|\Sigma|^3 + |\Sigma|^2$ anchored 1-radius path embeddings in the data graph $G$. 

\noindent{\bf Complexity Analysis.}  
Each anchored path is encoded and stored in a hash index $I_P$. If we extract the anchored 1-radius paths in the data graph $G$, the time complexity of obtaining all anchored paths in $P$ is $O(\sum_{(u, v) \in D'(G)}{d_u})$, where $D'(G)$ is the set of dense-dense edges in the data graph $G$. For each anchored path, we scan the vertex labels to encode its embedding. The time cost of encoding and merging operation is also $O(\sum_{(u, v) \in D'(G)}{d_u})$. Thus, the overall time complexity of building the hash index $I_P$ is $O(\sum_{(u, v) \in D'(G)}{d_u})$. 

\section{Subgraph Matching with Feature Embedding}
\label{sec:subgraph-matching-growth}

In this section, we introduce how to obtain the locations of all matches via anchored feature (anchored subgraph and path) embeddings (as discussed in Section~\ref{sec:gnn-based-feature-embedding}) in parallel.

\subsection{DFS-based Query Strategy}
\label{sub:dfs-query-strategy} 

Subgraph matching is finding the correspondence from query vertices to data graph vertices. The edges selected from the query graph $Q$ should cover all query vertices, without having to regard every edge in $Q$ as a matching unit. Thus, we design a \textit{depth-first search (DFS) query strategy}, and select the edges in the DFS traversal to match the edges in the data graph. The DFS query strategy divides the edges in the query graph into two types: \textit{DFS edges} (i.e., edges traversed by DFS), and \textit{non-DFS edges} (i.e., edges not traversed by DFS). 

Specifically, starting from a vertex in the query $Q$, we perform a DFS and obtain edges in the search order. Fig.~\ref{fig:dfs-query}(a) provides an example. Consider the query $Q$ in Fig.~\ref{fig:example-sm}, $Q$ has four DFS edges $\{(q_1, q_2), (q_2, q_3), (q_3q_4), (q_3, q_5)\}$, which are edges in the DFS pathway, and a non-DFS edge $(q_1, q_4)$, which is an edge not traversed by DFS. For each DFS edge in $Q$, we extract its maximum anchored $k$-radius subgraphs and maximum anchored $k$-radius paths (as shown in Fig.~\ref{fig:dfs-query}(b)).

\begin{figure}[!t]
	\setlength{\abovecaptionskip}{-0.05cm}
	\centering
	\includegraphics[width=0.98\linewidth]{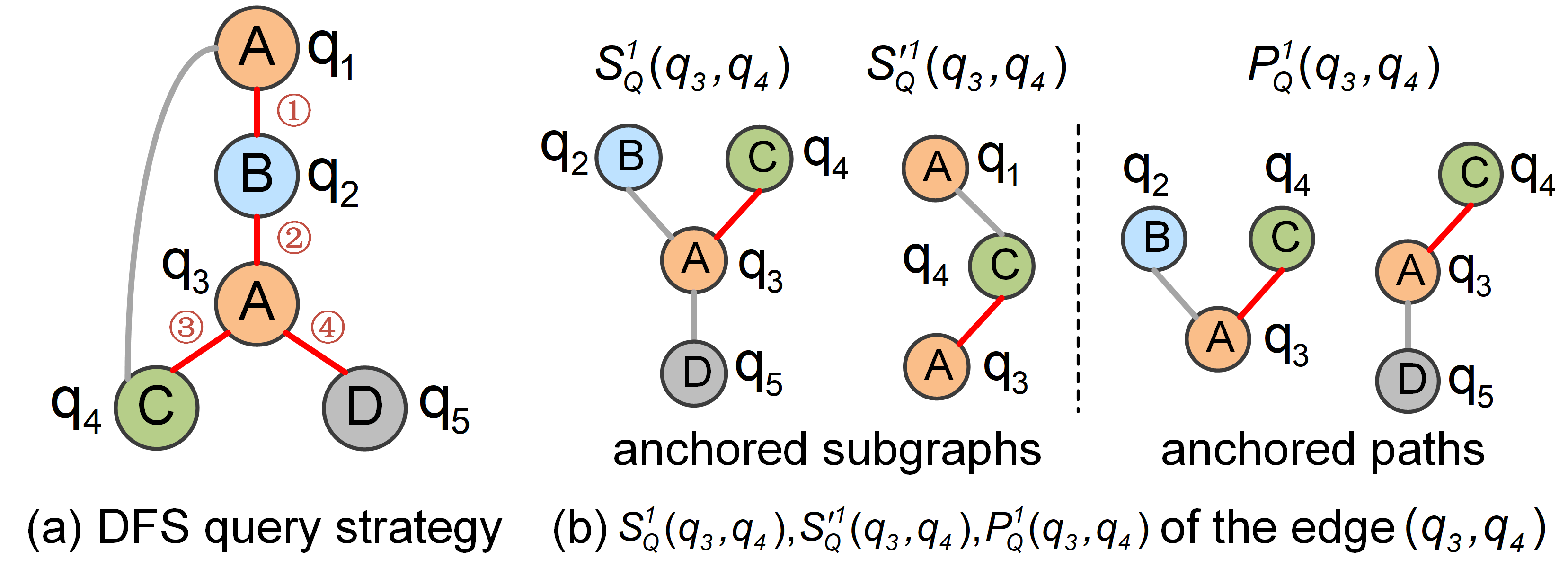}
	\caption{The DFS strategy for the query graph $Q$ and anchored subgraphs/paths for the edge $(q_3, q_4)$ in $Q$.}
	\label{fig:dfs-query}
	\vspace{0em}
\end{figure}

\subsection{Feature-based Anchor Matching Mechanism}
\label{sub:anchor-matching-maechanism} 

In subsection~\ref{sub:matching-anchor-graph-embedding} and~\ref{sub:matching-anchor-path-embedding}, if the query edge $(q_i, q_j)$ match the data graph edge $(u, v)$, the anchored subgraph (resp., anchored path) of $(q_i, q_j)$ must have the same embedding as a possible anchored subgraphs (resp., anchored paths) of $(u, v)$. In the offline phase, all edges in the data graph $G$ that have the same anchored feature embedding are stored in an entry in the index without omission (refer to Fig.~\ref{fig:query-emb-graph} and Fig.~\ref{fig:path-alignment}). In the online query phase, we process each DFS edge of the query graph $Q$ on the index. Thus, GNN-AE can guarantee a 100\% matching recall for each query DFS edge without false dismissals. Therefore, GNN-AE can obtain exact matching results. Algorithm~\ref{alg:subgraph-matching} describes the process of getting candidate matches for each query DFS edge (lines 2-7 and 10-17).

\begin{algorithm}[t]
	\footnotesize
	\caption{Exact Subgraph Matching with GNN-based Feature Embedding}
	\label{alg:subgraph-matching}
	\KwIn{a query graph $Q$, a trained GNN model $\mathcal{M}$, and hash indexes $I_S$, $I_S'$ and $I_P$ over the data graph $G$}
	\KwOut{a set $S$ of matching subgraphs}
	obtain the DFS edge set $E_{dfs}(Q)$ of the query, and query vertex permutation $\pi$ from a cost-model-based DFS query plan $\varphi$ \\
	\For{each DFS edge $(q_i, q_j) \in E_{dfs}(Q)$}{
		extract its maximum anchored $k$-radius subgraphs $S$ and $S'$  \\
		obtain $o(S)$ and $o(S')$ via the GNN model $\mathcal{M}$ \\
		encode each maximum anchored $k$-radius path $P^k_Q(q_i, q_j)$ \\
		obtain the encoding set $\Omega$ of all anchored paths \\
		$C_{q_iq_j} \gets$ \texttt{GetCandidate}$(o(S), o(S'), \Omega, I_S, I_{S'}, I_P)$ \\
		add $\{(q_i, q_j): C_{q_iq_j}\}$ to $C$ 
	}
	\Return{\texttt{MatchGrowth}$(\pi, \varphi, C)$} \\
	
	\SetKwFunction{MyFunction}{\texttt{GetCandidate}}
	\SetKwProg{Fn}{Function}{}{}
	\Fn{\MyFunction{$o(S), o(S'), \Omega, I_S, I_{S'}, I_P$}}{
		get candidate matches $C_S$ by retrieving $I_S$ via key $o(S)$ \\
		get candidate matches $C_{S'}$ by retrieving $I_{S'}$ via key $o(S')$ \\
		initial path candidate match set $C_P = E(G)$ \\
		\For{each path encoding $\omega \in \Omega$}{
			obtain candidate match set $C_\omega$ by retrieving $I_P$ via key $\omega$ \\
			$C_P \gets C_P \cap C_\omega$ 
		}
		\Return{$C_S \cup C_{S'} \cup C_P$}
	}   
	\SetKwFunction{MyProcedure}{\texttt{MatchGrowth}}
	\SetKwProg{Pro}{Procedure}{}{}
	\Pro{\MyProcedure{$\pi, \varphi, C$}}{
		build the candidate match table $T$ based on $\varphi$ and $C$ \\
		initial tree seeds are candidate matches of first edge $(q, q')$ in $\varphi$ \\
		\For(\tcp*[f]{\texttt{In parallel}}){each candidate $(u, v) \in C_{qq'}$}{ 
			initial a match tree $S_t = (u, v)$ \\
			\While(\tcp*[f]{\texttt{In parallel}}){the depth of $S_t.depth < |\pi|$}{
				assemble candidate matches in $T$ via hash join \\
				\eIf{no permutation conflicts \& non-DFS matching}{
					insert candidate match vertices to $S_t$
				}{
					Stop the growth of this branch in $S_t$  
				}	
			}
		}
		add all branches $g$ with depth $|\pi|$ to $S$ \\ 
		\Return{$S$}	
	}
\end{algorithm}

Specifically, for each DFS edge $(q_i, q_j)$ with $L(q_i) \leq L(q_j)$ in the query graph $Q$, we extract its maximum anchored $k$-radius subgraphs $S^k_Q(q_i, q_j)$, and maximum anchored subgraphs $S'^k_Q(q_i, q_j)$ which is spanned by extending vertices centered on $q_j$. We obtain their embedding vectors $o(S^k_Q(q_i, q_j))$ (and $o(S'^k_Q(q_i, q_j))$) via the trained GNN model (lines 3-4). Then, all maximum anchored $k$-radius paths $P^k_Q(q_i, q_j)$ of $(q_i, q_j)$ are extracted and encoded (lines 5-6). In \texttt{GetCandidate}, we take $o(S^k_Q(q_i, q_j))$ as the key to retrieve the hash index $I_S$ to obtain candidate matches $C_S$ with sparse-sparse or sparse-dense edge types. We use $o(S'^k_Q(q_i, q_j))$ to retrieve $I_{S'}$ to obtain candidate matches $C_{S'}$ with the dense-sparse edge type (lines 11-12). At last, we retrieve $I_P$ with the encoding of each anchored path as the key, and calculate the intersection to obtain candidate matches $C_P$ with the dense-dense edge type (lines 13-16). The union $C_S \cup C_{S'} \cup C_P$ of the above three candidate match sets is the final candidate matches of the edge $(q_i, q_j)$ in the query graph $Q$ (line 17).

\begin{proposition}\label{pro:candidate-anchor-set}
	Given the data graph $G$ and an edge $(q_i, q_j)$ in the query graph $Q$, the candidate match set of $(q_i, q_j)$ on graph $G$ is $C_{q_iq_j} = C_S \cup C_{S'} \cup C_P$, where $C_S \cap C_{S'} \cap C_P = \emptyset$.  
\end{proposition}

\noindent{\bf Graph-based Matching.}
In subsection~\ref{sub:matching-anchor-graph-embedding}, we illustrate the principle of building subgraph matching relationships via anchored subgraph embeddings. We can derive a necessary condition of an edge $(u, v)$ in the data graph $G$ to be a match of an edge $(q_i, q_j)$ in the query graph $Q$. Therefore, we have the following Lemma~\ref{lem:graph-matching}. 

\begin{lemma}[Graph-based Candidate Matches]
	\label{lem:graph-matching}
	Given an edge $(u, v)$ with $L(u) \leq L(v)$ in the data graph $G$ and an edge $(q_i, q_j)$ with $L(q_i) \leq L(q_j)$ in the query graph $Q$, the edge $(u, v)$ can be a candidate match of $(q_i, q_j)$, if an anchored subgraph $S^k_G(u, v)$ of $(u, v)$ and the maximum anchored subgraph $S^k_Q(q_i, q_j)$ of $(q_i, q_j)$ satisfy $o(S^k_G(u, v)) = o(S^k_Q(q_i, q_j)) \land L(u) = L(q_i) \land L(v) = L(q_j)$. 
\end{lemma}

\noindent{\bf Path-based Matching.}
In subsection~\ref{sub:matching-anchor-path-embedding}, we illustrate the principle of building subgraph matching relationships via anchored path embeddings. Similarly, we can derive another necessary condition of an edge $(u, v)$ in the data graph $G$ to be a match of an edge $(q_i, q_j)$ in the query graph $Q$. Thus, we have the Lemma~\ref{lem:path-matching} of path-based candidate matches below.

\begin{lemma}[Path-based Candidate Matches]
	\label{lem:path-matching}
	Given an edge $(u, v)$ with $L(u) \leq L(v)$ in the data graph $G$ and an edge $(q_i, q_j)$ with $L(q_i) \leq L(q_j)$ in the query graph $Q$, the edge $(u, v)$ can be a candidate match of $(q_i, q_j)$, if for any maximum anchored path $P^k_Q(q_i, q_j)$ of $(q_i, q_j)$, the edge $(u, v)$ has an anchored path $P^k_G(u, v)$ aligned with $P^k_Q(q_i, q_j)$.
\end{lemma}

\subsection{Matching Growth Algorithm}
\label{sub:matching-growth-algorithm} 

We present the effective matching growth method to obtain all exact matches for the query $Q$ (lines 18-30 of Algorithm~\ref{alg:subgraph-matching}). Consider the data graph $G$ and the query $Q$ in Fig.~\ref{fig:example-sm}, Fig.~\ref{fig:query-candidate} shows candidate matches of each DFS edge in $Q$. We organize them into a table based on the order of the DFS traversal.

\noindent{\bf Hash Join \& Refinement.}
We take candidate matches $C_{qq'}$ of the first edge $(q, q')$ in the DFS query plan $\varphi$ as initial tree seeds (line 20 of Algorithm~\ref{alg:subgraph-matching}). For example, in Fig.~\ref{fig:query-candidate}, the candidate matches $\{(v_5, v_3), (v_1, v_3), (v_9, v_{12}), (v_{13}, v_{12})\}$ of the query edge $(q_1, q_2)$ are initial tree seeds. For each seed, we perform the \textit{hash join} on the table of candidate matches to grow the match trees.

\begin{example}
	\label{exa:match-tree-growth}
	Consider the data graph $G$ and the query graph $Q$ in Fig.~\ref{fig:example-sm}, we obtain query vertex permutation $\pi = (q_1, q_2, q_3, q_4, q_5)$ according to the DFS search order. We grow match trees along $\pi$. In Fig.~\ref{fig:query-candidate}, the query edge $(q_1, q_2)$ has 4 candidate matches, so there are 4 match trees in Fig.~\ref{fig:match-tree}. For example, in tree 2, we take $(v_1, v_3)$ as the initial seed. To obtain matches of $(q_1, q_2, q_3)$ in tree 2, we join the candidate match $(v_1, v_3)$ of the edge $(q_1, q_2)$ with the candidate match $(v_3, v_5)$ of the edge $(q_2, q_3)$. Because the edge $(q_1, q_2)$ and the edge $(q_2, q_3)$ are connected via the vertex $q_2$ in the DFS traversal, and the matches of their $q_2$ are both $v_3$. After joining, we obtain a match $(v_1, v_3, v_5)$ of $(q_1, q_2, q_3)$. All exact matches are produced in match trees.
\end{example}

\begin{figure}[!t]
	\setlength{\abovecaptionskip}{0.01cm}
	\centering
	\includegraphics[width=0.92\linewidth]{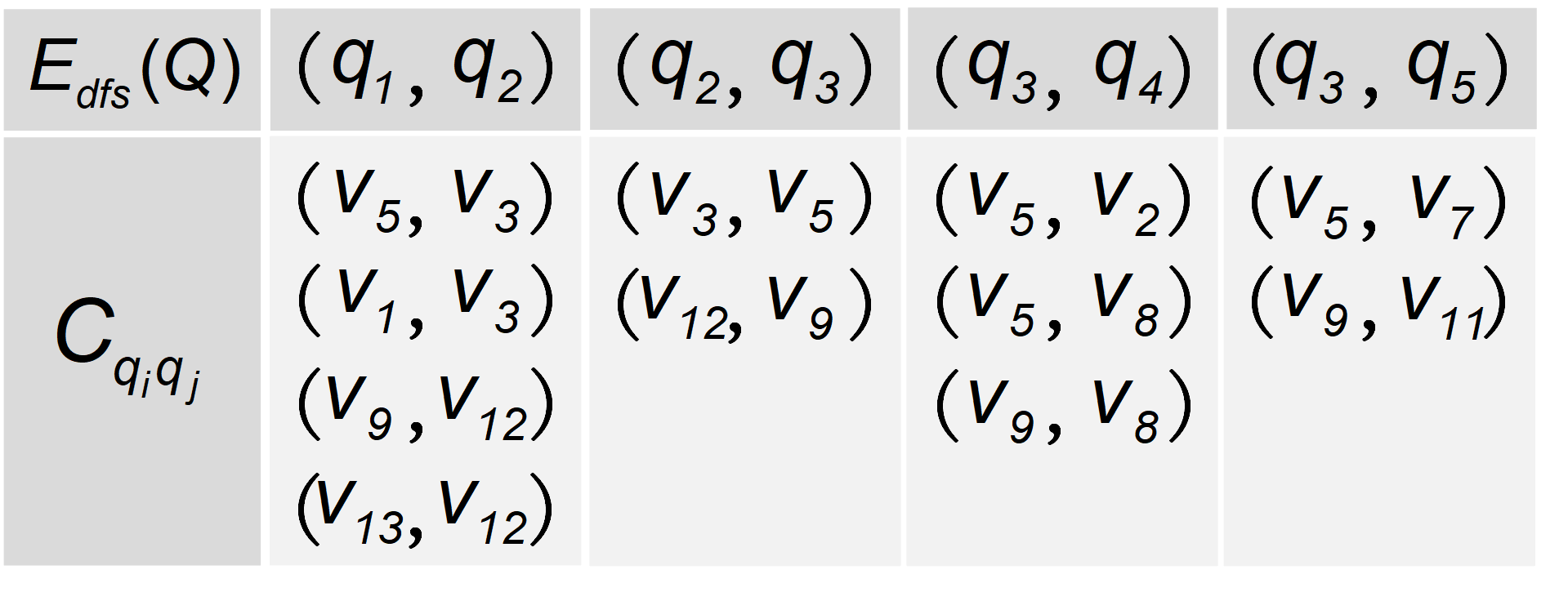}
	\caption{An example of the table of candidate matches.}
	\label{fig:query-candidate}
	\vspace{-1em}
\end{figure}

\begin{figure}[!t]
	\setlength{\abovecaptionskip}{0.05cm}
	\centering
	\includegraphics[width=0.82\linewidth]{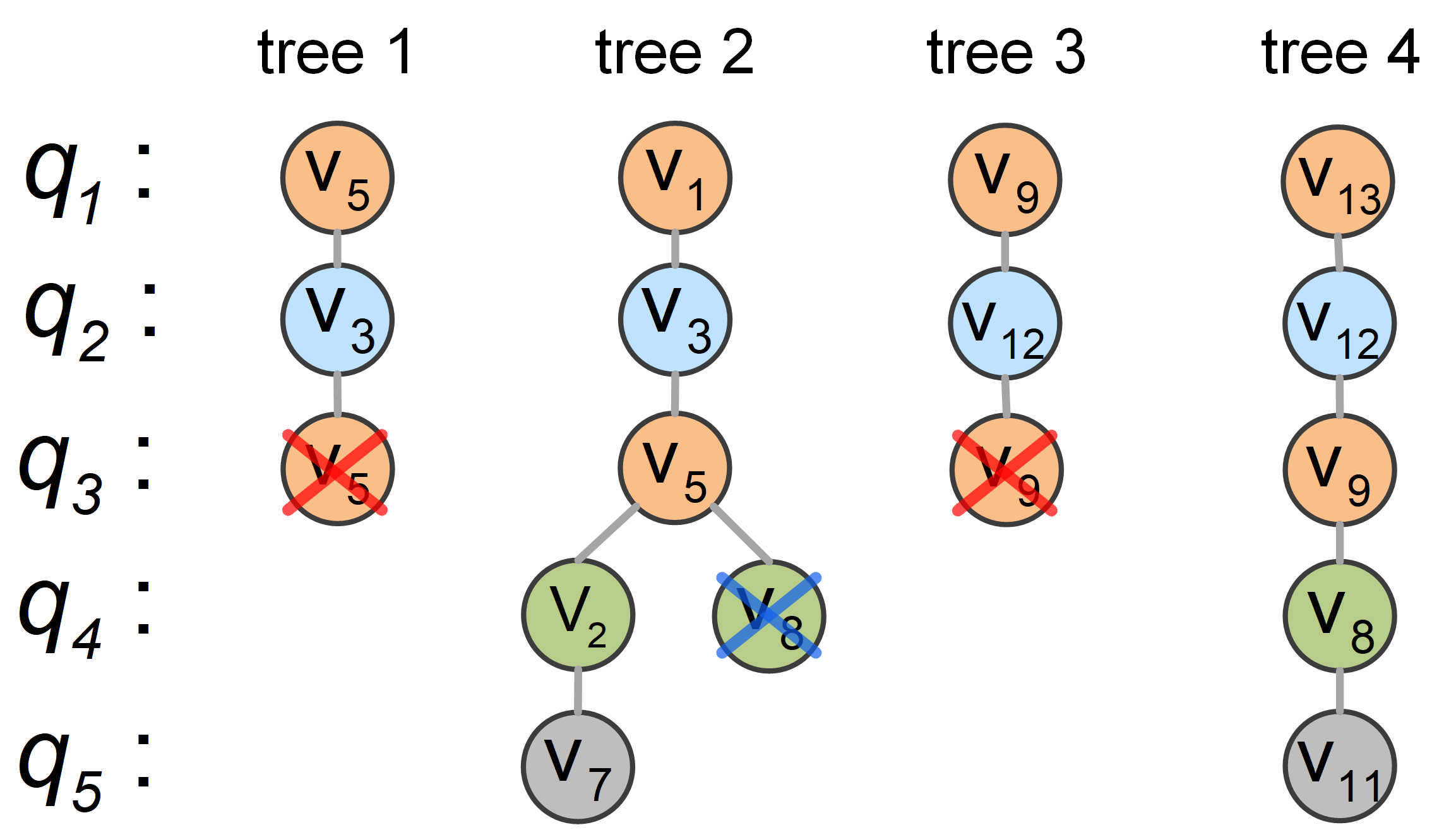}
	\caption{An example of the match tree growth.}
	\label{fig:match-tree}
	\vspace{0em}
\end{figure}

In Fig.~\ref{fig:match-tree}, there are 2 cases in which branches stop growing.

\noindent\textbullet~~\emph{Case~1. (Permutation Conflict)} \textit{If matches of two query vertices along permutation $\pi$ conflict, the branch stops growing.}

\begin{example}
	\label{exa:permutation-conflict-pruning}
	The match $\{(q_1, v_5), (q_2, v_3), (q_3, v_5)\}$ in tree 1 has a conflict and this branch stops growing.
\end{example}

\noindent\textbullet~~\emph{Case~2. (No Match for Non-DFS Edge)} \textit{If a non-DFS edge $(q_i, q_j)$ exists between $q_j$ and its forward query vertex $q_i$ in $\pi$, we check whether there is an edge $(v_i, v_j)$ between their matching vertices in $G$. If not exist, the branch stops growing.}

\begin{example}
	\label{exa:non-anchor-edges-pruning}
	For match $\{(q_1, v_1), (q_2, v_3), (q_3, v_5), (q_4, v_8)\}$ in tree 2, there is a non-DFS edge $(q_1, q_4)$, but no edge $(v_1, v_8)$ in $G$. Thus, the branch $(v_1, v_3, v_5, v_8)$ stops growing.
\end{example}

Finally, all branches $g$ with depth $|V(Q)|$ (or $|\pi|$) are all exact matches (line 29 of Algorithm~\ref{alg:subgraph-matching}).

\noindent{\bf Parallel Growth.}
In the match tree growth, the match of each query vertex $q_i$ only relies on the matches of its forward query vertices in the permutation $\pi$. Thus, we have the Lemma~\ref{lem:growth-iso}.

\begin{lemma}[Growth Isolation]
	\label{lem:growth-iso}
	All match trees or branches in the tree grow isolated from each other.
\end{lemma}

Therefore, in Algorithm~\ref{alg:subgraph-matching}, we can grow each match tree (line 21) or each branch in a tree (line 23) in a parallel manner. 

\noindent{\bf Complexity Analysis.}
In Algorithm~\ref{alg:subgraph-matching}, the time complexity of obtaining the set of DFS edge $E_{dfs}(Q)$ and permutation $\pi$ is $O(|V(Q)|+|E(Q)|)$. The time complexity of obtaining anchored 1-radius subgraph embeddings via the GNN model is $O(\sum_{(q_i, q_j) \in E_{dfs}(Q)}{(d_{q_i} + d_{q_j})})$, and the time cost of encoding anchored paths is $O(\sum_{(q_i, q_j) \in E_{dfs}(Q)}{d_{q_i}})$.

For the \texttt{GetCandidate} function, we assume that the encoding set of anchored paths for $(q_i, q_j)$ is $\Omega_{q_iq_j}$. The time complexity of the \texttt{GetCandidate} function is given by $O(\sum_{(q_i, q_j) \in E_{dfs}(Q)}{(\sum_{\omega \in \Omega_{q_iq_j}}{|C^P_\omega|} + |C_{q_iq_j}^{S}| + |C_{q_iq_j}^{S'}|)})$, where $|C^P_\omega|$ is the number of candidate matches for the path encoding $\omega \in \Omega_{q_iq_j}$, $|C_{q_iq_j}^{S}|$ is the number of candidate matches (with sparse-sparse \& sparse-dense edge types) of $(q_i, q_j)$ and $|C_{q_iq_j}^{S'}|$ is the number of candidate matches (with the dense-sparse edge type) of $(q_i, q_j)$. Most existing methods build an auxiliary index of the vertex 1-hop neighbors online for each query, which occupies much time in the online querying. However, our GNN-AE builds hash indexes of the vertex neighborhood \textit{offline} and \textit{one-time only}. In the online querying, each query edge directly retrieves candidate matches in the indexes with $O(1)$ via hash keys generated by GNNs. Thus, our GNN-AE is efficient in terms of obtaining candidates.

In \texttt{MatchGrowth} procedure, the time complexity is given by $O(\sum_{(q_i, q_j) \in E_{dfs}(Q)}{(|C_{q_iq_j}^{S}| + |C_{q_iq_j}^{S'}| + |C_{q_iq_j}^{P}|)})$, where $|C_{q_iq_j}^{P}|$ is the number of dense-dense candidate matches. Since GNNs have a high embedding quality with rigorous theoretical guarantees, \textit{the subgraph embedding conflict ratio via the GNNs is extremely low}. For example, in Fig.~\ref{fig:conlict-train-data}, the conflict ratio is less than $10^{-5}\%$ on data graphs in experiments. This means that almost all non-isomorphic subgraphs have unique embeddings without conflicts. As a result, each query edge can obtain a compact candidate set (i.e., fewer invalid results) that matches feature subgraphs. Thus, fewer invalid intermediate results are generated in the matching growth procedure compared to existing methods. 

In summary, time complexity of Algorithm~\ref{alg:subgraph-matching} is given by $O(\sum_{(q_i, q_j) \in E_{dfs}(Q)}{(\sum_{\omega \in \Omega_{q_iq_j}}{|C^P_\omega|} + |C_{q_iq_j}| + d_{q_i} + d_{q_j})})$, where $|C_{q_iq_j}|$ is the number of candidate matches for $(q_i, q_j)$.

\section{Cost-Model-Based DFS Query Plan}
\label{sec:cost-model-dfs-plan}

In this section, we provide a cost-model-based DFS query plan (line 14 of Algorithm~\ref{alg:gnn-ae-framework}) to find a low-query-cost DFS strategy that further enhances the parallel matching growth algorithm in subsection~\ref{sub:matching-growth-algorithm}.

\begin{figure}[!t]
	\setlength{\abovecaptionskip}{0.01cm}
	\centering
	\includegraphics[width=0.82\linewidth]{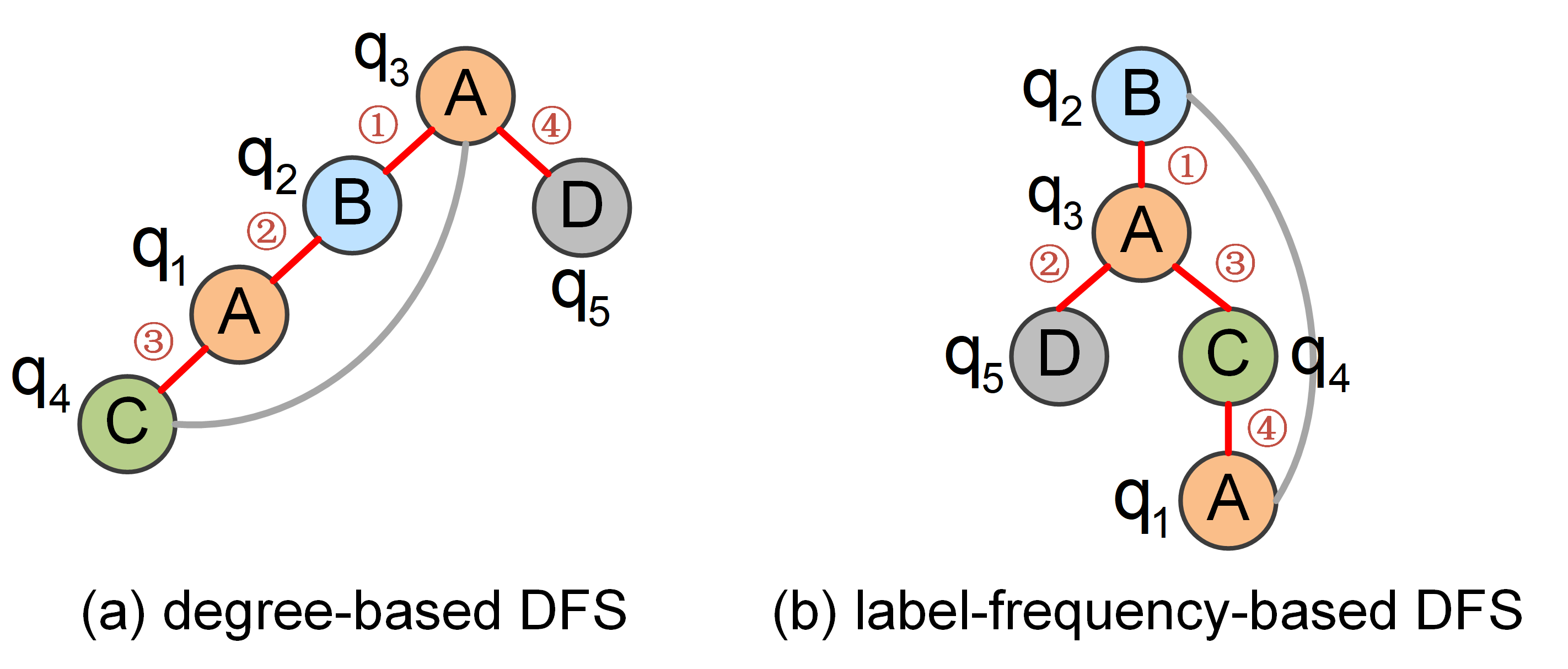}
	\caption{The degree-based and label-frequency-based DFS query plans.}
	\label{fig:cost-dfs-query}
	\vspace{-0.2em}
\end{figure}

\begin{algorithm}[t]
	\footnotesize
	\caption{Cost-Model-Based DFS Query Plan}
	\label{alg:dfs-query-plan}
	\KwIn{a query graph $Q$}
	\KwOut{a query DFS edge set $E_{dfs}(Q)$ in the query plan $\varphi$}
	initial $E_{dfs}(Q) = \emptyset$ and $Cost(\varphi) = +\infty$ \\
	\tcp{apply MaxDeg, MinLF, or Rand strategy}
	obtain a set $M_{st}$ of DFS start vertices \\
	\For{each start vertex $\mu \in M_{st}$}{
		initial $E_{tmp}(Q) = \emptyset$ and $Cost_{tmp}(\varphi) = 0$ \\
		DFS to select a vertex $q_i$ with minimum query cost $f(\mu, q_i)$ \\
		add $(\mu, q_i)$ to $E_{tmp}(Q)$ \\
		\While{at least one vertex in $V(Q)$ is not visited}{
			DFS to select a vertex with minimum query cost $f(q_i, q_j)$  \\
			add $(q_i, q_j)$ to $E_{tmp}(Q)$ \\
			$Cost_{tmp}(\varphi) \gets Cost_{tmp}(\varphi) + f(q_i, q_j)$
		}
		\If{$Cost_{tmp}(\varphi) < Cost(\varphi)$}{
			$E_{dfs}(Q) \gets E_{tmp}(Q)$ \\
			$Cost(\varphi) \gets Cost_{tmp}(\varphi)$ 
		}	
	}
	\Return{$E_{dfs}(Q)$}
\end{algorithm}

\subsection{Cost Model}
\label{sub:cost-model}

We design a formal cost model to estimate the query cost of the query plan $\varphi$. The query cost of subgraph matching is mainly limited by the number of candidate matches $|C^P_\omega|$ and $|C_{q_iq_j}|$. Thus, intuitively, fewer candidate matches for a query edge would lower query cost. Therefore, we define the query cost $Cost(\varphi)$ for $\varphi$ as follows:
\begin{equation}\label{eqn:cost-model}
	Cost(\varphi) = \sum_{(q_i, q_j) \in E_{dfs}(Q)} f(q_i, q_j),
\end{equation}
where $E_{dfs}(Q)$ is the query DFS edge set and $f(q_i, q_j)$ is the query cost of a query edge $(q_i, q_j)$. We aim to find a good DFS query plan $\varphi$ to minimize the total query cost $Cost(\varphi)$.

\noindent{\bf Edge Query Cost.}
We discuss how to compute the query cost $f(q_i, q_j)$ for a query edge $(q_i, q_j)$ in Equation~\ref{eqn:cost-model}. Intuitively, for an edge $(q_i, q_j)$, if the degrees of vertices $q_i$ and $q_j$ are high, the number of candidate matches of $(q_i, q_j)$ is expected to be small, which incurs a low query cost. Thus, we can set $f(q_i, q_j) = -(d_{q_i} + d_{q_j})$, where $d_{q_i}$ is the degree of vertex $q_i$ and $d_{q_j}$ is the degree of vertex $q_j$.

Alternatively, we can adopt other query cost metrics, such as vertex label frequency. For an edge $(q_i, q_j)$ in the query $Q$, we can compute the minimum label frequency in the 1-hop neighborhood of $q_i$ and $q_j$. That is, $f(q_i, q_j) = min(\Gamma(q_i)) + min(\Gamma(q_j))$, where $\Gamma(q_i)$ and $\Gamma(q_j)$ are the vertex label frequency set in the 1-hop neighborhood of $q_i$ and $q_j$ respectively.

Fig.~\ref{fig:cost-dfs-query} shows an example of the degree-based DFS query plan and the label-frequency-based DFS query plan. Consider the data graph $G$ and the query graph $Q$ in Fig.~\ref{fig:example-sm}. In Fig.~\ref{fig:cost-dfs-query}(a), we assume that DFS starts from the vertex $q_3$ with the maximum degree in the query $Q$. In the DFS procedure, we select the vertex with the minimum edge query cost $f(q_i, q_j) = -(d_{q_i} + d_{q_j})$ from the neighbors of the current vertex as the next vertex to traverse. We thus obtain a degree-based DFS query plan and the DFS edge set $\{(q_3, q_2), (q_2, q_1), (q_1, q_4), (q_3, q_5)\}$. Similarly, in Fig.~\ref{fig:cost-dfs-query}(b), we assume that DFS starts from the vertex $q_2$ (or $q_5$) with the minimum vertex label frequency. We can obtain a label-frequency-based DFS query plan and the the DFS edge set $\{(q_2, q_3), (q_3, q_5), (q_3, q_4), (q_4, q_1)\}$ by using the edge query cost function $f(q_i, q_j) = min(\Gamma(q_i)) + min(\Gamma(q_j))$.

\subsection{DFS-based Query Plan}
\label{sub:dfs-query-plan} 

Algorithm~\ref{alg:dfs-query-plan} shows the cost-model-based DFS query plan. For a query $Q$, we initialize $E_{dfs}(Q) = \emptyset$ and query cost $Cost(\varphi)$ (line 1). We select a vertex set $M_{st}$ in $Q$, each of which can be used as a start vertex of DFS (line 2). For a vertex $\mu \in M_{st}$, we perform DFS and minimize the cost $f(q_i, q_j)$ at each step expansion (lines 3-10). For different start vertices in $M_{st}$, we always keep the best-so-far DFS edge set in $E_{dfs}(Q)$ and the minimum cost in $Cost(\varphi)$ (lines 11-13). Finally, we obtain a DFS edge set $E_{dfs}(Q)$ with the lowest cost (line 14). 

\noindent{\bf DFS Start vertex Selection.}
We provide three strategies to select the set $M_{st}$ of DFS start vertices (line 2 of Algorithm~\ref{alg:dfs-query-plan}).

\begin{itemize}[leftmargin=1em]
	\item \textbf{Top-$K$ Maximum Degree vertices (MaxDeg)}: select top-$k$ vertices with maximum degree in the query graph $Q$. 
	\item \textbf{Top-$K$ Minimum Label Frequency vertices (MinLF)}: select top-$k$ vertices with minimum label frequency in $Q$.
	\item \textbf{$K$ Random vertices (Rand)}: randomly select $k$ vertices in the query graph $Q$ to form $M_{st}$.
\end{itemize}

\section{Experimental Evaluation}
\label{sec:experimental}

\subsection{Experimental Setup}
\label{sub:experimental-setup}

\begin{table}[!t]
	\centering
	\fontsize{6.3pt}{8pt}\selectfont 
	\renewcommand\arraystretch{1.3}
	\caption{Properties of real-world graph datasets}
	\label{tab:real-datasets}
	\begin{tabular}{|c|c|c|c|c|c|} 
		\hline
		\textbf{Category} & \textbf{Dataset} & \boldmath{$|V(G)|$} & \boldmath{$|E(G)|$} & \boldmath{$|\Sigma|$} & \boldmath{$avg\_deg(G)$} \\ 
		\hline
		\hline
		\multirow{2}{*}{\textbf{Biology}} & \textbf{Yeast (ye)} & 3,112 & 12,519 & 71 & 8.0 \\ 
		\cline{2-6}
		& \textbf{HPRD (hp)} & 9,460 & 34,998 & 307 & 7.4 \\ 
		\hline
		\hline
		\textbf{Lexical} & \textbf{WordNet (wn)} & 76,853 & 120,399 & 5 & 3.1 \\
		\hline
		\hline
		\multirow{2.2}{*}{\textbf{Social}} & \textbf{DBLP (db)} & 317,080 & 1,049,866 & 15 & 6.6 \\ 
		\cline{2-6}
		& \textbf{Youtube (yt)} & 1,134,890 & 2,987,624 & 25 & 5.3 \\
		\hline
		\hline
		\textbf{Citation} & \textbf{US Patents (up)} & 3,774,768 & 16,518,947 & 20 & 8.8 \\
		\hline
	\end{tabular}
\end{table}

We conduct subgraph matching experiments on a Linux machine equipped with an Intel Xeon E5-2620 v4 CPU and 220GB memory. The GNN model is implemented by PyTorch, where embeddings are offline computed on an Nvidia GeForce GTX 2080Ti Linux machine. The online subgraph matching is implemented in C++ for fair comparison with baselines.

For the GNN model, we set 2 GIN layers, the hidden dimension $n = 10$, and the output embedding dimension $m = 3$. The training epochs $\rho = 20$, and $k = 1$ for convenient evaluation. We use the Adam optimizer and set the learning rate $\eta = 0.001$. The bath sizes are adjusted in $[1024, 2048]$. We adopt a parallel mode with 8 threads for matching growth. 

\noindent{\bf Baseline Methods.}
We compare the performance of our GNN-AE approach with 7 representative exact subgraph matching baseline methods: GraphQL (GQL)~\cite{he2008graphs}, QuickSI (QSI)~\cite{shang2008taming}, RI~\cite{bonnici2013subgraph}, Hybrid (mixed by CFL~\cite{bi2016efficient}, RI~\cite{bonnici2013subgraph} and GQL)~\cite{sun2020memory}, RapidMatch (RM)~\cite{sun2021rapidmatch}, BICE~\cite{choi2023bice}, and GNN-PE~\cite{ye2024efficient}. 

\noindent{\bf Real/Synthetic Datasets.}
We conduct experiments on 6 real-world graphs from various domains and 3 synthetic graphs.

\underline{\textit{Real-world Graphs:}}
We selected six real-world datasets from four domains used by previous works~\cite{bi2016efficient, han2019efficient, shang2008taming, sun2020memory, zhao2010graph}. Table~\ref{tab:real-datasets} lists the properties of the real-world datasets. 

\underline{\textit{Synthetic Graphs:}}
The synthetic graphs were generated by the NetworkX~\cite{hagberg2020}. We produce three synthetic graphs: a random regular graph (Syn-RG), a small-world graph following the Newman-Watts-Strogatz model (Syn-WS)~\cite{watts1998collective}, and a scale-free graph following the Barabasi-Albert model (Syn-BA)~\cite{barabasi1999}. To evaluate the GNN-AE efficiency, we vary the characteristics of synthetic graphs. Table~\ref{tab:parameter-settings} depicts the parameters of synthetic graphs, where default parameter values are in bold. In the evaluation, we vary one parameter value while keeping the other parameters at default values.

\begin{table}[!t]
	\centering
	\fontsize{7pt}{8pt}\selectfont 
	\renewcommand\arraystretch{1.3}
	\caption{Parameter settings}
	\label{tab:parameter-settings}
	\begin{tabular}{|p{4.4cm}||p{3.35cm}|}   
		\hline
		\textbf{Parameter} & \textbf{Values} \\ 
		\hline
		\hline
		acceptable vertex degree threshold $d^*$ & 6, 8, \textbf{10}, 12, 14, 16 \\ 
		\hline
		the size, $|V(Q)|$, of query graph $Q$ & 4, 6, \textbf{8}, 10, 12, 16, 24, 32 \\
		\hline
		$avg\_deg(G)$ of synthetic data graph $G$ & 3, 4, \textbf{5}, 6, 8 \\
		\hline
		the number, $|\Sigma|$, of synthetic data graph $G$ & 20, 50, \textbf{100}, 200, 300, 500 \\
		\hline
		the size, $|V(G)|$, of synthetic data graph $G$ & 10K, 50K, \textbf{80K}, 100K, 500K, 1M \\
		\hline
	\end{tabular}
\end{table}

\noindent{\bf Query Graphs.}
We follow previous works~\cite{sun2021rapidmatch, sun2020memory, bi2016efficient, han2019efficient, ye2024efficient, choi2023bice} to generate query sizes and types for each data graph $G$. Specifically, for each dataset, we perform a random walk on $G$ until getting the specified number of vertices and extract the induced subgraph. Table~\ref{tab:parameter-settings} lists the size $|V(Q)|$ of query $Q$. Except for queries with $|V(Q)| = 4$, each query set on real-world graphs contains two categories: dense query graph (i.e., $avg\_deg(Q) > 3$) and sparse query graph (i.e., $avg\_deg(Q) \le 3$). For each query size and category, we generate 100 queries.

\noindent{\bf Reproducibility.} Our source code and data are available \href{https://github.com/BenYoungLab/GNN-AE}{\underline{here}}. 

\noindent{\bf Evaluation Metrics.}
We report the efficiency of our GNN-AE in terms of the total query time (including generating a DFS query plan, obtaining feature embeddings, getting candidate matches, and matching growth) for all queries by default. We also evaluate the \textit{filtering power} of our feature (anchored subgraph \& path) embeddings, which is the percentage of invalid (i.e., does not belong to the final matches) candidate matches that can be filtered out by our feature embeddings. At last, we also report offline pre-computation costs of GNN-AE.

\subsection{Efficiency and Effectiveness Evaluation}
\label{sub:efficiency-effectiveness}

\noindent{\bf The GNN-AE Efficiency on Real/Synthetic Datasets.}
In Fig.~\ref{fig:efficiency-baselines} and Fig.~\ref{fig:efficiency-online-cost}, we report the efficiency of GNN-AE.

\begin{figure}[!t]
	\centering
	\includegraphics[width=0.96\linewidth]{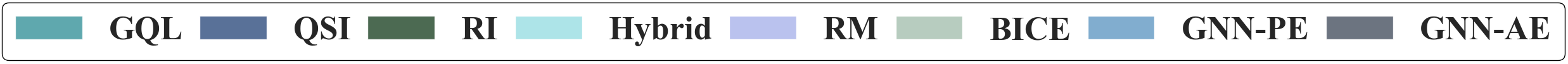}
	\vskip -0.15cm
	\begin{minipage}{\linewidth}
		\centering
		\subfloat[real-world datasets]{\includegraphics[width=0.46\linewidth,height=1.1in]{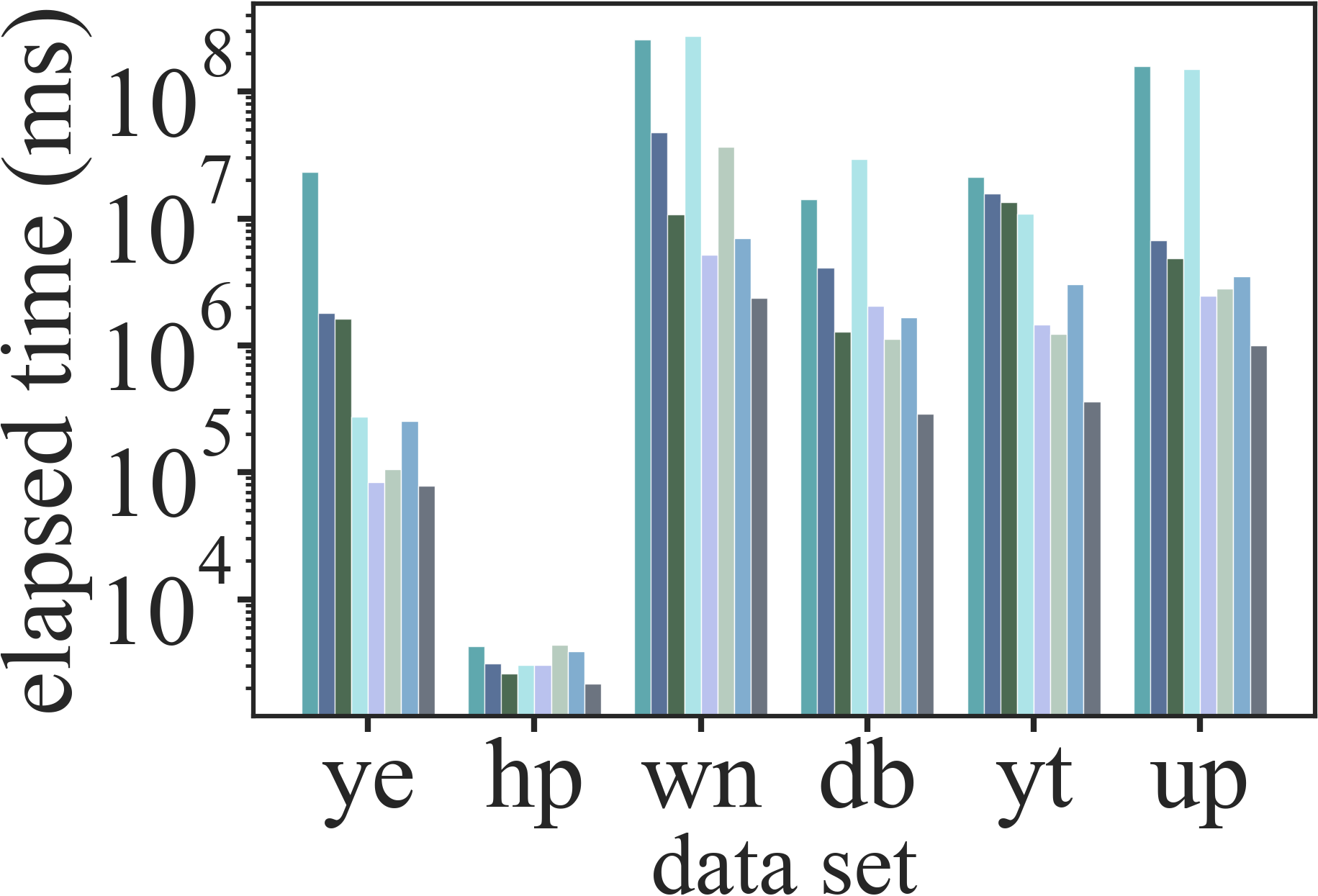}}
		\hspace{0.2cm}	
		\subfloat[synthetic datasets]{\includegraphics[width=0.46\linewidth,height=1.1in]{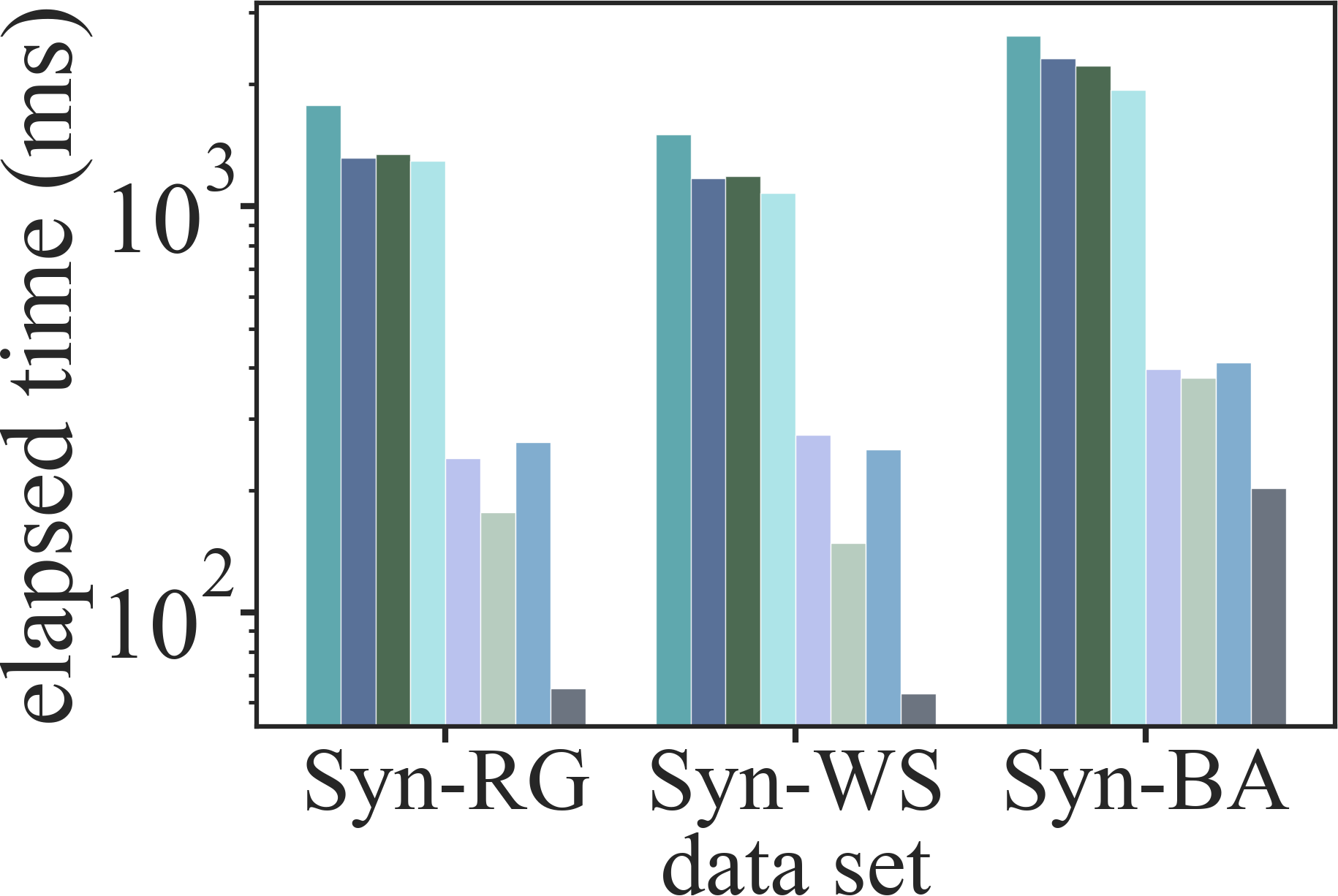}}
	\end{minipage}
	\caption{The GNN-AE efficiency on real and synthetic datasets.}
	\label{fig:efficiency-baselines}
	\vspace{-0.4em}
\end{figure}

\begin{figure}[!t]
	\centering
	\includegraphics[width=0.96\linewidth]{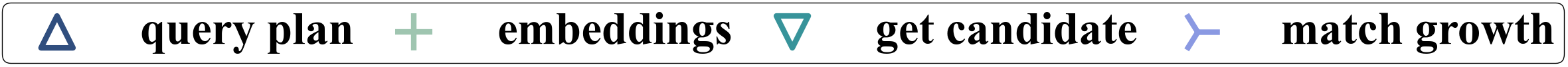}
	\vskip -0.15cm
	\begin{minipage}{\linewidth}
		\centering
		\subfloat[online computation on \textit{ye}]{\includegraphics[width=0.46\linewidth,height=1.1in]{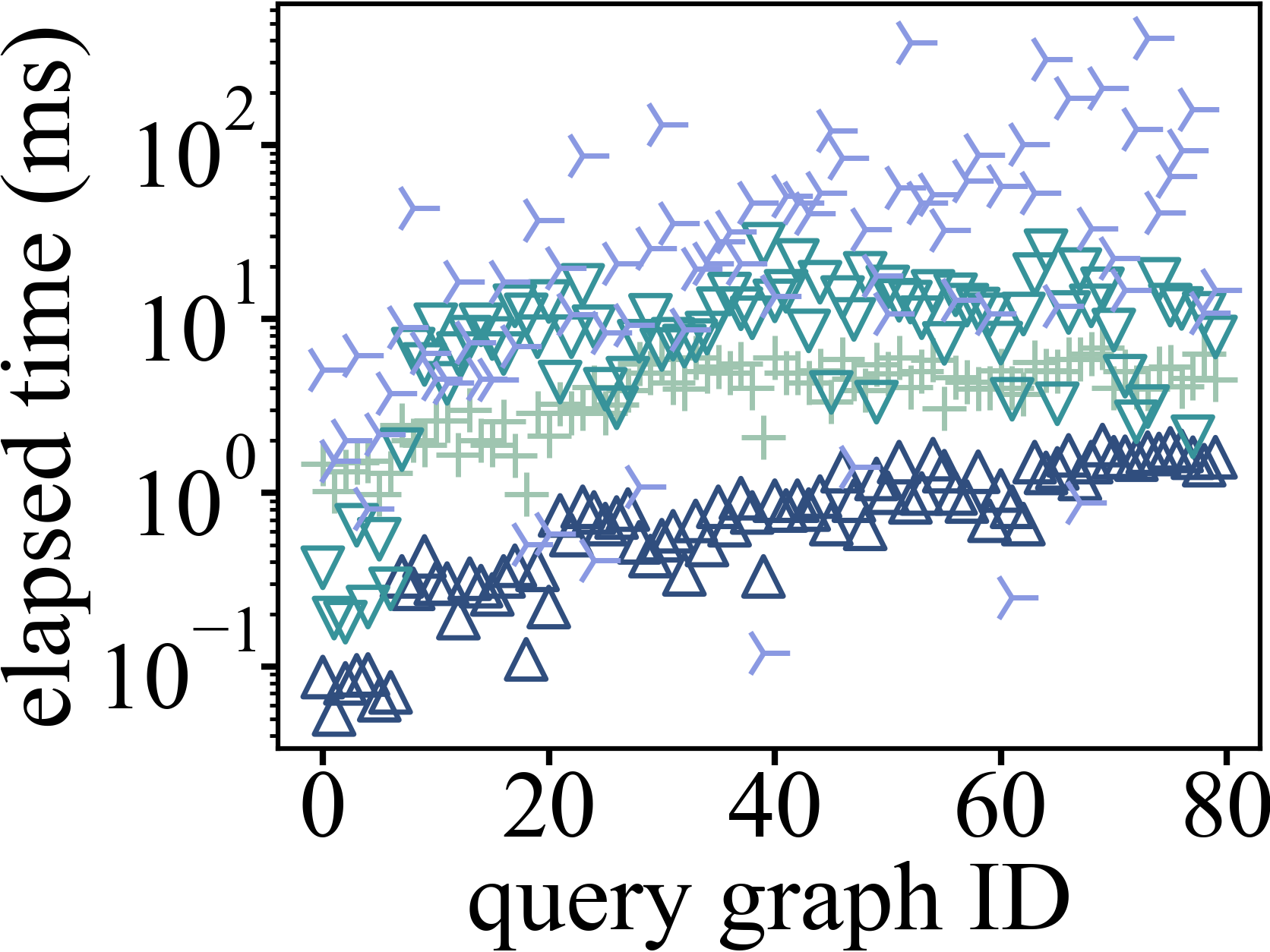}}
		\hspace{0.2cm}	
		\subfloat[online computation on \textit{wn}]{\includegraphics[width=0.46\linewidth,height=1.1in]{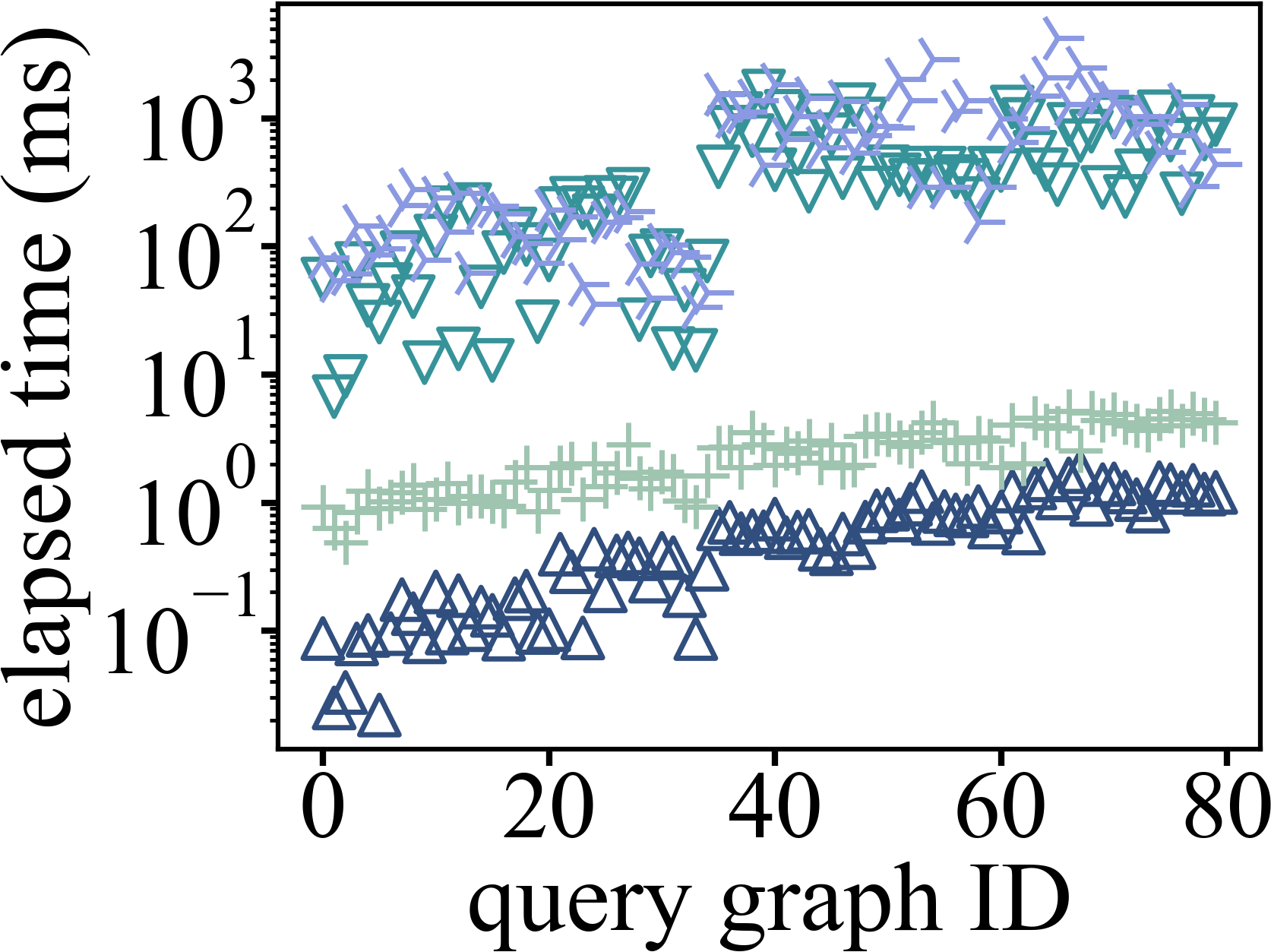}}
	\end{minipage}
	\caption{The GNN-AE online computation cost.}
	\label{fig:efficiency-online-cost}
	\vspace{0em}
\end{figure}  

\underline{\textit{Compare with Baselines:}}
Fig.~\ref{fig:efficiency-baselines} illustrates the total query time for all queries of our GNN-AE and 7 baselines over real-world and synthetic datasets. Overall, our GNN-AE consistently outperforms baseline methods. Specifically, for some real-world (e.g., \textit{db}, \textit{yt}, and \textit{up}) and synthetic (e.g., \textit{Syn-RG}, and \textit{Syn-WS}) datasets, our GNN-AE can perform better than baselines by up to 1--2 orders of magnitude. In particular, our GNN-AE significantly outperforms exploration-based backtracking search methods (e.g., GQL, QSI, and Hybrid). 

In subsequent experiments, we also evaluate the efficiency of our GNN-AE under varying query and data graph characteristics (e.g., query graph patterns, $|V(Q)|$, $|V(G)|$, $avg\_deg(G)$, and $|\Sigma|$ of $G$). To better show the curve trends, we omit the baseline results in subsequent experiments.

\underline{\textit{Online Computation Cost:}}
Fig.~\ref{fig:efficiency-online-cost} reports the computation cost of four steps (including DFS query plan, obtain feature embeddings, get candidate matches, and match growth) in the online subgraph matching phase (i.e., Algorithm~\ref{alg:subgraph-matching}). We select 80 query graphs of \textit{ye} and \textit{wn} respectively, and record the elapsed time of four steps. We find that the get candidate matches and match growth steps take up most of the elapsed time for subgraph matching. For the DFS query plan and obtaining feature embeddings, the time cost of our GNN-AE remains low (i.e., $<1.7$ $ms$, and $<6.5$ $ms$, respectively). 

\noindent{\bf The GNN-AE Efficiency w.r.t Query Graph Patterns.}
Fig.~\ref{fig:efficiency-query-patterns} presents the total query time of our GNN-AE on different query graph patterns on \textit{wn} dataset, including dense (i.e., $avg\_deg(Q) > 3$, denoted as $D$) and sparse (i.e., $avg\_deg(Q) \le 3$, denoted as $S$) query graphs with varying sizes. The figure shows that our GNN-AE can perform better than baselines on most query graph patterns. Especially, for most query patterns, GNN-AE can perform better than traditional exploration-based backtracking search methods (e.g., GQL, QSI, and Hybrid) and BICE by 1-2 orders of magnitude.

\begin{figure}[!t]
	\setlength{\abovecaptionskip}{0.02cm}
	\centering
	\includegraphics[width=\linewidth]{efficiency-bars-legend.png}
	\vskip 0.1cm
	\includegraphics[width=0.96\linewidth]{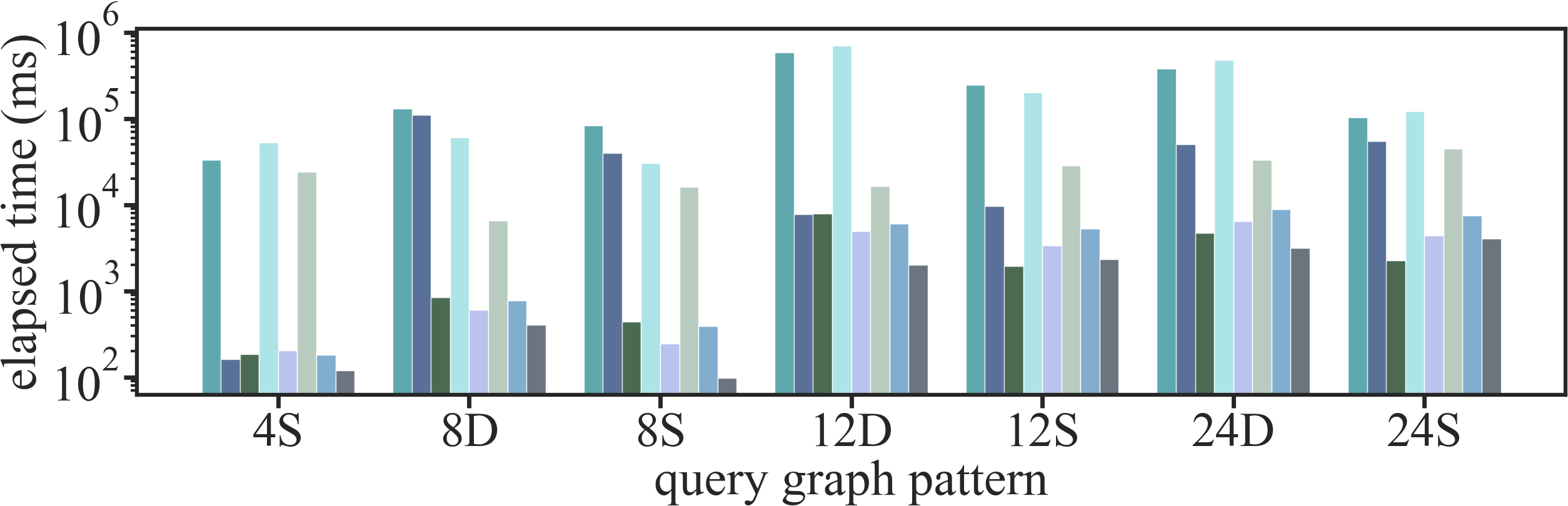}
	\caption{The GNN-AE efficiency w.r.t query graph patterns.}
	\label{fig:efficiency-query-patterns}
	\vspace{-1.9em}
\end{figure}

\begin{figure}[!t]
	\setlength{\abovecaptionskip}{0.2cm}
	\centering
	\begin{minipage}{\linewidth}
		\centering
		\subfloat[real-world datasets]{\includegraphics[width=0.38\linewidth,height=1in]{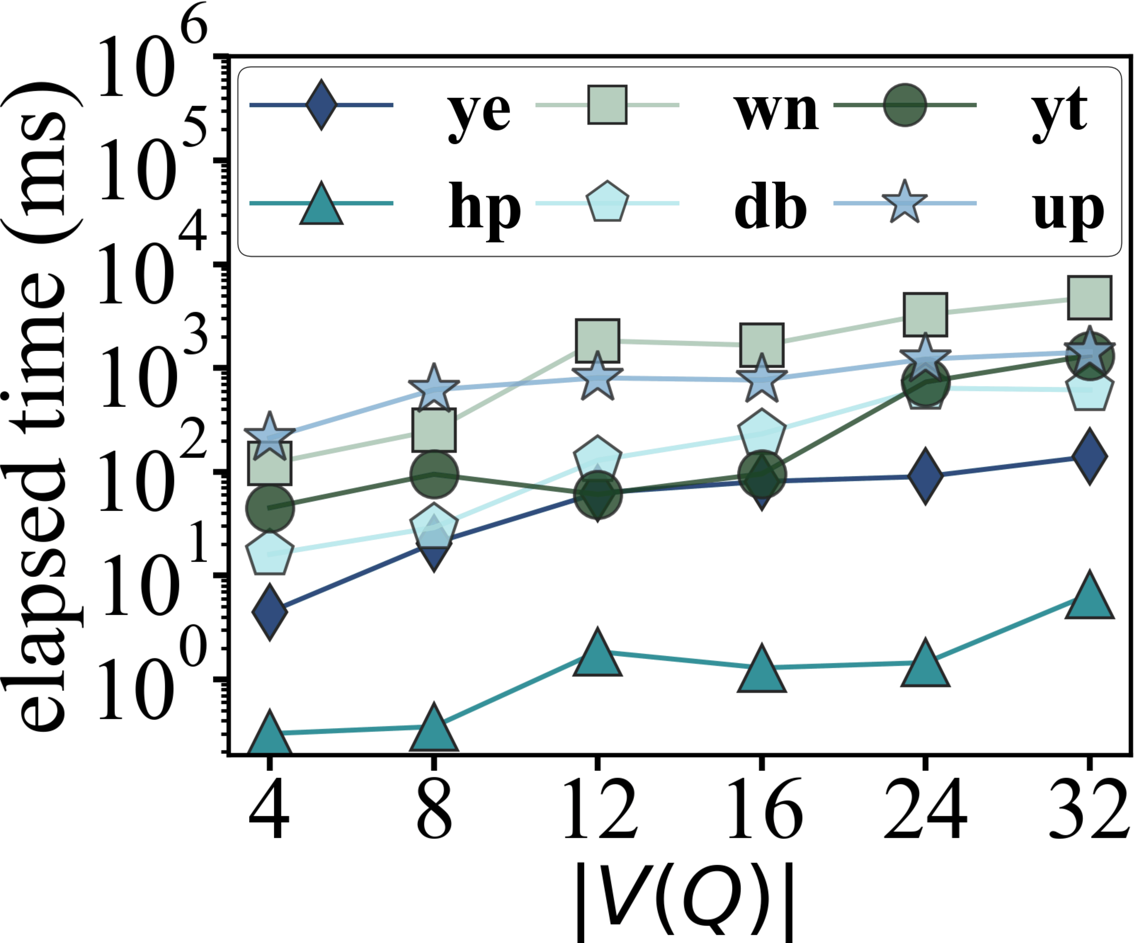}}
		\hspace{0.5cm}	
		\subfloat[synthetic datasets]{\includegraphics[width=0.38\linewidth,height=1in]{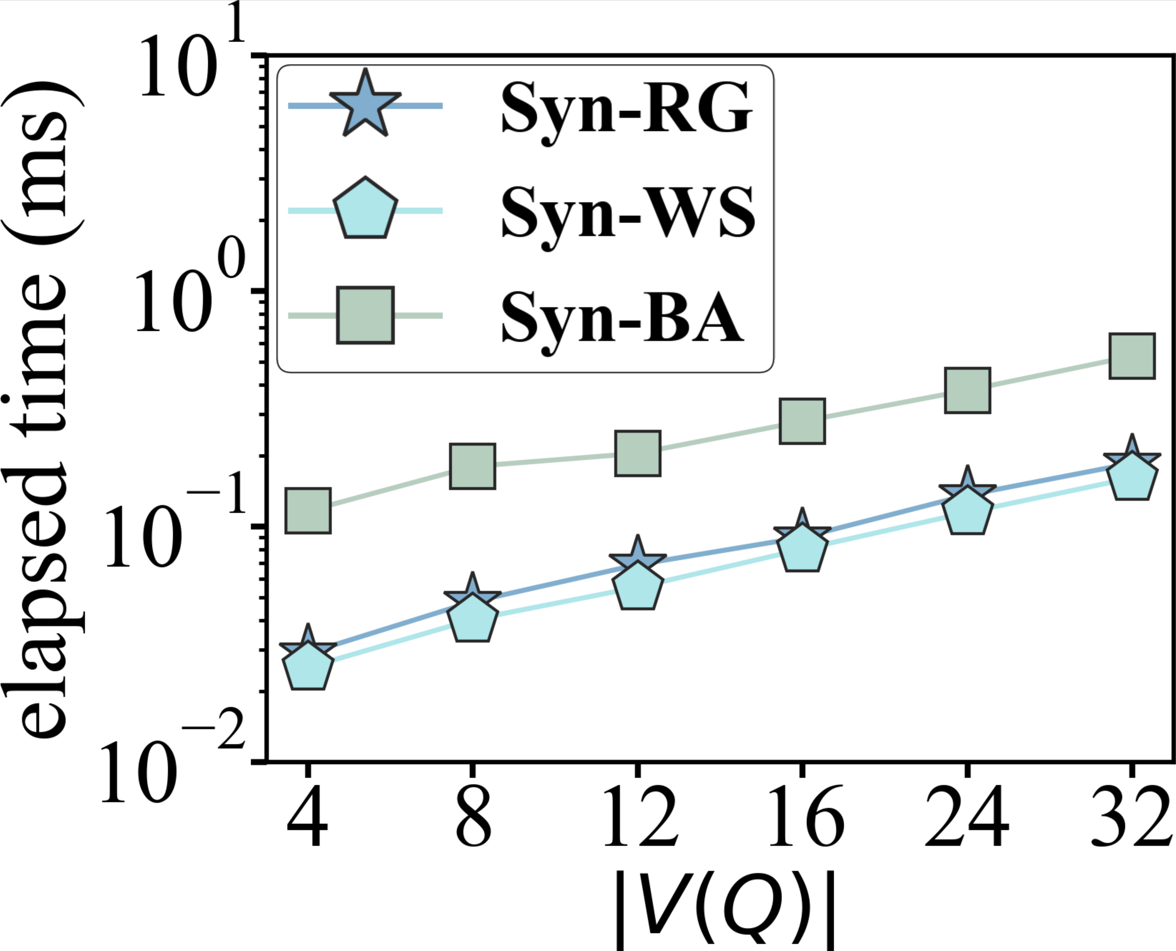}}
	\end{minipage}
	\caption{The GNN-AE efficiency w.r.t query graph size $|V(Q)|$.}
	\label{fig:efficiency-query-size}
	\vspace{-1.9em}
\end{figure}

\begin{figure}[!t]
	\setlength{\abovecaptionskip}{0.2cm}
	\centering
	\begin{minipage}{\linewidth}
		\centering
		\subfloat[$|V(G)|$]{\includegraphics[width=0.31\linewidth,height=0.905in]{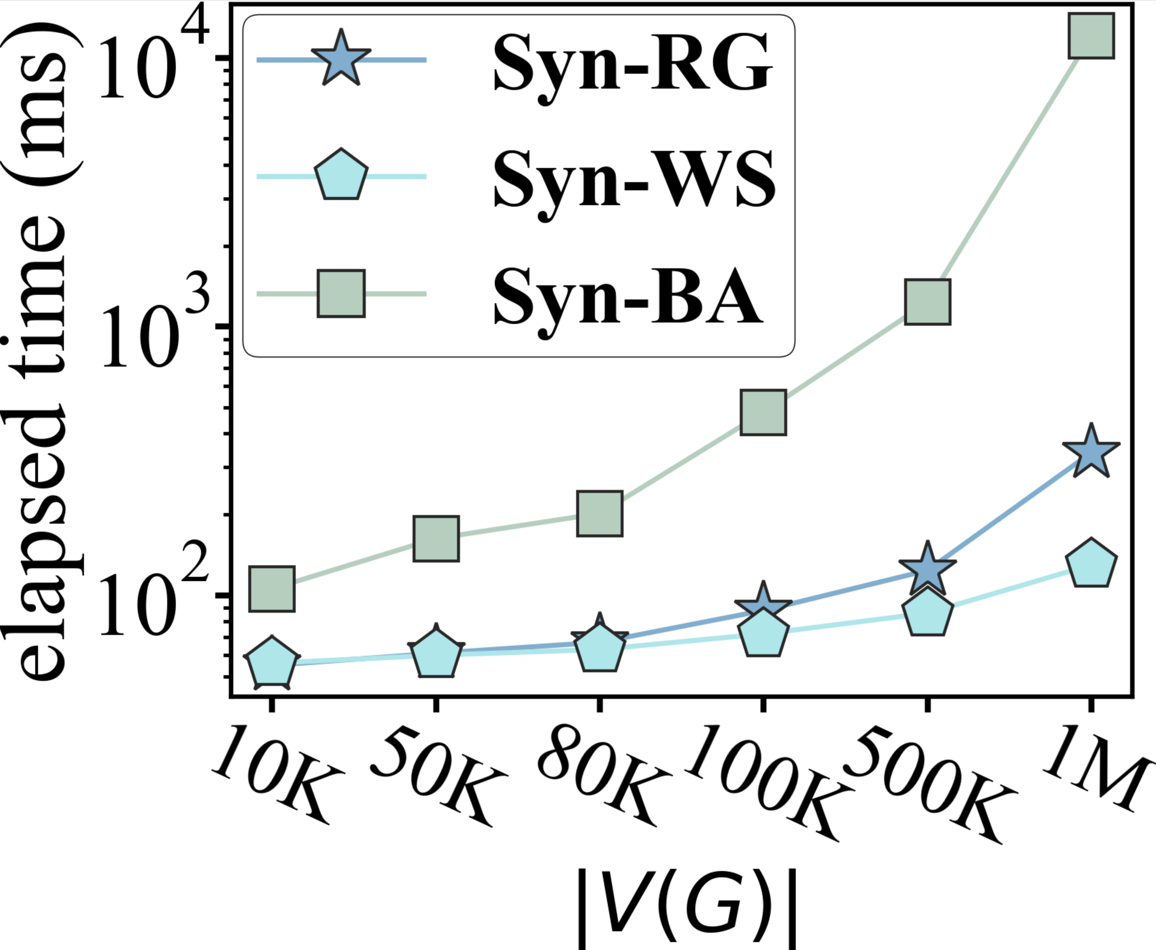}}
		\hspace{0.15cm}	
		\subfloat[$avg\_deg(G)$]{\includegraphics[width=0.31\linewidth,height=0.95in]{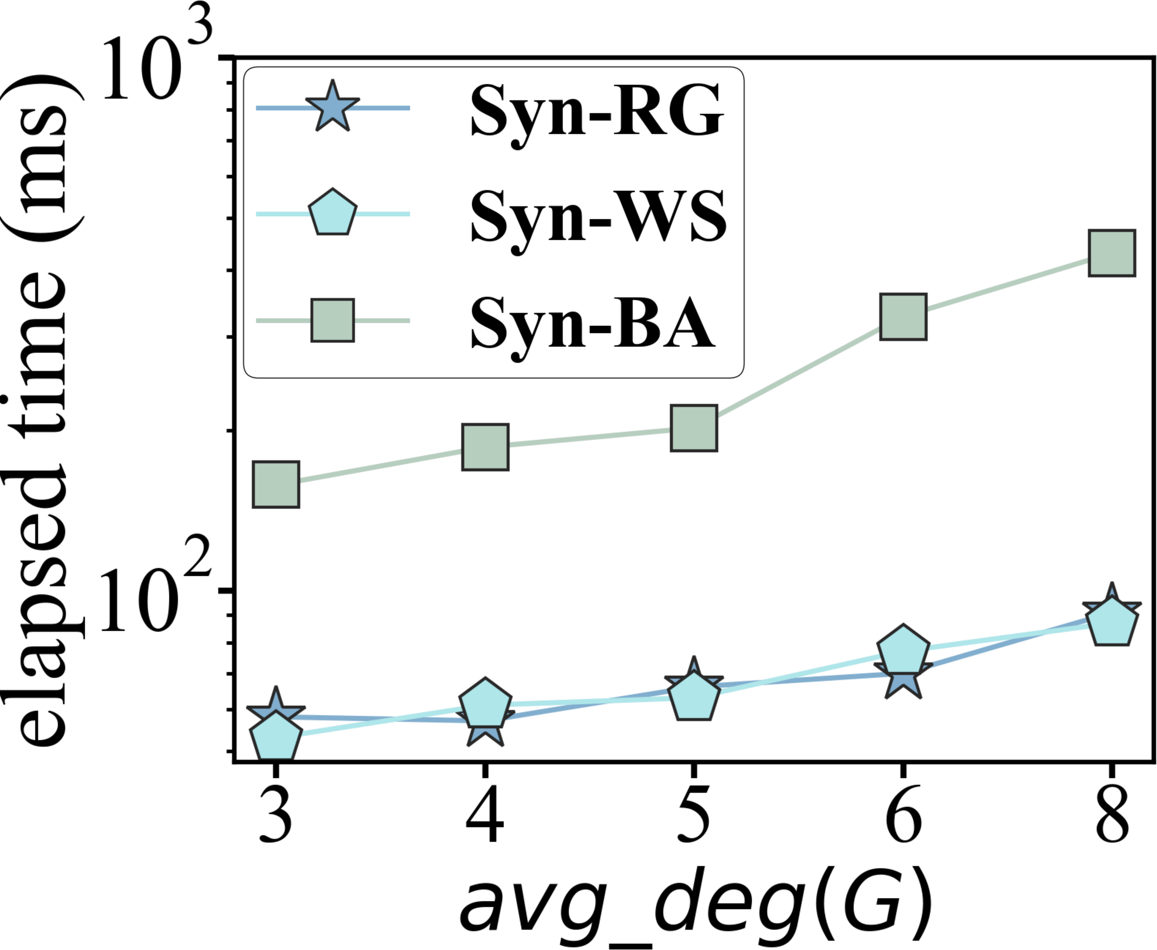}}
		\hspace{0.15cm}	
		\subfloat[$|\Sigma|$]{\includegraphics[width=0.31\linewidth,height=0.9in]{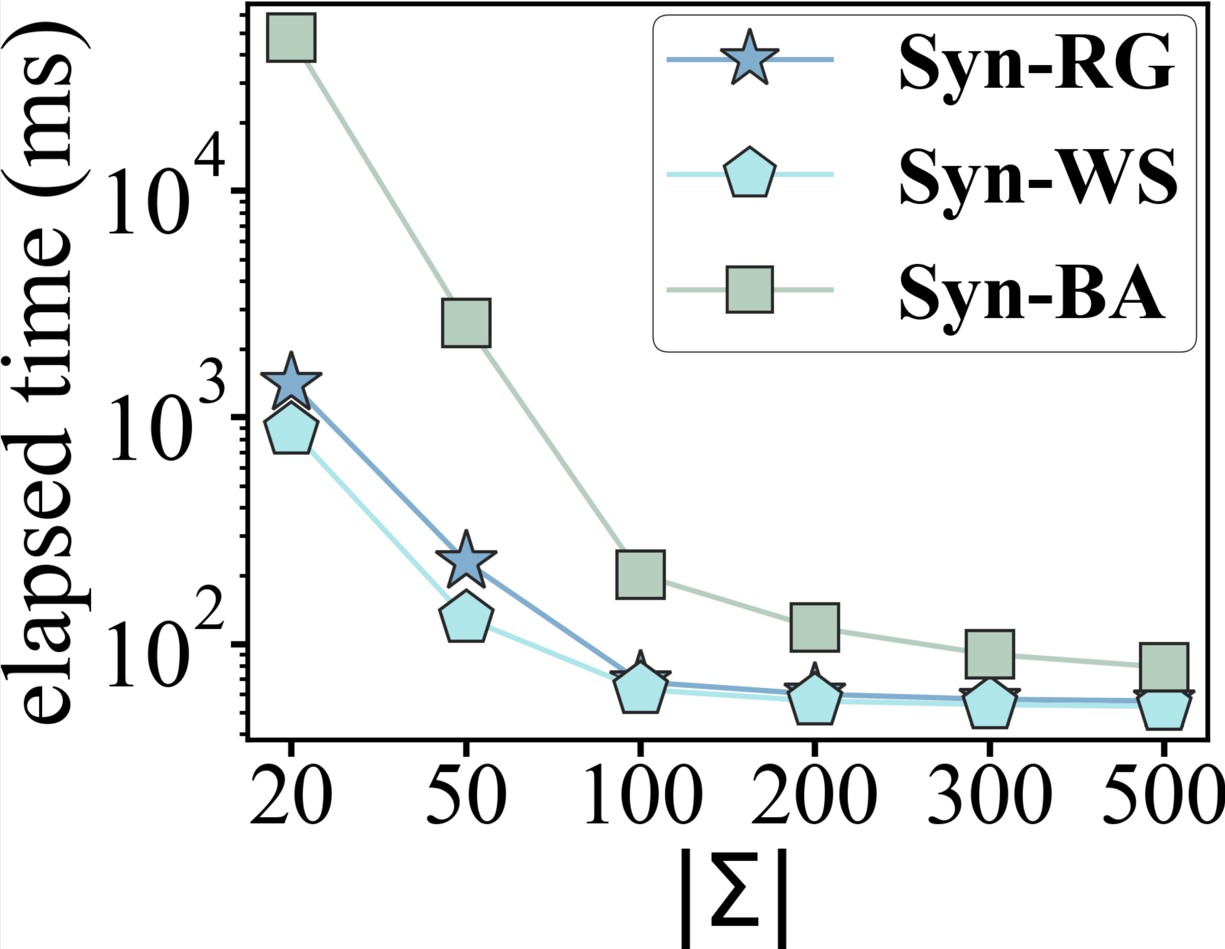}}
	\end{minipage}
	\caption{The GNN-AE scalability w.r.t varying data graph.}
	\label{fig:efficiency-data-graph}
	\vspace{0em}
\end{figure}

\noindent{\bf The GNN-AE Efficiency w.r.t Query Graph Size.} 
Fig.~\ref{fig:efficiency-query-size} illustrates the total query time of GNN-AE with varying the query graph size $|V(Q)|$ from 4 to 32. When $|V(Q)|$ increases, more DFS edges in the query are processed. According to the time complexity in the subsection~\ref{sub:matching-growth-algorithm}, the elapsed time of online subgraph matching tends to increase. Thus, larger query graph size $|V(Q)|$ incurs higher elapsed time. For varying query graph sizes, GNN-AE achieves a low average time cost.

\noindent{\bf The GNN-AE Scalability w.r.t Data Graph.}
Fig.~\ref{fig:efficiency-data-graph} evaluates the total query time of GNN-AE for all queries with different data graph characteristics, including the data graph size $|V(G)|$, $avg\_deg(G)$, and $|\Sigma|$ of $G$. Intuitively, larger $|V(G)|$ and $avg\_deg(G)$, more candidate matches for the query, which incurs an increased time cost. When $|\Sigma|$ increases, the feature embeddings have a stronger filtering power, which produces fewer candidate matches for the query. Thus, the elapsed time of GNN-AE shows a downward trend. Nevertheless, for varying data graphs, GNN-AE still remains a low time cost (i.e., $<47$ $sec$), which confirms the scalability of GNN-AE.

\noindent{\bf The Feature Embedding Filtering Power.}
Fig.~\ref{fig:filtering-real-syn} presents the filtering power of our proposed feature embeddings (i.e., anchored subgraph \& path embeddings). From subfigures, we can see that for all real-world (except \textit{wn}) and synthetic datasets, the filtering power of our feature embeddings can reach $99.37\% \sim 99.99\%$, which confirms the effectiveness of the feature embeddings. For \textit{wn}, most vertices have the same label ($|\Sigma|$ is only 5), which makes it challenging. Thus, feature embeddings have a relatively low filtering power on \textit{wn}.

\begin{figure}[!t]
	\setlength{\abovecaptionskip}{0.2cm}
	\centering
	\begin{minipage}{\linewidth}
		\centering
		\subfloat[real-world datasets]{\includegraphics[width=0.38\linewidth,height=0.9in]{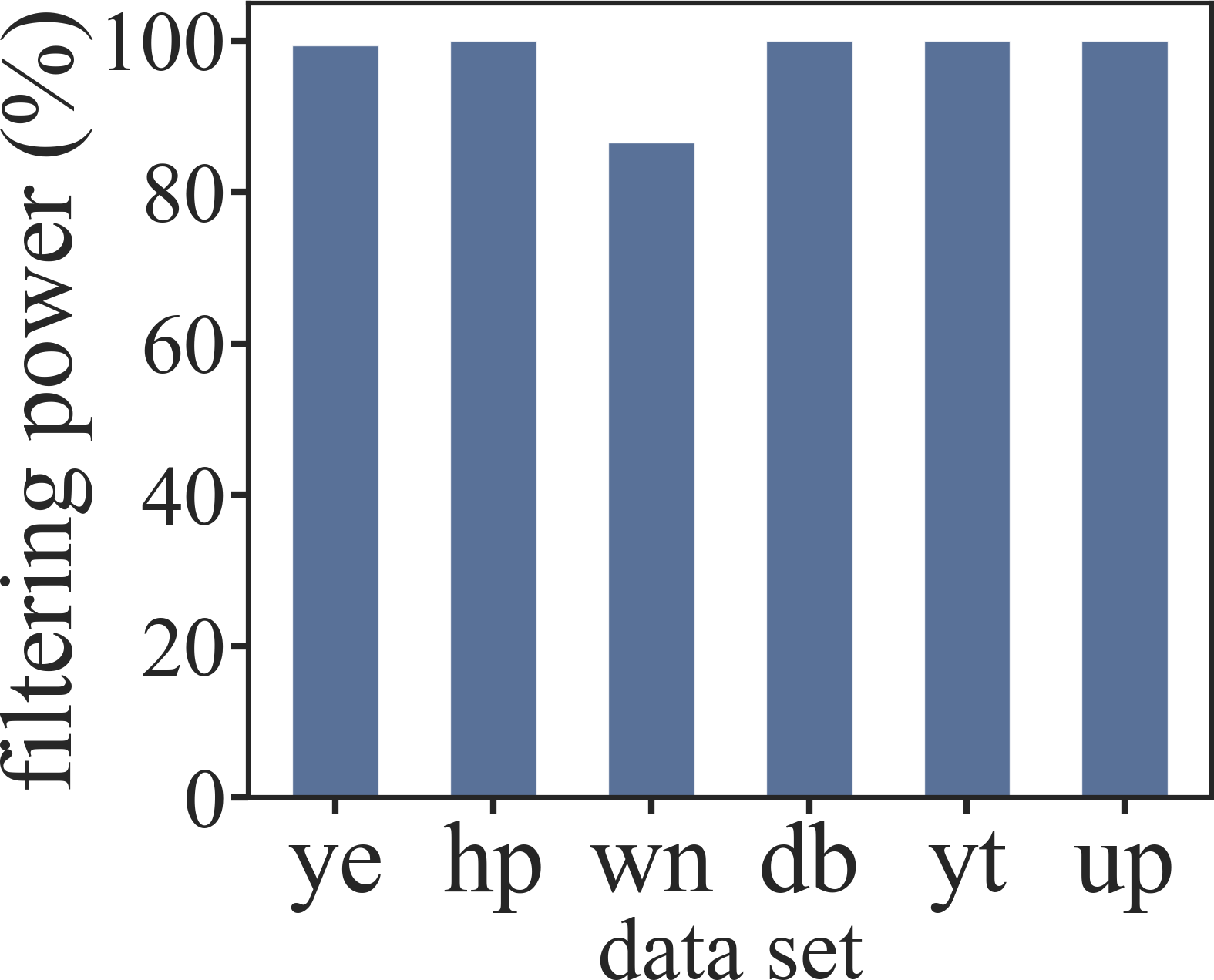}}
		\hspace{0.5cm}	
		\subfloat[synthetic datasets]{\includegraphics[width=0.38\linewidth,height=0.9in]{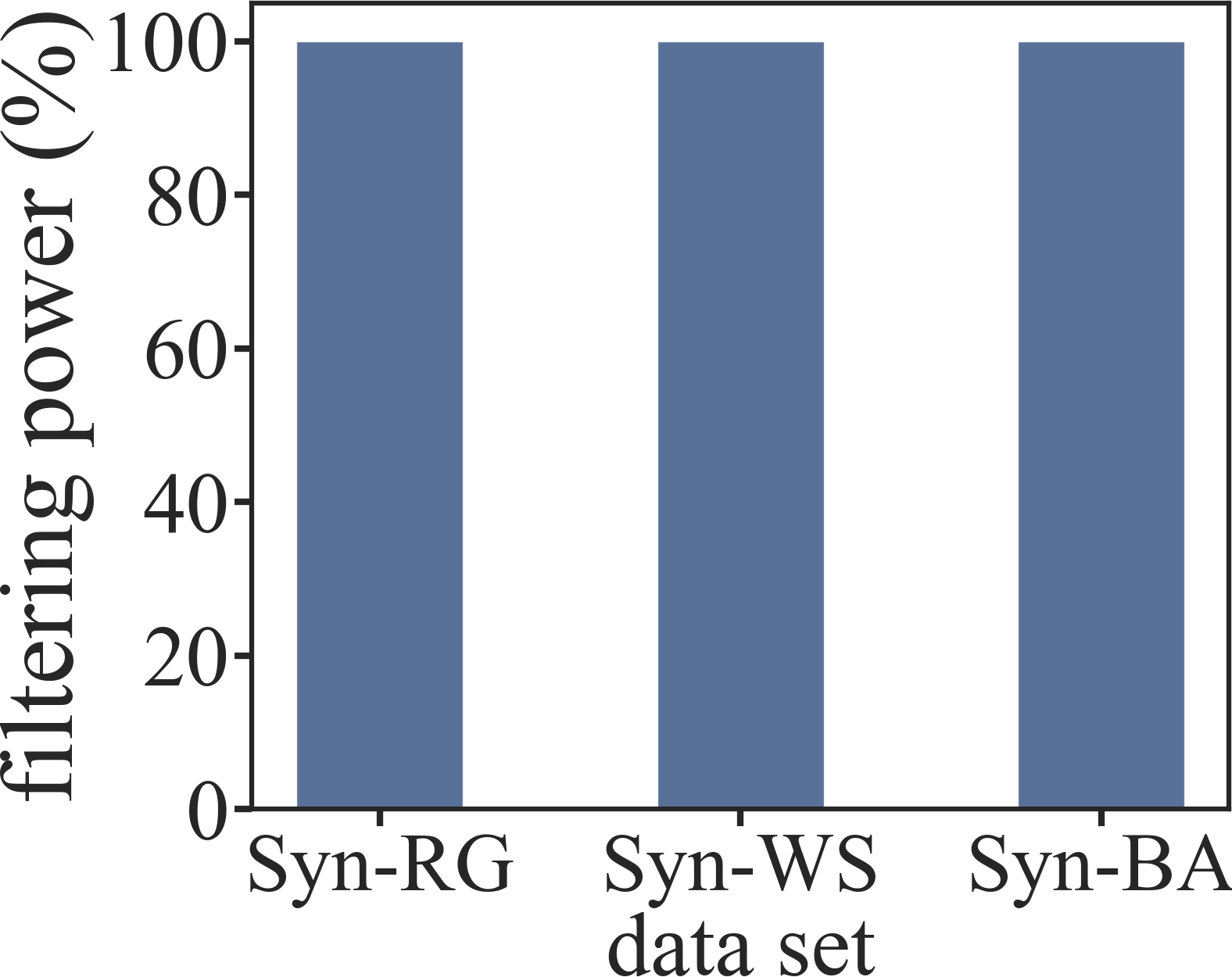}}
	\end{minipage}
	\caption{The feature embedding filtering power.}
	\label{fig:filtering-real-syn}
	\vspace{-0.5em}
\end{figure}

\subsection{Offline Pre-Computation Evaluation}
\label{sub:offline-pre-computation}

\begin{figure}[!t]
	\setlength{\abovecaptionskip}{0.02cm}
	\centering
	\includegraphics[width=0.66\linewidth,height=1.0in]{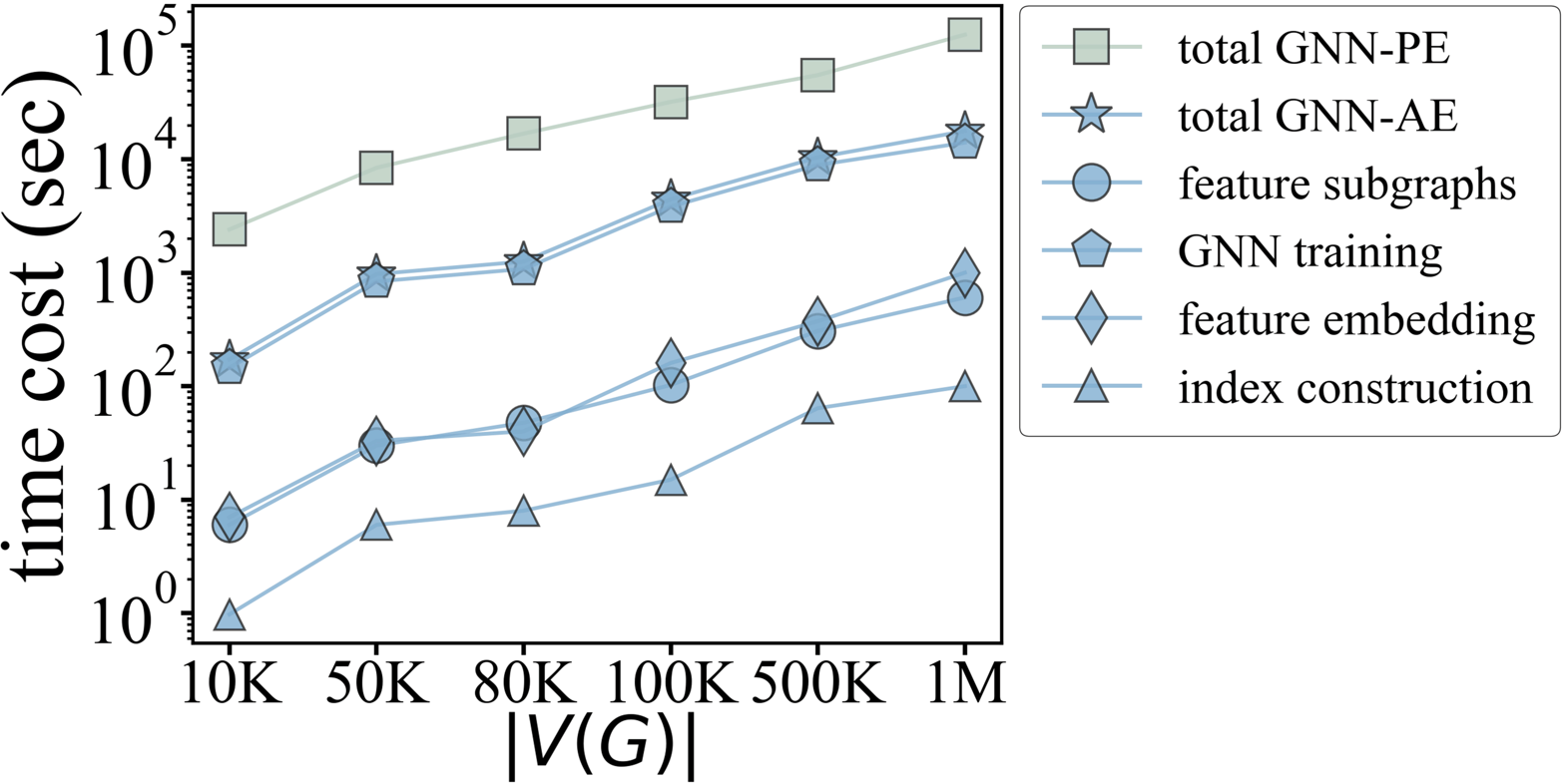}
	\caption{The GNN-AE offline pre-computation cost.}
	\label{fig:pre-computation}
	\vspace{-1.7em}
\end{figure}

\begin{figure}[!t]
	\setlength{\abovecaptionskip}{0.2cm}
	\centering
	\begin{minipage}{\linewidth}
		\centering
		\subfloat[index construction time]{\includegraphics[width=0.41\linewidth,height=1.16in]{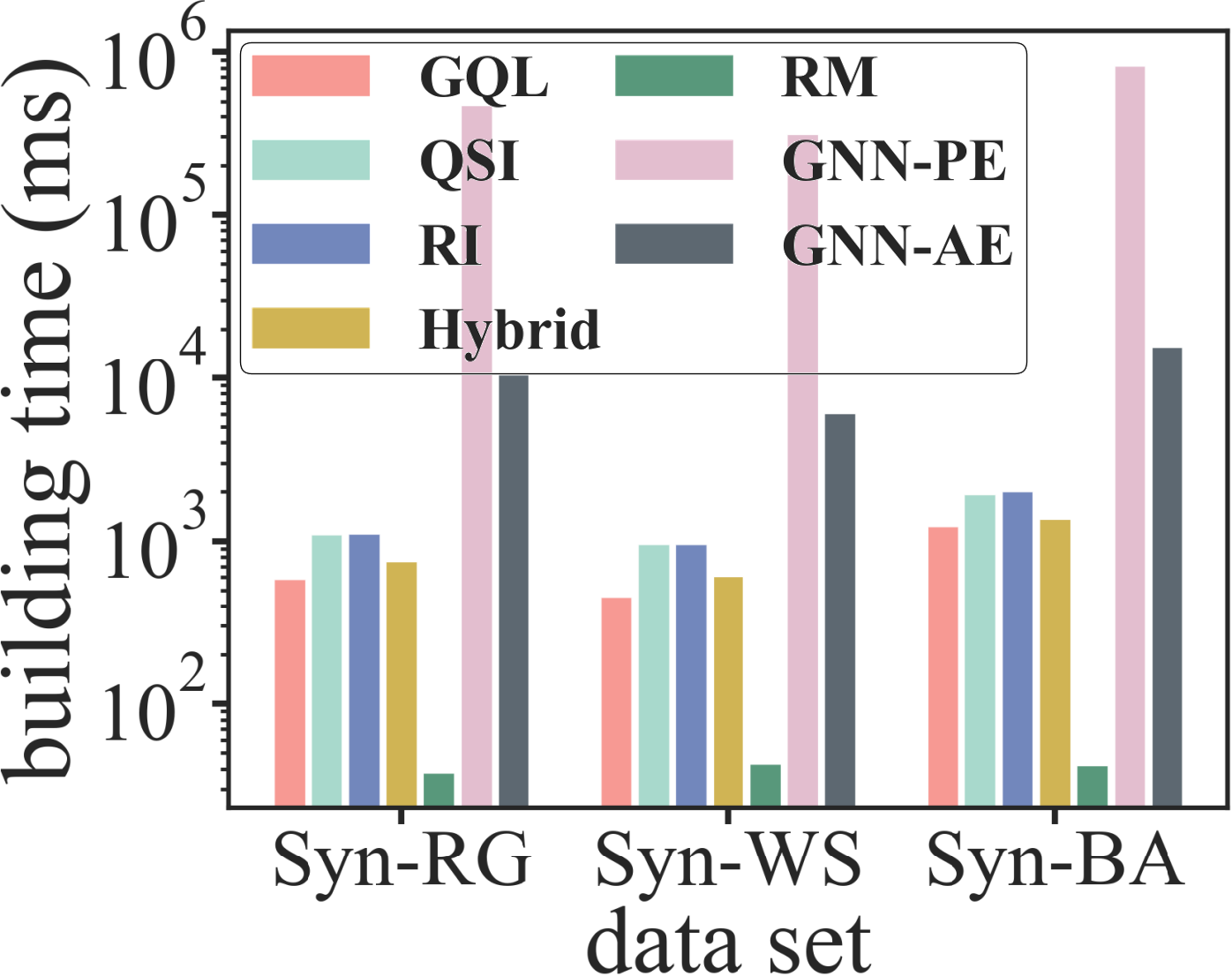}}
		\hspace{0.3cm}
		\subfloat[index storage cost]{\includegraphics[width=0.41\linewidth,height=1.18in]{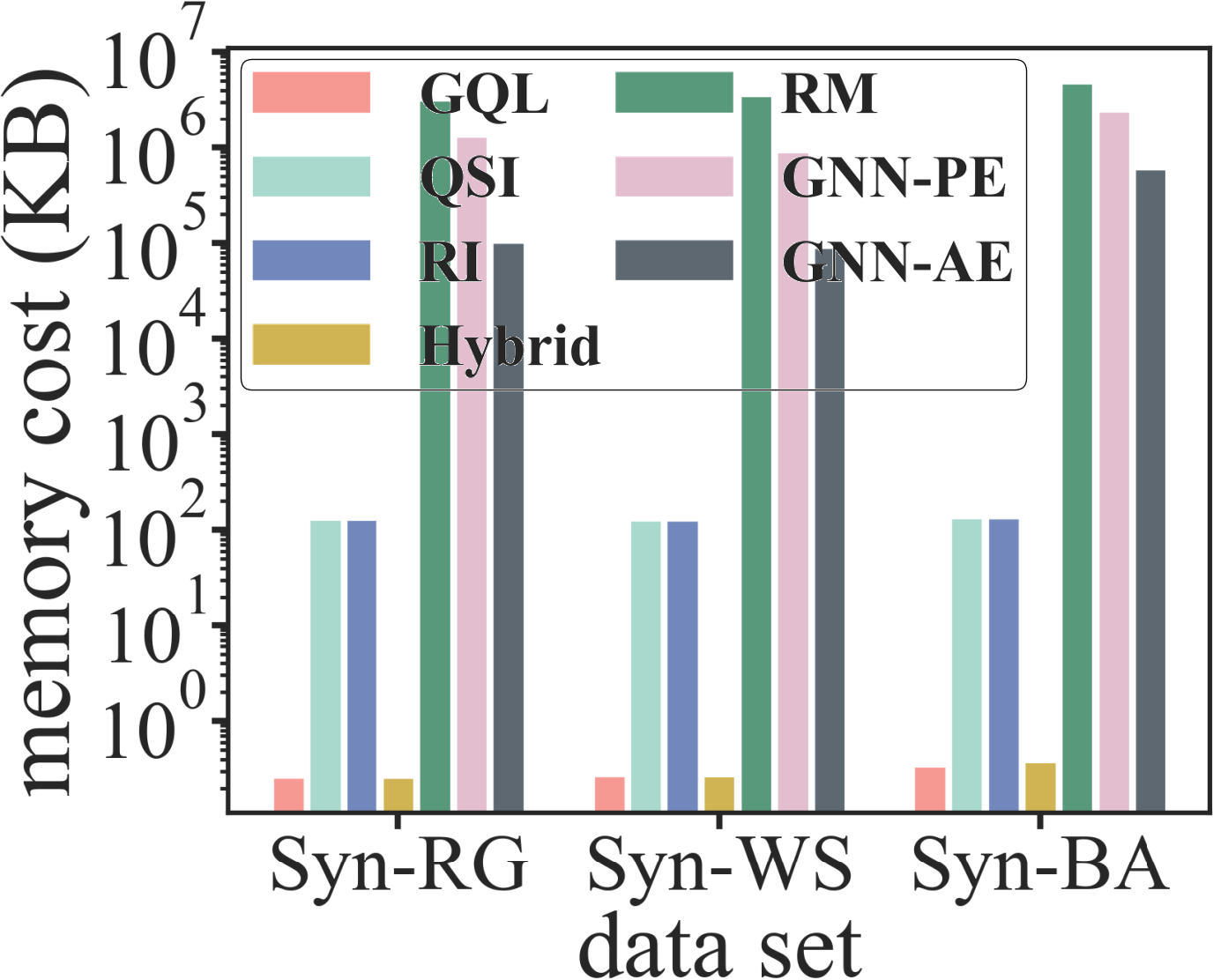}}
	\end{minipage}
	\caption{The GNN-AE index construction time and space costs.}
	\label{fig:index-cost}
	\vspace{0em}
\end{figure}

\noindent{\bf The GNN-AE Offline Pre-Computation Cost.}
The GNN-AE offline pre-computation consists of obtaining feature subgraphs (anchored subgraphs \& anchored paths), GNN training, generating feature embeddings, and index construction. Fig.~\ref{fig:pre-computation} evaluates the offline pre-computation cost of GNN-AE with different data graph size $|V(G)|$ on the synthetic graph \textit{Syn-WS}. In the subfigure, we find that GNN training takes up most of the offline pre-computation cost. Since GNN-PE requires the GNN model to be trained to overfit the training data set, GNN-PE has a high offline pre-computation cost. However, our GNN-AE can achieve a lower per-computation cost than GNN-PE by 1 order of magnitude for varying $|V(G)|$.

\noindent{\bf The GNN-AE Index Construction Time \& Space Costs.}
Fig.~\ref{fig:index-cost} compares the index construction time and memory cost of GNN-AE with 6 baseline methods that require auxiliary indexes. Although GNN-AE requires higher index construction time than exploration-based backtracking search methods (e.g., GQL, QSI, RI, and Hybrid) and RM, our GNN-AE has a lower index construction cost than the advanced GNN-PE approach. For the index storage cost, we calculate the average index memory cost for query graphs. The memory cost of RM on some query graphs exceeds $500MB$. GNN-AE only stores the anchored (subgraph \& path) embeddings and their corresponding data graph edges in the indexes, rather than storing a larger number of paths and their embeddings like GNN-PE. Thus, our GNN-AE has a lower index storage cost. Note that, different from most baseline methods that construct an auxiliary index for each query graph during online subgraph matching, the index construction of our GNN-AE is \textit{offline} and \textit{one-time only}. Thus, the indexes in our GNN-AE are suitable for processing numerous subgraph matching requests. 

\subsection{Other Parameter Evaluation}
\label{sub:parameter-evaluation}

\begin{figure}[!t]
	\centering
	\begin{minipage}{\linewidth}
		\centering
		\subfloat[degree threshold $d^*$]{\includegraphics[width=0.43\linewidth,height=1.16in]{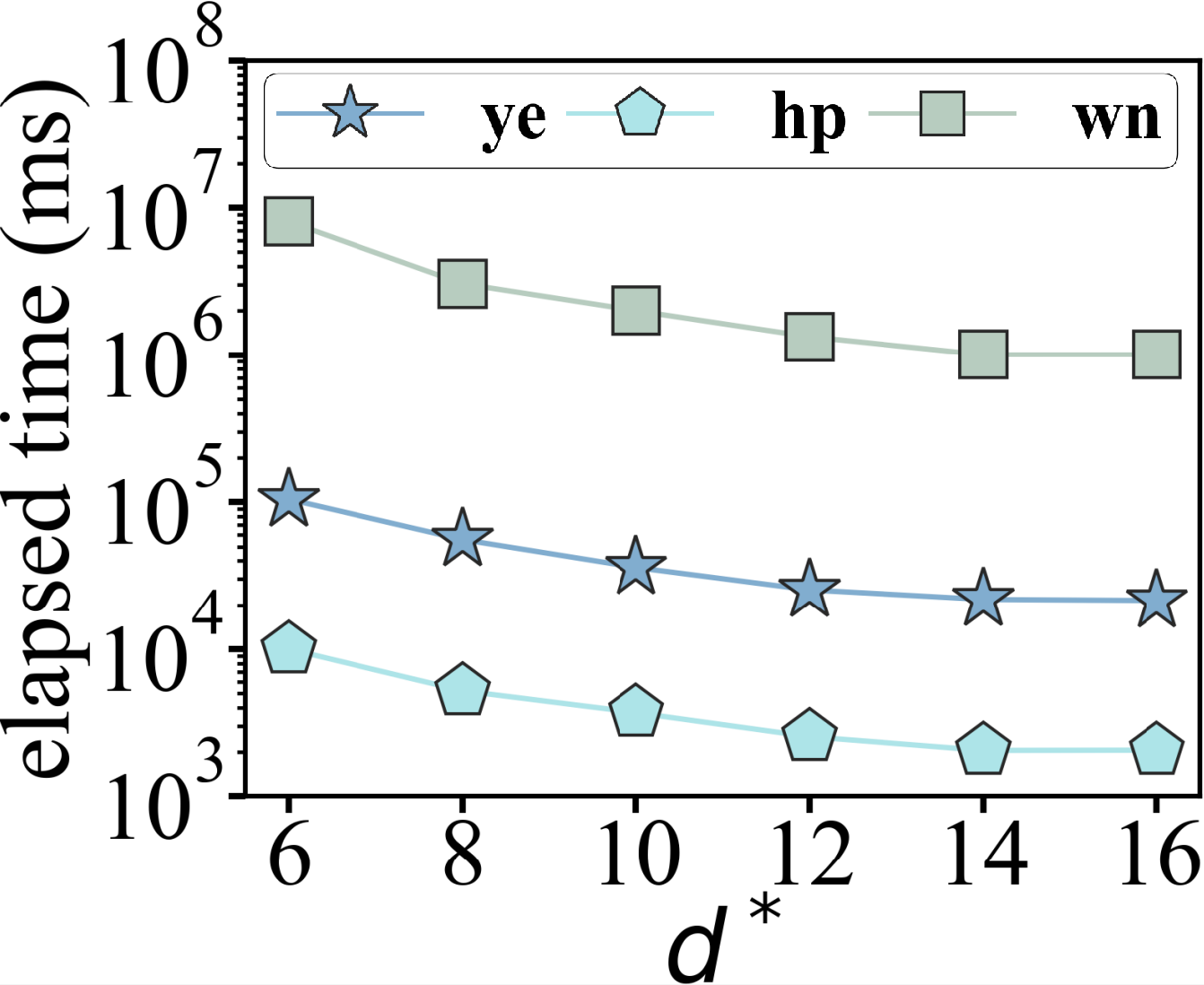}}
		\hspace{0.3cm}	
		\subfloat[DFS query strategies]{\includegraphics[width=0.43\linewidth,height=1.1in]{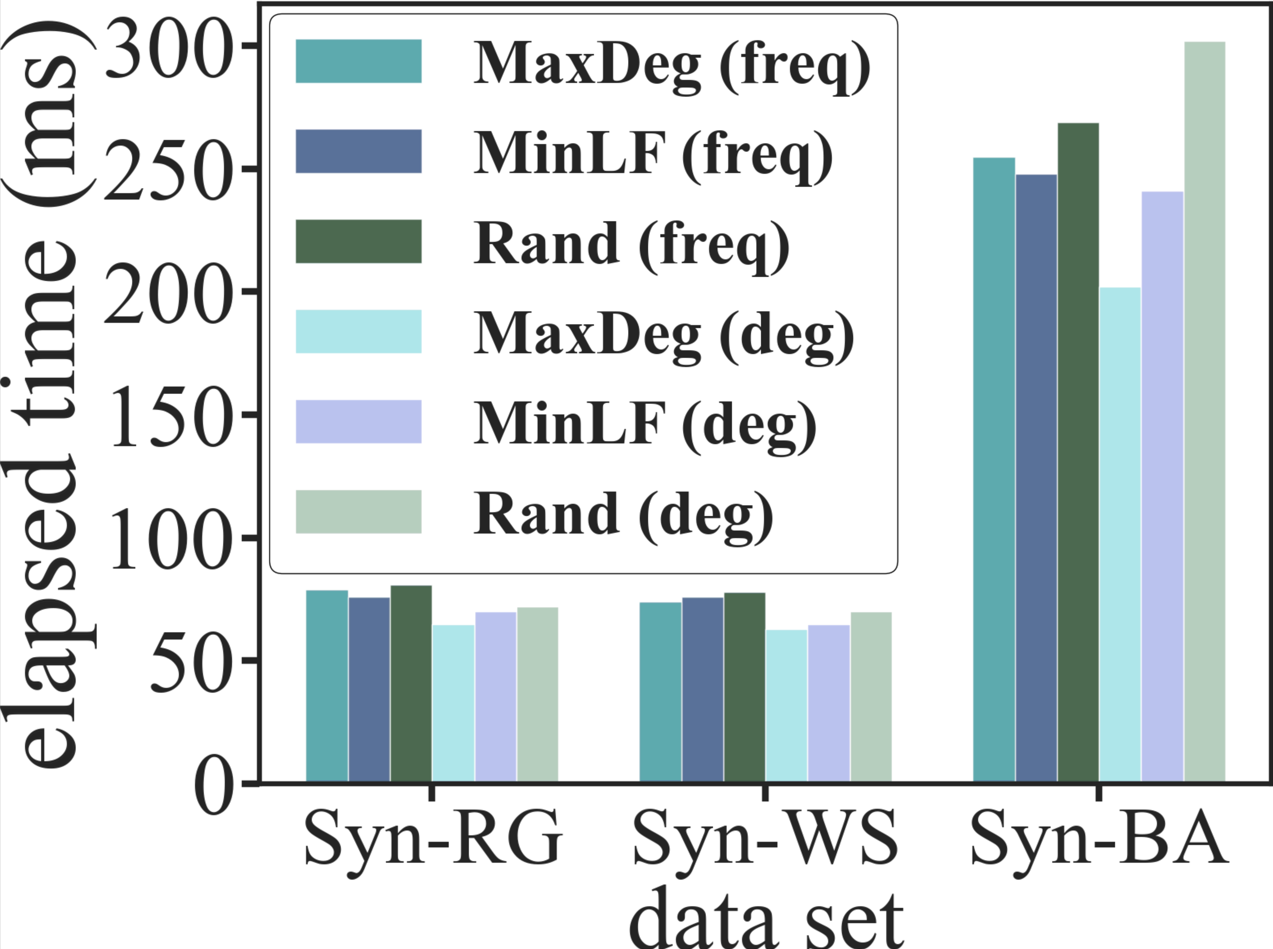}}
	\end{minipage}
	\caption{The GNN-AE efficiency w.r.t. varying $d^*$, and DFS query strategies.}
	\label{fig:paremeter-tuning}
	\vspace{-0.4em}
\end{figure}

\noindent{\bf The GNN-AE Efficiency w.r.t Threshold $d^*$.}
Fig.~\ref{fig:paremeter-tuning}(b) shows the total query time of GNN-AE for all queries on 3 real-world graphs, by varying the vertex degree threshold $d^*$. When $d^*$ increases, more edges on the data graph obtain feature embeddings via the anchored subgraph rather than the path, which incurs higher filtering power for embeddings. The elapsed time of GNN-AE shows a downward trend (i.e., reduced up to 73\% on \textit{ye}, 75\% on \textit{hp}, and 82\% on \textit{wn}). 

\noindent{\bf The GNN-AE Efficiency w.r.t DFS Query Plan.}
Fig.~\ref{fig:paremeter-tuning}(c) shows the total query time of our GNN-AE for all queries with different DFS query plan strategies on 3 synthetic graphs. Generally, degree-based query cost methods perform better than label-frequency-based ones. This implies that the degree-based query cost strategy may obtain fewer candidate matches for each DFS edge in the query graph. In addition, the combination MaxDeg (deg) obtains the optimal performance. 

\section{Related Work}
\label{sec:related-work}

\noindent{\bf Exact Subgraph Matching.}
Existing exact subgraph matching methods usually follow the \textit{exploration-based} or \textit{join-based} paradigms~\cite{sun2020memory, sun2021rapidmatch}. The exploration-based methods adopt the backtracking search, which recursively extend intermediate results by mapping query graph vertices to data graph vertices along a matching order~\cite{sun2012efficient, carletti2017challenging, bhattarai2019ceci, han2019efficient, han2013turboiso, zhao2010graph, zhang2009gaddi, choi2023bice, jiang2024ive}. The join-based methods model the subgraph matching as a relational query problem, and perform a multi-way join to obtain all matching results~\cite{aberger2017, mhedhbi2019, nguyen2015join, ammar2018, lai2019distributed, sun2021rapidmatch}.

There are several new works~\cite{wang2022reinforcement, ye2024efficient} that attempt to utilize deep-learning-based techniques to deal with the exact subgraph matching, but they both have obvious drawbacks. RL-QVO in~\cite{wang2022reinforcement} employs the \textit{Reinforcement Learning} (RL) and GNNs to generate a high-quality matching order for subgraph matching algorithms, but it relies on historical query graphs to train the model, which limits its scalability and robustness. GNN-PE in~\cite{ye2024efficient} adopts a GNN model to learn the dominance relationships between the vertex (and its 1-hop neighbors) and their substructures on the data graph, and then uses the path with dominance relationships to obtain matches. GNN-PE requires the GNN model to be trained to overfit the training data set, which results in a high computation cost.  

\noindent{\bf Approximate Subgraph Matching.}
An alternative methodology is to obtain top-$k$ approximate results for subgraph matching, which returns subgraphs in the data graph $G$ that are similar to the given query $Q$~\cite{gori2005exact, dutta2017neighbor, zhu2011structure, kpodjedo2014using}. Approximate subgraph matching usually adopt various graph similarity measures such as edge edit distance~\cite{zhang2010sapper} to qualify occurrences of the query $Q$ on the data graph $G$. If the edit distance between the query $Q$ and a subgraph $g$ of $G$ is no more than some threshold $\theta$, then $g$ is an approximate result of $Q$.

Recently, several advanced works~\cite{liu2023d2match, roy2022interpretable, lou2020neural, doan2021interpretable} attempt to utilize deep-learning-based techniques to handle the approximate subgraph matching. However, these approaches usually only approximately assert subgraph isomorphism relations between the query and the data graph, and can not retrieve locations of all matching. Beyond that, they have no theoretical guarantees or error bounds about the approximate accuracy.  

\noindent{\bf Other Learning-based Graph Computation.}
In recent years, learning-based methods have emerged in a few related graph computation work, such as subgraph counting~\cite{zhao2021learned, liu2020neural, wang2022neural, hou2024learnsc}, shortest path distance~\cite{huang2021learning}, graph edit distance~\cite{yang2021noah} and graph keyword search~\cite{hao2021ks}. However, these graph computation problems (e.g., count or distance) are different from the subgraph matching problem, and the learning techniques in these works cannot be applied to solve the subgraph matching problem. 

\section{Conclusion}
\label{sec:conclusion}

The proposed GNN-AE framework is a learning-based framework that can exactly answer subgraph matching queries. It indexes small feature subgraphs in the data graph offline, thereby avoiding creating auxiliary summary structures online and improving online query efficiency. Our method uses GNNs to perform graph isomorphism on indexed feature subgraphs, thus obtaining high-quality candidate matches more efficiently than traditional graph isomorphism test methods. Using both anchored subgraphs and anchored paths as features can improve query efficiency while reducing index storage cost. The proposed parallel matching growth algorithm and the cost-model-based DFS query planning method efficiently obtain the locations of all matches at a low query cost. The extensive experimental evaluation verifies the efficiency of our GNN-AE compared to baseline methods.

\section*{Acknowledgments}

This work was partially supported by the National Natural Science Foundation of China (No.~62472123, No.~62072138).

\balance

\bibliographystyle{IEEEtran}
\bibliography{ref.bib}

\vfill

\end{document}